\documentclass[12pt]{article}
\usepackage{amsmath,amsthm}
\usepackage{graphicx}
\usepackage{enumerate}
\usepackage{natbib}
\usepackage{url} 
\RequirePackage[colorlinks,citecolor=blue,urlcolor=blue]{hyperref}
\usepackage{graphicx}
\usepackage{enumerate}
\usepackage{latexsym}
\usepackage{bigstrut}
\usepackage{amssymb}
\usepackage{multirow}
\usepackage{float}
\usepackage{tikz}
\usetikzlibrary{decorations.pathreplacing,calligraphy}
\usepackage{bm}
\usepackage{paralist}
\usepackage{mathabx}
\usepackage{natbib}
\usepackage[ruled,vlined]{algorithm2e}
\usepackage{commath}
\usepackage{subcaption}
\usepackage{soul,xcolor}
\usepackage{subcaption}
\usepackage[inline]{enumitem}
\usepackage[figuresright]{rotating}
\usepackage[labelfont=bf]{caption}

\usepackage{thmtools}
\renewcommand\thmcontinues[1]{Continued}

\usepackage{pifont}
\usepackage{threeparttable}
{\end{quotation}}

\usepackage[normalem]{ulem}

\DeclareMathOperator*{\argmin}{argmin}
\renewcommand{\baselinestretch}{1.23}
\numberwithin{equation}{section}

\def \mcF{\mathcal{F}}

\def \mcC{\mathcal{C}}

\def \mcD{\mathcal{D}}

\def \mcB{\mathcal{B}}

\def\sign{\textrm{sign}}

\setstcolor{red}

\def \mcA{\mathcal{A}}

\def \bfZ{\mathbf{Z}}

\def \ev{\mathbb{E}}
\def \sign{\textrm{sign}}
\def \pr{\mathbb{P}}

\def \cid{\xrightarrow[\text{}]{\text{$\mathbb{L}$}}}
\def \cip{\xrightarrow[\text{}]{\text{$\mathbb{P}$}}}
\def \cae{\xrightarrow[\text{}]{\text{$a.s.$}}}
\newcommand{\vo}{\vec{o}\@ifnextchar{^}{\,}{}}

\newcommand{\floor}[1]{\left\lfloor #1 \right\rfloor}

\theoremstyle{plain}

\newtheorem*{remark}{Remark}
\newtheorem{theorem}{\indent Theorem}
\newtheorem*{theorem*}{\indent Theorem}

\newtheorem{Assumption}{Assumption}

\newtheorem{AssumptionB}{Assumption}

\newtheorem{Definition}{Definition}

\newtheorem{lemma}{\indent Lemma}
\newtheorem{proposition}{\indent Proposition}
\newtheorem{Remark}{Remark}
\theoremstyle{definition}
\newtheorem{Example}{Example}
\DeclareSymbolFont{largesymbolsA}{U}{txexa}{m}{n}
\DeclareMathSymbol{\varprod}{\mathop}{largesymbolsA}{16}
\usepackage{xparse}

\graphicspath{{../}{./}}
\newcommand{\blind}{1}

\usepackage{multibib}
\newcites{App}{References}

\graphicspath{{./Figure/}}

\begin{document}

\def\spacingset#1{\renewcommand{\baselinestretch}%
{#1}\small\normalsize} \spacingset{1}

\if1\blind
{
  \title{\bf Statistical Inference with Stochastic Gradient Methods under $\phi$-mixing Data}
  \author{Ruiqi Liu\hspace{.2cm}\\
    Department of Mathematics and Statistics, Texas Tech University\\
    and\\
    Xi Chen\hspace{.2cm}\\
    Leonard N. Stern School of Business, New York University\\
    and\\
    Zuofeng Shang\hspace{.2cm}\\
    Department of Mathematical Sciences, New Jersey Institute of Technology}
  \maketitle
} \fi

\if0\blind
{
  \bigskip
  \bigskip
  \bigskip
  \begin{center}
    {\LARGE\bf Statistical Inference with Stochastic Gradient Methods under $\phi$-mixing Data}
\end{center}
  \medskip
} \fi

\bigskip
\begin{abstract}
Stochastic gradient descent (SGD) is a scalable and memory-efficient optimization algorithm for large datasets and stream data, which has drawn a great deal of attention and popularity. The applications of SGD-based estimators to statistical inference such as interval estimation have also achieved great success. However, most of the related works are based on i.i.d. observations or Markov chains. When the observations come from a mixing time series, how to conduct valid statistical inference remains unexplored. As a matter of fact, the general correlation among observations imposes a challenge on interval estimation.  Most existing methods may ignore this correlation and lead to invalid confidence intervals. In this paper, we propose a mini-batch SGD estimator for statistical inference when the data is $\phi$-mixing. The confidence intervals are constructed using an associated mini-batch bootstrap SGD procedure. Using ``independent block'' trick from \cite{yu1994rates}, we show that the proposed estimator is asymptotically normal, and its limiting distribution can be effectively approximated by the bootstrap procedure. The proposed method is memory-efficient and easy to implement in practice. Simulation studies on synthetic data and an application to a real-world dataset confirm our theory.
\end{abstract}
\noindent%
\vfill

\newpage
\spacingset{1.9} 

\section{Introduction}
We consider the following minimization problem:
\begin{eqnarray}
\theta^*=\argmin_{\theta\in \Theta}\left\{L(\theta):=\ev\left\{l(Z,\theta)\right\}\right\},\label{eq:model}
\end{eqnarray}
where $\theta^*\in \Theta \subset \mathbb{R}^d$ represents the parameter of interest, $Z$ is a copy of  some random vector with an unknown distribution, and $l(z, \theta)$ is the loss function. Furthermore, assume that we can access a sequence of potentially dependent observations $\{Z_t\}_{t=1}^\infty$, where each $Z_t$ follows the same distribution as $Z$. Because the objective function in \eqref{eq:model} is unavailable, classical methods approximate $\theta^*$ by minimizing the empirical version of $L(\theta)$ based on $n$ samples:
\begin{eqnarray}
\widehat{\theta}=\argmin_{\theta \in \Theta}\frac{1}{n}\sum_{t=1}^n l(Z_t, \theta).\label{eq:mle}
\end{eqnarray}
The minimization problem in \eqref{eq:mle} is commonly solved using gradient-based iterative algorithms. However, its (sub)gradient\footnote{If $\theta \mapsto l(z, \theta)$ is nonsmooth, its subgradient can be used instead. For simplicity, we also denote the subgradient as $\nabla l(Z_t, \theta)$ when the context is clear.} $\sum_{t=1}^n \nabla l(Z_t, \theta)/n$ involves a summation over $n$ items, which becomes computationally inefficient for large sample sizes ($n$). Additionally, many real-world scenarios involve streaming data, where observations are collected sequentially. Due to storage limitations, systems delete outdated data and retain only the most recent records. In such cases, calculating the (sub)gradient becomes infeasible.

\subsection{Related Works}
A computationally attractive algorithm for solving problem \eqref{eq:model} is stochastic gradient descent (SGD), which processes only one observation during each iteration. Specifically, given an initial value $\widehat{\theta}_0$, SGD updates the value recursively as follows:
\begin{eqnarray}
\widehat{\theta}_t=\Pi \left\{\widehat{\theta}_{t-1}-\gamma_t \nabla l(Z_t, \widehat{\theta}_{t-1})\right\}, \label{eq:vanilla:sgd}
\end{eqnarray}
where  $\gamma_t > 0$ is a predetermined learning rate (or step size), and $\Pi$ denotes the projection operator onto $\Theta$.

SGD was originally proposed in the seminal work \cite{robbins1951stochastic}. It has achieved considerable success across various areas, including image processing (\citealp{cole2003multiresolution, klein2009adaptive}), recommendation systems (\citealp{khan2019fractional, shi2020large}), and inventory control (\citealp{ding2021feature, yuan2021marrying}). In the era of big data, SGD's scalability and ease of implementation have garnered significant attention. Theoretical analysis of SGD can be classified into two main directions based on different application goals. 

The first line of research revolves around quantifying the regret of SGD, defined as $L(\widehat{\theta}_T) - L(\theta^*)$, where $T$ represents the number of iterations. Existing literature indicates that, with an appropriately chosen learning rate $\gamma_t$, SGD's regret can achieve a convergence rate of $1/T$ for strongly convex objective functions (e.g., see \citealp{bottou2018optimization,gower2019sgd}), and a rate of $1/\sqrt{T}$ for general convex cases \cite{nemirovski2009robust}.

The second research direction focuses on applying SGD to statistical inference. Under appropriate conditions, it has been demonstrated that the SGD estimator is asymptotically normal (\citealp{pelletier2000asymptotic}). Notably, the SGD estimator may not exhibit root-$T$ consistency unless the learning rate satisfies $\gamma_t \asymp t^{-1}$, which contrasts with classical parametric estimators. To accelerate convergence and achieve a root-$T$ estimator, the renowned Polyak-Ruppert averaging procedure was independently proposed by \cite{polyak1990new} and \cite{ruppert1988efficient}. This procedure constructs an estimator by averaging the trajectory $\widehat{\theta}_1, \ldots, \widehat{\theta}_T$. The proposed estimator's root-$T$ consistency was established, and its asymptotic normality was proved by  \cite{polyak1992acceleration}. A series of interesting works have been developed based on the Polyak-Ruppert averaging procedure. For instance, \cite{chen2021statistical} and \cite{chen2022online} devised SGD-based algorithms for online decision-making problems involving decision rules.  \cite{lee2022fast, lee2022fastb} extended the distributional results from \cite{polyak1992acceleration} to a functional central limit theorem. Building upon this novel theoretical result, the authors proposed an online inference procedure based on an asymptotically pivotal statistic with a nontrivial mixed normal limiting distribution. More recently, \cite{su2018uncertainty} introduced a hierarchical incremental gradient descent (HIGrad) procedure for inferring unknown parameters. In comparison to SGD, the flexible structure of HiGrad facilitates easier parallelization.

A key assumption in the aforementioned studies is the independence among $Z_t$'s. When considering correlations between observations, the SGD literature often assumes two types of dependence. The first type is Markovian dependence, and a series of works have established convergence rates and limiting distributions for SGD and its variants. This includes applications in reinforcement learning (\citealp{melo2008analysis, Roy2016tamingnoise, dalal2018finite, xu2019two, qu2020finite, shi2021statistical, ramprasad2022online, shi2022off, chen2022reinforcement}), Bayesian learning (\citealp{liang2007stochastic, NEURIPS2020_liang}), and federated learning (\citealp{li2022statistical}). The second type of dependence is $\phi$-mixing, which is considered as a more general dependence assumption compared to Markovian dependence (\citealp{davydov1973mixing}). However, results concerning SGD with $\phi$-mixing data are limited, with two related works being \cite{agarwal2012generalization} and \cite{ma2022data}. Under an additional condition $\|\widehat{\theta}_t - \widehat{\theta}_{t-1}\| \to 0$, \cite{agarwal2012generalization} established convergence guarantees for general stochastic algorithms. On the other hand, \cite{ma2022data} demonstrated that the convergence rate of SGD can be enhanced through subsampling and mini-batch techniques. Although the additional condition is well-established for i.i.d. or Markovian data, its applicability remains unexplored for $\phi$-mixing time series.

\textbf{Notation: } Let $\cid$, $\cip$, and $\cae$ denote the convergence in distribution, convergence in probability, and convergence almost surely. We say $X_t| Y_t\cid X$ in probability for random vectors $X_t, X\in \mathbb{R}^d$ and $Y_t\in \mathbb{R}^{k_t}$ if $\sup_{f\in \textrm{BL}_1}|\ev\{f(X_t)|Y_t\}-\ev\{f(X)\}|\cip 0$ as $t\to \infty$, where $\textrm{BL}_1=\{f:\mathbb{R}^d\to \mathbb{R}: \sup_{x\in \mathbb{R}^d}|f(x)|\leq 1, |f(x)-f(y)|\leq |x-y| \textrm{ for all } x,y\in \mathbb{R}^d\}$, and the dimension $k_t$ is possibly diverging with $t$. For two sequence $a_t$ and $b_t$, we say $a_t\lesssim b_t$ if $a_t\leq Cb_t$ for some constant $C>0$ and all $t\geq 1$. We define $a_t\asymp b_t$ if  $a_t\lesssim b_t$ and $b_t\lesssim a_t$.  We use $\|\cdot\|$ to denote the Euclidean norm of a vector and Frobenius norm of a matrix.



\subsection{Challenges and Our Contribution}
Despite the fruitful results, conducting valid statistical inference using SGD based on mixing observations remains unexplored. The major difficulties in extending the existing distributional results to mixing time series can be summarized as twofold. 

Theoretically, the general dependence structure of mixing data significantly complicates the convergence analysis of SGD when compared to scenarios involving i.i.d. or Markov observations. Existing convergence analysis (\citealp{agarwal2012generalization,ma2022data}) hinges on the condition $\|\widehat{\theta}_t-\widehat{\theta}_{t-1}\|\to 0$, and establishing this condition without Markov assumptions on the observations remains unexplored. Moreover, the proof of limiting distribution of an estimator requires handling higher-order error terms, and the convergence rates of the higher-order error terms may not be fast enough to establish convergence in distribution due to dependence.

Methodologically, statistical inference such as building confidence intervals involves estimating covariance matrices or quantiles of the limiting distributions.  In the context of streaming data settings, several common approaches for building confidence intervals are employed, including  plug-in estimators (e.g., see \citealp{chen2020statistical,liu2022105017}), the random scaling method in \cite{lee2022fast,lee2022fastb}, self-normalized estimators (e.g., see \citealp{pelletier2000asymptotic,gahbiche2000estimation}), and the bootstrap SGD proposed in \cite{fang2018online}. It is worth noting that covariance matrices often encompass autocorrelation coefficients when observations are correlated. Effectively estimating these coefficients through plug-in methods presents a challenging problem. Additionally, in certain statistical models such as least absolute deviation regression or quantile regression (as exemplified in Examples \ref{example:LAD} and \ref{example:quantile}), the covariance matrices can become intricate and might involve nonparametric components, posing another challenge for plug-in methods. The random scaling method introduced by \cite{lee2022fast,lee2022fastb}  was designed to handle complicated asymptotic covariance matrices. This approach employs a normalization matrix to eliminate the dependence on the unknown covariance matrix, resulting in an asymptotically pivotal statistic. This method is efficient for handling large datasets and robust against variations in tuning parameters (e.g., learning rate) for SGD. Nonetheless, its limiting distribution is relying on a novel functional central limit theorem, and extending this framework to incorporate mixing observations would be challenging. Self-normalized estimators for covariance matrices were studied in \cite{pelletier2000asymptotic,gahbiche2000estimation}. Recently, \cite{chen2020statistical} and \cite{zhu2023online} developed variants known as batch-means estimators, which involve dividing the SGD trajectory into multiple batches. These estimators exhibit appealing theoretical properties like strong consistency and moment convergence in the context of i.i.d. data. However, their validity with correlated observations remains unexplored. Another practically convenient algorithm is the bootstrap SGD in  \cite{fang2018online},  which is rooted on the following estimator
\begin{eqnarray}
\widehat{\theta}_t^*=\Pi\left\{\widehat{\theta}_{t-1}^*-\gamma_t V_t\nabla l(Z_t, \widehat{\theta}_{t-1}^*)\right\}. \label{eq:vanilla:sgd:boostrap}
\end{eqnarray}
Here $V_t$'s are i.i.d. random weights with mean one and unit variance that are independent from the data. With i.i.d. data,  it has been demonstrated that:
\begin{eqnarray*}
\sqrt{T}\left(\frac{1}{T}\sum_{t=1}^T \widehat{\theta}_t-\theta^*\right)&\cid& \xi,\\
 \sqrt{T}\left(\frac{1}{T}\sum_{t=1}^T \widehat{\theta}_t^*-\frac{1}{T}\sum_{t=1}^T \widehat{\theta}_t\right)\bigg| Z_1,\ldots, Z_T &\cid& \widetilde{\xi} \quad \textrm{ in probability},
\end{eqnarray*}
where $\xi$ and $\widetilde{\xi} $ are two centered normal distributions sharing a common covariance matrix. Recently, \cite{ramprasad2022online} extended this method to Markov chains. As a matter of fact,  the above bootstrap algorithm typically ignores the  autocorrelation coefficients and may fail to produce reliable confidence intervals when the observations are not independent or Markovian, which is illustrated by the following concrete example.
\begin{proposition}\label{prop:failure:bootstrap:sgd}
Let $Y_1,\ldots, Y_T\in \mathbb{R}$ be a sequence such that $Y_t$ is independent from $\{Y_s: |t-s|\geq 2\}$. Moreover,  assume that $\ev(Y_t)=\theta^*$, $\ev[(Y_t-\theta^*)^2]=r(0)>0$, and $\ev[(Y_t-\theta^*)(Y_{t+1}-\theta^*)]=r(1)\neq 0$ for all $t$. Consider the following two iteration procedures:
\begin{eqnarray*}
\widehat{\theta}_t=\widehat{\theta}_{t-1}+\gamma_t  (Y_t-\widehat{\theta}_{t-1}),\;\;\widehat{\theta}_t^*=\widehat{\theta}_{t-1}^*+\gamma_t  V_t(Y_t-\widehat{\theta}_{t-1}^*),
\end{eqnarray*}
for $t=1,2,\ldots, T$. Here $\gamma_t=t^{-\rho}$ for some $\gamma\geq 0$, $\rho\in (1/2, 1)$, and $V_t$'s are i.i.d. random variables with mean one and unit variance, which are independent from $Y_t$'s. Then it follows that
\begin{eqnarray*}
\sqrt{T}\left(\frac{1}{T}\sum_{t=1}^T \widehat{\theta}_t-\theta^*\right)&\cid& N\left(0, r(0)+2r(1)\right),\\
 \sqrt{T}\left(\frac{1}{T}\sum_{t=1}^T \widehat{\theta}_t^*-\frac{1}{T}\sum_{t=1}^T \widehat{\theta}_t\right)\bigg| Z_1,\ldots, Z_T &\cid& N\left(0, r(0)\right) \textrm{ in probability}.
\end{eqnarray*}
\end{proposition}
In Proposition \ref{prop:failure:bootstrap:sgd}, the observations $Y_t$'s are neither independent nor necessarily Markovian. Consequently, the conditions in \cite{fang2018online} and \cite{ramprasad2022online} are not satisfied. As a result, this disparity gives rise to distinct limiting distributions between the SGD and bootstrap SGD. For instance, in cases with positive correlation (i.e., $r(1) > 0$), the confidence interval constructed by the bootstrap SGD procedure will generally exhibit a shorter length and consequently yield a lower confidence level. This observation is reinforced by our numerical findings in Section \ref{sec:simulation}. Hence, when dealing with correlated observations,  the bootstrap SGD may not be a reliable choice for interval estimation.

In this paper, we introduce a mini-batch SGD procedure involving block sampling, tailored for situations where observations are $\phi$-mixing. Concretely, we partition the entire time series into several blocks of increasing sizes. Subsequently, two mini-batch SGD trajectories are constructed, each estimated from alternate blocks. At the end of each iteration, we employ a weighted Polyak-Ruppert averaging procedure on both trajectories. Our contributions can be summarized as follows.
\begin{enumerate}[label=(\roman*)]
\item The proposed procedure makes use of the  ``independent block" trick from \cite{yu1994rates}. However, this trick cannot be directly applied in the context of stream data. Therefore, we construct  two mini-batch SGD trajectories based on alternate blocks, which is methodologically novel.
\item  We show that the proposed estimator is asymptotic normal and efficient under mild conditions.  Notably, the efficiency is achieved by a novel weighted Polyak-Ruppert averaging procedure. This result  signifies a noteworthy extension from conventional SGD to mini-batch SGD, spanning from i.i.d. observations to $\phi$-mixing time series.
\item We additionally  develop a mini-batch bootstrap  SGD  procedure for interval estimation, which is practically convenient. By employing the alternate block technique, we effectively account for correlations among observations, thus enabling the construction of reliable confidence intervals. Consequently, this approach effectively addresses the limitations of the bootstrap SGD illustrated in Proposition \ref{prop:failure:bootstrap:sgd}.
\end{enumerate}

The rest of the paper is organized as follows. Section \ref{sec:algorithm}
describes the proposed mini-batch SGD estimator. Section \ref{sec:asymptotic} provides the asymptotic results of the proposed estimator and develops a bootstrap procedure for interval estimation. In Section \ref{sec:simulation}, simulation studies are conducted to examine the finite-sample performances of the proposed procedure. We apply the proposed method to a real-world dataset in Section  \ref{sec:application}. The proofs of main results and additional numerical studies are deferred to the Supplement.

\section{Mini-batch Stochastic Gradient Descent via Block Sampling}\label{sec:algorithm}
Given a time series $\{Z_t\}_{t=1}^\infty$, we consider dividing the observations into nonoverlappling blocks whose indexes are
\begin{align*}
I_1&=\{1,\ldots, B_1\},& J_1&=\{B_1+1,\ldots, 2B_1\}\\
I_2&=\{2B_1+1,\ldots, 2B_1+B_2\},& J_2&=\{2B_1+B_2+1,\ldots, 2B_1+2B_2\},\\
&\vdots &\vdots&\\
I_t&=\left\{2\sum_{i=1}^{t-1}B_i+1,\ldots,  2\sum_{i=1}^{t-1}B_i+B_t\right\}, &J_t&=\left\{2\sum_{i=1}^{t-1}B_i+B_t+1,\ldots,  2\sum_{i=1}^{t}B_i\right\}.
\end{align*}
Here $B_t$'s are some predetermined integers. Using the above blocks, let us define random vectors $W_t^a=\{Z_i, i\in I_t\}$  and $W_t^b=\{Z_i, i\in J_t\}$. 
\begin{figure}[h!]
\centering
\begin{tikzpicture}[scale=0.45]
    \draw (-12,0)-- (2,0); 
    \draw (-12,0.2) -- (-12,-0.2) node[below] {};
    \draw (-9,0.2) -- (-9,-0.2) node[below] {};
    \draw (-6,0.2) -- (-6,-0.2) node[below] {};
    \draw (-2,0.2) -- (-2,-0.2) node[below] {};
    \draw (2,0.2) -- (2,-0.2) node[below] {};

    \path (2,1.5)-- (5,1.5) node[below,xshift=-0.8cm] {$\ldots$}; 

    \draw (5,0)-- (17,0); 
    \draw (5,0.2) -- (5,-0.2) node[below] {};
    \draw (11,0.2) -- (11,-0.2) node[below] {};
    \draw (17,0.2) -- (17,-0.2) node[below] {};

    \path (17,1.5)-- (18,1.5) node[below] {$\ldots$}; 
    
    \draw [decorate, decoration={brace,amplitude=10pt, mirror}] (-12,0) -- (-9,0) node [black,midway,yshift=-0.7cm] {$B_1$};
    \draw [decorate, decoration={brace,amplitude=10pt, mirror}] (-9,0) -- (-6,0) node [black,midway,yshift=-0.7cm] {$B_1$};
    \draw [decorate, decoration={brace,amplitude=10pt, mirror}] (-6,0) -- (-2,0) node [black,midway,yshift=-0.7cm] {$B_2$};
    \draw [decorate, decoration={brace,amplitude=10pt, mirror}] (-2,0) -- (2,0) node [black,midway,yshift=-0.7cm] {$B_2$};
    \draw [decorate, decoration={brace,amplitude=10pt, mirror}] (5,0) -- (11,0) node [black,midway,yshift=-0.7cm] {$B_t$};
    \draw [decorate, decoration={brace,amplitude=10pt, mirror}] (11,0) -- (17,0) node [black,midway,yshift=-0.7cm] {$B_t$};
    
	\filldraw [fill=blue!30!]  (-12,0) rectangle ++(3,3) node [black,midway] {$W_1^a$};
	\filldraw [fill=green!30!]  (-9,0) rectangle ++(3,3) node [black,midway] {$W_1^b$};
	\filldraw [fill=blue!30!]  (-6,0) rectangle ++(4,3) node [black,midway] {$W_2^a$};
	\filldraw [fill=green!30!]  (-2,0) rectangle ++(4,3) node [black,midway] {$W_2^b$};
	\filldraw [fill=blue!30!]  (5,0) rectangle ++(6,3) node [black,midway] {$W_t^a$};
	\filldraw [fill=green!30!]  (11,0) rectangle ++(6,3) node [black,midway] {$W_t^b$};
\end{tikzpicture}
\caption{Visualization of the data blocks.}
\label{figure:block}
\end{figure}
To estimate $\theta^*$, the following mini-batch stochastic gradient descent algorithm is proposed:
\begin{eqnarray}
\theta_t^k=\Pi\left\{\theta_{t-1}^k-\gamma_t \widehat{H}_t(W_t^k, \theta_{t-1}^k)\right\}, \textrm{ for } k=a,b. \label{eq:iteration:estimator}
\end{eqnarray}
Here $\Pi$ is the projection operator on to $\Theta$, $\gamma_t$ is the learning rate, and 
\begin{eqnarray*}
\widehat{H}_t(W_t^a,  \theta)=\frac{1}{B_t}\sum_{i\in I_t}\nabla l(Z_i, \theta),\;\; \widehat{H}_t(W_t^b,  \theta)=\frac{1}{B_t}\sum_{i\in J_t}\nabla l(Z_i, \theta). 
\end{eqnarray*}
At the end of $T$-th iteration, 
the estimator is constructed based on the following weighted  Polyak-Ruppert averaging procedure:
\begin{eqnarray*}
\overline{\theta}_T=\frac{1}{2\sum_{t=1}^T B_t}\sum_{t=1}^TB_t(\theta_t^a+\theta_t^b).
\end{eqnarray*}
If we define $n_t=2\sum_{i=1}^t B_i$ to be the number of observations used in the first $t$-iterations, then the above weighted  Polyak-Ruppert estimator can be computed in an online fashion as follows:
\begin{eqnarray*}
n_t=n_{t-1}+2B_t,\quad \overline{\theta}_t=\frac{n_{t-1}}{n_t}\overline{\theta}_{t-1}+\frac{1}{n_t}B_t(\theta_t^a+\theta_t^b).
\end{eqnarray*}
Hence, the proposed algorithm is memory-efficient, requiring the storage of only the values of $n_t$ and $\overline{\theta}_t$ during the iteration.

The concept behind the proposed mini-batch SGD algorithm is inspired by the "independence block" trick presented in \cite{yu1994rates}. To illustrate this, let us define $\mcD_t^a = {W_s^a, s=1, \ldots, t}$ as the data utilized up to the $t$-th iteration of $\theta_t^a$. With $B_{t-1}$ observations being between $\mcD_{t-1}^a$ and $W_t^a$, the $\phi$-mixing condition (formally defined in Section \ref{sec:consistency}) implies that $\ev\{\widehat{H}_t(W_t^a, \theta_{t-1}^a)|\theta_{t-1}^a\}\approx \nabla L(\theta_{t-1}^a)$ under the condition that $B_{t-1}$ is sufficiently large. This approximation significantly contributes to our analytical framework.

\begin{figure}[h!]
\centering
\begin{tikzpicture}[scale=0.45]
    \draw (-12,0)-- (-2,0); 
    \draw (-12,0.2) -- (-12,-0.2) node[below] {};
    \draw (-9,0.2) -- (-9,-0.2) node[below] {};

    \draw (-6,0.2) -- (-6,-0.2) node[below] {};
    \draw (-2,0.2) -- (-2,-0.2) node[below] {};
    \filldraw [fill=yellow!30!]  (-12,0) rectangle ++(18,5) node [black,midway, yshift=0.76cm] {$\mcD_{t-1}^a$};
    
    \path (-2,1.5)-- (1,1.5) node[below,xshift=-0.8cm] {$\ldots$}; 
    
    \draw (1,0)-- (17,0); 
    \draw (1,0.2) -- (1,-0.2) node[below] {};
    \draw (6,0.2) -- (6,-0.2) node[below] {};
    \draw (11,0.2) -- (11,-0.2) node[below] {};
    \draw (17,0.2) -- (17,-0.2) node[below] {};
    \path (17,1.5)-- (18,1.5) node[below] {$\ldots$}; 

    \draw [decorate, decoration={brace,amplitude=15pt, mirror}] (6,0) -- (11,0) node [black,midway,yshift=-0.8cm] {$B_{t-1}$};

	\filldraw [fill=blue!30!]  (-12,0) rectangle ++(3,3) node [black,midway] {$W_1^a$};
	\filldraw [fill=blue!30!]  (-6,0) rectangle ++(4,3) node [black,midway] {$W_2^a$};
	\filldraw [fill=blue!30!]  (1,0) rectangle ++(5,3) node [black,midway] {$W_{t-1}^a$};
     \filldraw [fill=blue!30!]  (11,0) rectangle ++(6,3) node [black,midway] {$W_{t}^a$};

\end{tikzpicture}
\caption{Visualization of ``independence block" trick.}
\label{figure:block:trick}
\end{figure}

The proposed algorithm can be applied to many important statistical models, and we provide some examples below.

\begin{Example}[label=linear]\label{example:linear}(Ordinary Least Squares (OLS) Regression) Let the random vector be $Z=(Y, X^\top)^\top$ with $Y\in \mathbb{R}$ and $X\in \mathbb{R}^d$ satisfying $Y=X^\top\theta^*+\epsilon$. Here  $\epsilon\in \mathbb{R}$ is the random noise. The loss function can be chosen as $l(z, \theta)=(y-x^\top \theta)^2/2$, and the corresponding gradient is $\nabla l(z, \theta)=-(y-x^\top \theta)x$.
\end{Example}
\begin{Example}[label=lad]\label{example:LAD}(Least Absolute Deviation (LAD) Regression) Consider the same model in Example \ref{example:linear}, the LAD regression has a loss function  $l(z, \theta)=|y-x^\top \theta|$ and a subgradient $\nabla l(z, \theta)=-\sign(y-x^\top \theta)x$. Here $\sign(\cdot)$ is the sign function.
\end{Example}
\begin{Example}[label=logistic]\label{example:logistic} (Logistic Regression) Suppose that the observation $Z=(Y, X^\top)^\top$ with $Y\in \{-1, 1\}$ and $X\in \mathbb{R}^d$ satisfies $\pr(Y=y|X=x)=[1+\exp(-yx^\top\theta^*)]^{-1}$. The loss function is $l(z, \theta)=\log(1+\exp(-yx^\top\theta))$ with the gradient $\nabla l(z, \theta)=-yx [1+\exp(yx^\top\theta)]^{-1}$.
\end{Example}
\begin{Example}[label=quantile]\label{example:quantile} (Quantile Regression) Let the random vector be $Z=(Y, X^\top)^\top$ with $Y\in \mathbb{R}$ and $X\in \mathbb{R}^d$ satisfying $Y=X^\top\theta^*+\epsilon$ and $\pr(\epsilon\leq 0|X)=\tau$ for a fixed quantile $\tau\in (0, 1)$. The loss function becomes $l(z, \theta)=(y-x^\top\theta)(\tau-I\{y\leq x^\top\theta\})$ with a subgradient $\nabla l(z, \theta)=x(I\{y\leq x^\top\theta\}-\tau)$.
\end{Example}

\section{Asymptotic Theory}\label{sec:asymptotic}
In this section, we establish some asymptotic results for the proposed estimator, including consistency and the limiting distribution. Moreover, we also develop a valid bootstrap procedure for interval estimation.
\subsection{Consistency}\label{sec:consistency}
In the following, we will show that the proposed estimator $\overline{\theta}_T$ is consistent. Before moving forward, let us review the definition of $\phi$-mixing and state some technical assumptions.
\begin{Definition}
Let $(\Omega, \mcC, Q)$  be a probability space, and let $\mcA, \mcB$ be two sub-sigma algebras of $\mcC$. The $\phi$-mixing coefficient of $\mcA$ and $\mcB$ is defined as
\begin{eqnarray*}
\phi_Q(\mcB, \mcA)=\sup_{A\in \mcA, B\in \mcB, Q(B)>0}|Q(B|A)-Q(B)|.
\end{eqnarray*}
Let $\bfZ=\{Z_t\}_{t=1}^\infty$ be a stationary time series, and let $\mcF_{s}^{k}$ be the sigma algebra generated by $\{Z_t\}_{t=s}^k$ for $1\leq s \leq k\leq \infty$. Moreover, let $Q$ be the joint measure of $\{Z_t\}_{t=1}^\infty$ define on $\mcF_1^\infty$.  The $\phi$-mixing coefficient of the time series is defined as
\begin{eqnarray*}
\phi_{\bfZ}(t)= \sup_{s\geq 1}\phi_Q(\mcF_1^{s}, \mcF_{s+t}^\infty).
\end{eqnarray*}
\end{Definition}
Intuitively, the mixing coefficients characterize the dependence between the observations $\{Z_i\}_{i=1}^s$ and $\{Z_i\}_{i=s+t}^\infty$. Notably, $\phi$-mixing is considered as a more general dependence assumption than Markovian dependence (\citealp{davydov1973mixing}). In the literature, it is often assumed that the mixing coefficients converge to zero at different rates when the time gap $t\to \infty$, and we provide two commonly used settings below (e.g., see \citealp{Ronmixing1997}).
\begin{itemize}
\item  Algebraically $\phi$-mixing:  $\phi_\bfZ(t)\leq c t^{-\eta}$ for some $c, \eta>0$. 
\item  Geometrically  $\phi$-mixing:  $\phi_\bfZ(t)\leq c_1\exp(-c_2 t^\eta)$ for some $c_1, c_2, \eta>0$.
\end{itemize}
Examples of $\phi$-mixing data include autoregressive processes  (\citealp{athreya1986mixing}), moving average processes (\citealp{fanYao2003}), Metropolis-Hastings samplers (\citealp{Jarner2002AOS}),  queuing systems (\citealp{agarwal2012generalization}), some Markov chains (\citealp{davydov1973mixing,meyn2012markov}). In particular,  it can be verified that the moving average processes with finite orders are  geometrically  $\phi$-mixing but not Markovian.

\begin{Assumption}\label{Assumption:A0}
\begin{enumerate}[label=(\roman*)]
\item  \label{A1:mixing:rate} The sequence  $\bfZ=\{Z_t\}_{t=1}^\infty$ is stationary and $\phi$-mixing. Moreover, its mixing coefficients satisfy  $\phi_\bfZ(t)\leq \phi(t)$ for some  sequence $\phi(t)\downarrow 0$.
\item  \label{A1:positive:variance} Let $r(t)=\ev\{\nabla l(Z_{k+t}, \theta^*) \nabla l^\top(Z_{k}, \theta^*)\}\in \mathbb{R}^{d\times d}$ for $t\geq 0$, and we assume that the matrix $r(0)+2\sum_{t=1}^\infty r(t)$ is positive definite.
\end{enumerate}
\end{Assumption}

\begin{Assumption}\label{Assumption:A1}
There exists a constants $K>0$ such that the following conditions hold.
\begin{enumerate}[label=(\roman*)]
\item \label{A1:compact:Theta} The parameter space $\Theta$ is a compact subset of $\mathbb{R}^d$.
\item \label{A1:identification} The objective function $L(\theta)$ is strongly convex and twice continuously differentiable in $\Theta$. Moreover,  $\theta^*$ is the unique solution of $\nabla L(\theta)=0$, and $\theta^*$ is in the interior of $\Theta$.
\item  \label{A1:lipchiz:hession} The Hessian matrix $G=\nabla^2 L(\theta^*)\in \mathbb{R}^{d\times d}$ is positive definite. Furthermore,  the inequality $\|\nabla^2 L(\theta)-\nabla^2 L(\theta^*)\|\leq K\|\theta-\theta^*\|$ holds for all $\theta\in \Theta$.
\item \label{A1:moment:condition} For some even integer $p\geq 4$ and some function $M(\cdot)$, it holds that $\sup_{\theta \in \Theta}\|\nabla l(Z, \theta)\|\leq M(Z)$ and $\ev\{M^p(Z)\}<\infty$.
\item \label{A1:expectation:H:theta}  The gradient (or subgradient) satisfies that $\ev\{\nabla l(Z, \theta)\}=\nabla L(\theta)$ for all $\theta\in \Theta$.
\end{enumerate}
\end{Assumption}
\begin{Assumption}\label{Assumption:A2}
Let $p\geq 4$ be the constant introduced in Assumption \ref{Assumption:A1}\ref{A1:moment:condition}. The following conditions hold.
\begin{enumerate}[label=(\roman*)]
\item \label{A2:leanring:rate} The learning rate satisfies $\gamma_t\asymp t^{-\rho}$ for some $\rho\in (1/2, 1)$.
\item \label{A2:mixing:condition} The mixing coefficient satisfies $\sum_{t=1}^\infty \phi^{\frac{1}{2}}(t)<\infty$.
\item  \label{A2:batch:rate} The batch size satisfies $B_t\asymp t^{\beta}$ for some $\beta> 0$, and $\sum_{t=1}^\infty \gamma_t \phi^{\frac{1}{2}-\frac{1}{p}}(B_t)<\infty$. 
\end{enumerate}
\end{Assumption}
Assumption \ref{Assumption:A0} constitutes the distributional assumptions governing the time series. Specifically, Assumption \ref{Assumption:A0}\ref{A1:mixing:rate} requires the stationarity of the time series and diminishing $\phi$-mixing coefficients, which aligns with similar conditions found in \cite{yu1994rates}. While it is possible to relax the stationarity condition, an alternative approach is to allow $Z_t$ to have a time-varying marginal distribution $F_t$ that gradually converges to the marginal distribution $F$ of $Z$ in~\eqref{eq:model}. In that case, one must impose rate conditions on the convergence $F_t \to F$ (e.g., see \citealp{ramprasad2022online, li2023SGDMarkov}). Because this extension requires more delicate arguments, we assume for simplicity that $Z_t$ shares a common distribution with $Z$. Assumption \ref{Assumption:A0}\ref{A1:positive:variance} serves as a conventional assumption to ensure the positive definiteness of the asymptotic covariance matrix for dependent observations. This requirement is frequently encountered to establish the limiting distributions in the context of dependent data (e.g., see \citealp{fanYao2003}).

The conditions outlined in Assumption \ref{Assumption:A1} are closely aligned with existing works. For instance, the compactness condition in Assumption \ref{Assumption:A1}\ref{A1:compact:Theta} is a standard requirement in parametric estimation problems (e.g., see \citealp{ferguson2017course}). Assumptions \ref{Assumption:A1}\ref{A1:identification} through \ref{Assumption:A1}\ref{A1:moment:condition} introduce regularity conditions concerning the loss function $l(z, \theta)$ and its expected value $L(\theta)$, which were commonly imposed in the SGD literature (e.g., see \citealp{chen2020statistical,su2018uncertainty,liu2022105017}).  Assumption \ref{Assumption:A1}\ref{A1:expectation:H:theta} indicates that  $\nabla l(Z, \theta)$ should be an unbiased estimator of $\nabla L(\theta)$, and this condition plays a crucial role in the validity of SGD. All the conditions in Assumption \ref{Assumption:A1} can be verified for Examples \ref{example:linear}-\ref{example:logistic}.

Assumption \ref{Assumption:A2} introduces some rate conditions. More specifically, Assumption \ref{Assumption:A2}\ref{A2:leanring:rate} outlines the learning rate for the $t$-th iteration. It satisfies that $\sum_{t=1}^\infty \gamma_t = \infty$ and $\sum_{t=1}^\infty \gamma_t^2 < \infty$, which were commonly assumed in the literature  (e.g., \citealp{polyak1992acceleration, fang2018online, su2018uncertainty}). Assumption \ref{Assumption:A2}\ref{A2:mixing:condition} is the $\phi$-mixing condition, assuming weak correlation among $Z_t$'s. The purpose of the rate condition $\sum_{t=1}^\infty \phi^{1/2}(t)<\infty$  is to make use of the moment inequality developed in \cite{yokoyama1980moment}. Assumption \ref{Assumption:A2}\ref{A2:batch:rate} indicates that the batch size $B_t$ should diverge  with a parameter $\beta>0$ controlling for its speed. Hereafter, we briefly discuss the sufficient conditions of $B_t$ in the algebraically and geometrically $\phi$-mixing scenarios.

\begin{itemize}
\item  Algebraically $\phi$-mixing:  If $\phi(t)= c t^{-\eta}$ for some $c>0, \eta>0$, then \ref{Assumption:A2}\ref{A2:batch:rate} is satisfied with $\beta>2p(1-\rho)/(p\eta-2\eta)$. This indicates the existence of suitable batch sizes. 
\item  Geometrically  $\phi$-mixing: If  $\phi(t)=c_1\exp(-c_2 t^\eta)$ for some $c_1, c_2, \eta>0$, then \ref{Assumption:A2}\ref{A2:batch:rate} is satisfied with $\beta>0$. Due to the extremely weak dependence in the geometrically $\phi$-mixing scenarios, we allow the batch sizes grow at any polynomial rate. 
\end{itemize}

Based on the above assumptions, the following theorem justifies the validity of the proposed procedure.
\begin{theorem}\label{theorem:consistency}
Under Assumptions \ref{Assumption:A0}-\ref{Assumption:A2}, the following statements hold.
\begin{enumerate}[label=(\roman*)]
\item \label{item:theorem:consistency:1} For $k=a, b$, it holds that $\theta_T^k\cae \theta^*$ and $\overline{\theta}_T\cae \theta^*$ as $T\to \infty$.
\item \label{item:theorem:consistency:2} There is a constant $C>0$ such that $\ev(\|\theta_t^k-\theta^*\|^2)\leq C\gamma_t+C\phi^{\frac{1}{2}-\frac{1}{p}}(B_t)$ for all $t\geq 1$ and $k=a, b$.
\end{enumerate}
\end{theorem}

Statement \ref{item:theorem:consistency:1} in Theorem \ref{theorem:consistency} shows that the proposed SGD estimators are strongly consistency. Similar results were obtained by \cite{polyak1992acceleration} with i.i.d. observations. The proof relies the Robbins-Siegmund Theorem in \cite{robbins1971convergence}. Statement \ref{item:theorem:consistency:2} in Theorem \ref{theorem:consistency} is closely related but different from the existing works. In the i.i.d. setting with $B_t=1$, it well known (e.g., \citealp{pelletier2000asymptotic}) that the second moment of the SGD estimator is of a $\gamma_t$ order. In comparison, our result involves an additional order $\phi^{1/2-1/p}(B_t)$ due to $\phi$-mixing.




\subsection{Asymptotic Normality}
To establish the limiting distribution of the proposed estimator, additional conditions about the batch size should be imposed, which are used to handle some higher-order error terms.
\begin{Assumption}\label{Assumption:A3}
Let $p, \rho, \beta$ be the parameters introduced in Assumptions \ref{Assumption:A1}\ref{A1:moment:condition} and \ref{Assumption:A2}. We assume they satisfy the following conditions.
\begin{enumerate}[label=(\roman*)]
\item \label{A3:rate:phi}  $0<\beta<\min\{2\rho-1, 1-\rho\}$. 
\item \label{A3:rate:batch}  $\lim_{t\to \infty}t^{\frac{2\rho+\beta+1}{\beta}}\phi^{\frac{1}{2}-\frac{1}{p}}(t)=0$. 
\end{enumerate}
\end{Assumption}
Compared with the lower bound condition of $B_t$ in Assumption \ref{Assumption:A2}\ref{A2:batch:rate}, Assumption \ref{Assumption:A3} presents some sufficient conditions to control for the growing rate of $B_t\asymp t^{\beta}$ including both upper and lower bounds. Specifically, the upper bound in  Assumption \ref{Assumption:A3}\ref{A3:rate:phi} is to bound the bias terms introduced during the iteration of our algorithm, while the lower bound  in \ref{Assumption:A3}\ref{A3:rate:batch}  is imposed to handle the bias due to correlation. 
\begin{Remark}\label{remark:tuning:selection}
Notice that  the function $\min\{2\rho-1, 1-\rho\}$ is maximized at $\rho=2/3$. With this choice of $\rho$,  Assumption \ref{Assumption:A3} is simplified to $\beta\in (0, 1/3)$ and  $\lim_{t\to \infty}t^{7/(3\beta)+1}\phi^{1/2-1/p}(t)=0$. Hence, we obtain the following sufficient conditions to guarantee the existence of $\beta$ satisfying \ref{Assumption:A3} for algebraically and geometrically mixing sequences.
\begin{itemize}
\item  Algebraically $\phi$-mixing:  If $\phi(t)= c t^{-\eta}$ for some $c>0, \eta> 16p/(p-2)$, then the choices $\rho=2/3$ and $\beta\in (1/(3+\Delta), 1/3)$ satisfy  Assumption \ref{Assumption:A3}. Here $\Delta=3(p-2)\delta/(14p)$ with $\delta=\eta-16p/(p-2)>0$.
\item  Geometrically  $\phi$-mixing: If  $\phi(t)=c_1\exp(-c_2 t^\eta)$ for some $c_1, c_2, \eta>0$, then the choices $\rho=2/3$ and $\beta\in (0,1/3)$ satisfy   Assumption \ref{Assumption:A3}.
\end{itemize}
Based on the above discussion, we suggest using $\rho=2/3$ and $\beta\approx 1/3$. Our numerical results in Section \ref{sec:simulation} also support this suggestion.
\end{Remark}

\begin{theorem}\label{theorem:asymptotic:normality} Suppose that Assumptions \ref{Assumption:A0}-\ref{Assumption:A3} hold. Then it follows that
\begin{eqnarray*}
\sqrt{2\sum_{t=1}^T B_t}(\overline{\theta}_T-\theta^*)\cid N(0, \Sigma),
\end{eqnarray*}
where $\Sigma=G^{-1}\{r(0)+2\sum_{k=1}^\infty r(k)\}G^{-1}$ with $r(k)=\ev\{\nabla l(Z_{t+k}, \theta^*) \nabla l^\top(Z_{t}, \theta^*)\}$ being the autocorrelation coefficient.
\end{theorem}

Theorem \ref{theorem:asymptotic:normality} reveals the asymptotic normality of the proposed estimator $\overline{\theta}_T$, and the asymptotic covariance matrix has a sandwich form.  In the special case with i.i.d. observations,  since $r(k)=0$ when $k\geq 1$, Theorem \ref{theorem:asymptotic:normality} implies that
\begin{eqnarray*}
\sqrt{n_T}(\overline{\theta}_T-\theta^*)\cid N(0, G^{-1}r(0)G^{-1}),
\end{eqnarray*}
where $n_T=2\sum_{t=1}^T B_t$ is the sample size.  The above convergence coincides with the existing results in \cite{polyak1992acceleration}, and the asymptotic covariance matrix attains the efficiency bound.  Hence, Theorem \ref{theorem:asymptotic:normality} essentially is a generalization from i.i.d. observations to $\phi$-mixing data and from SGD to mini-batch SGD.


We end this section by discussing the expressions of $G$ and $\nabla l(z, \theta^*)$ in Examples \ref{example:linear}-\ref{example:quantile}.
\begin{Example}[continues=linear]
Suppose the error term satisfies that $\ev(\epsilon X)=0$. It can be verified that $\nabla l(Z,\theta^*)=-(Y-X^\top \theta^*)X=-\epsilon X$ and $G=\ev(XX^\top)$ provided the corresponding expected value exists.
\end{Example}

\begin{Example}[continues=lad] 
Let  $f_\epsilon(\cdot |x)$ be the conditional  density  of $\epsilon$ given $X=x$, and suppose that the conditional median of $\epsilon$ given $X$ is zero almost surely. We can verify that  $\nabla l(Z,\theta^*)=-X\sign(Y-X^\top \theta^*)=-X\sign(\epsilon)$ and $G=2\ev\{f_\epsilon(0|X)XX^\top\}$ under some mild regularity conditions.
\end{Example}

\begin{Example}[continues=logistic]
Under some standard assumptions of logistic regression, it is not difficult to verify that
\begin{eqnarray*}
\nabla l(Z,\theta^*)=\frac{-YX}{1+\exp(YX^\top\theta^*)},\; G=\ev\bigg(\frac{\exp(X^\top \theta^*)}{[1+\exp(X^\top \theta^*)]^2}XX^\top\bigg).
\end{eqnarray*}
\end{Example}

\begin{Example}[continues=quantile]
Let  $f_\epsilon(\cdot |x)$ be the conditional  density  of $\epsilon$ given $X=x$, and suppose that $\pr(\epsilon\leq 0|X)=\tau$ almost surely. We can verify that  $\nabla l(Z,\theta^*)=X(I\{Y\leq X^\top\theta^*\}-\tau)=X(I\{\epsilon\leq 0\}-\tau)$ and $G=\ev\{f_\epsilon(0|X)XX^\top\}$ under some regularity conditions.
\end{Example}

\subsection{Bootstrap Inference}\label{sec:bootstrap}
The results developed  in previous sections provide theoretical justifications of the proposed estimator. However, the asymptotic covariance matrix $\Sigma$ is generally unknown in practice. Classical plug-in estimators of the covariance matrix may be inefficient or infeasible in stream data settings as it may require estimating nonparametric components (please see Examples \ref{example:LAD} and \ref{example:quantile}) and autocorrelation coefficients $r(k)$. As shown in Proposition \ref{prop:failure:bootstrap:sgd}, the bootstrap SGD in \cite{fang2018online, ramprasad2022online} may ignore the correlation among data. To address the limitations, we proposed a  mini-batch bootstrap SGD algorithm for interval estimation. Before proceeding, we first introduce the following  mini-batch bootstrap SGD estimators:
\begin{eqnarray}
\theta_t^{*k}=\Pi\left\{\theta_{t-1}^{*k}-\gamma_t V_t \widehat{H}_t(W_t^k, \theta_{t-1}^{*k})\right\}, \textrm{ for } k=a,b,  \label{eq:iteration:estimator:bootstrap}
\end{eqnarray}
where $V_t$'s are i.i.d random variables with mean one and unit variance that are independent from $Z_t$'s.  The bootstrap version of the averaged estimator  is given by
\begin{eqnarray*}
\overline{\theta}_T^*=\frac{1}{2\sum_{t=1}^T B_t}\sum_{t=1}^TB_t(\theta_t^{*a}+\theta_t^{*b}).
\end{eqnarray*}
The following theorem justifies the large sample properties of $\overline{\theta}_T^*$.

\begin{theorem}\label{theorem:bootstrap:asymptotic}
Suppose that Assumptions \ref{Assumption:A0}-\ref{Assumption:A3} are satisfied and $\ev(|V_t|^p)<\infty$ with $p$ being the moment parameter defined in Assumption \ref{Assumption:A1}\ref{A1:moment:condition}.  Then it holds that
\begin{eqnarray*}
\sqrt{\sum_{t=1}^T 2B_t}(\overline{\theta}_T^*-\overline{\theta}_T)\bigg|\mcD_T \cid N(0, \Sigma) \textrm{ in probability},
\end{eqnarray*}
where $\mcD_T=\{Z_i, i\in I_t\cup J_t, t=1,\ldots, T\}$ is the collection of observations used in the first $T$ iterations, and $\Sigma$ is the covariance matrix in Theorem \ref{theorem:asymptotic:normality}.
\end{theorem}

Loosely speaking, Theorem \ref{theorem:bootstrap:asymptotic} shows that the conditional distribution of $\overline{\theta}_T^*-\overline{\theta}_T$ given $\mcD_T$ is asymptotically equivalent to the distribution of $\overline{\theta}_T-\theta^*$. A practical implication of this equivalence is building confidence intervals for $g(\theta^*)$ for some differentiable function $g:\Theta \to \mathbb{R}$. To be specific, we can simultaneously obtain $N$ random samples of $\overline{\theta}^*_T$, say $\overline{\theta}^{*(1)}_T, \ldots, \overline{\theta}^{*(N)}_T$. A $(1-\alpha)\times 100\%$ level confidence interval of $g(\theta^*)$ can be constructed using the $\alpha/2$ and $1-\alpha/2$ quantiles of the samples $g(\overline{\theta}^{*(1)}_T), \ldots, g(\overline{\theta}^{*(N)}_T)$. We formally summarize this procedure in Algorithm \ref{alg:bootstrap}. Compared with the bootstrap SGD procedure in \cite{fang2018online} and \cite{ramprasad2022online}, Proposition \ref{prop:failure:bootstrap:sgd} and Theorem \ref{theorem:bootstrap:asymptotic} together highlight the necessarity of alternating block trick for bootstrap in $\phi$-mixing data.
\begin{algorithm}[h!]
\caption{Bootstrap Mini-batch SGD Algorithm}\label{alg:bootstrap}
\SetKwInOut{Input}{Input}
\SetKwInOut{Output}{Output}
\KwIn{Data $\{Z_t\}_{t=1}^\infty$, learning rate $\gamma_t$, block size $B_t$, bootstrap sample size $N$, number of iterations $T$, initial values $\theta_0^s, \theta_0^{*s,(j)}, \overline{\theta}_0=0, \overline{\theta}_0^{*(j)}=0$ for $j=1,\ldots, k$ and $s=a,b$, significance level $\alpha\in (0, 1)$}
\For{$t=1$ \KwTo $T$}
	{   
		Construct indexes sets $I_t, J_t$ according to the block size $B_t$

		Update $\theta_t^s=\Pi\left\{\theta_{t-1}^s-\gamma_t \widehat{H}_t(W_t^s, \theta_{t-1}^s)\right\}$ for $s=a,b$
		
		Update $n_t=n_{t-1}+2B_t$		
		
		Update $\overline{\theta}_t=\frac{n_{t-1}}{n_t}\overline{\theta}_{t-1}+\frac{1}{n_t}B_t(\theta_t^a+\theta_t^b)$
		
		
		\For{$j=1$ \KwTo $N$}
		{
			Generate the random weight $V_t^{(j)}$
			
			Update $\theta_t^{*s,(j)}=\Pi\left\{\theta_{t-1}^{*s,(j)}-\gamma_t V_t^{(j)}\widehat{H}_t(W_t^s, \theta_{t-1}^{*s,(j)})\right\}$ for $s=a,b$
			
			Update $\overline{\theta}_t^{*(j)}=\frac{n_{t-1}}{n_t}\overline{\theta}_{t-1}^{*(j)}+\frac{1}{n_t}(\theta_t^{*a,(j)}+\theta_t^{*b,(j)})$
		}
	} 

Let $u_{\alpha}$ and $l_\alpha$ be the the $1-\alpha/2$ and $\alpha/2$ quantiles of $g(\overline{\theta}_T^{*(1)}), \ldots, g(\overline{\theta}_T^{*(N)})$

\KwOut{point estimator $g(\overline{\theta}_T)$, confidence interval $(l_\alpha, u_\alpha)$}
\end{algorithm}

\section{Monte Carlo Simulation}\label{sec:simulation}
This section presents a series of simulation studies to examine the finite-sample performance of the proposed procedure. We consider eight distinct models, with Models 1 and 2 serving to investigate  the assertions in Proposition \ref{prop:failure:bootstrap:sgd}.

\begin{enumerate}[wide, labelwidth=!, labelindent=0pt]
\item[\textbf{Model 1 (i.i.d. Mean Estimation):}]  $Z_t$'s are i.i.d. generated from $N(\theta^*, 0.5^2)$ with $\theta^*=0$. The loss function is $l(z, \theta)=(z-\theta)^2/2$.

\item[\textbf{Model 2 (Mixing Mean Estimation):}]  $Z_t=(Y_{t}+Y_{t+1})/2$, where $Y_t$'s are i.i.d. generated from $N(\theta^*, 1)$ with $\theta^*=0$. We chose the loss function to be $l(z, \theta)=(z-\theta)^2/2$.

\item[\textbf{Model 3 (i.i.d. Median Estimation):}]  Observations are generated by the same procedure in Model 1.  We use the loss function $l(z, \theta)=|z-\theta|$ to estimate the median.

\item[\textbf{Model 4 (Mixing Median Estimation):}] Observations are generated by the same procedure in Model 2.  We consider the loss function $l(z, \theta)=|z-\theta|$.

\item[\textbf{Model 5 (i.i.d. OLS):}] $Y_t=X_t^\top \theta^*+\epsilon_t$ where $\theta^*=(\theta_1^*, \theta_2^*, \theta_3^*)^\top=(-0.2,0.3,0.1)^\top$. Here $X_t=(X_{t1}, X_{t2}, X_{t3})^\top \in \mathbb{R}^3$ with $X_{t1}\sim N(1, 0.1)$, $X_{t2}\sim N(1, 0.5)$, and $X_{t3}\sim N(-2, 1)$ being independent normal random variables. The noise terms $\epsilon_t$'s are i.i.d. $N(0, 1)$. The linear regression in Example \ref{example:linear} is applied.

\item[\textbf{Model 6 (Mixing OLS):}] Data is generated by the same setting as that in Model 5, except that $\epsilon_t=\widetilde{\epsilon}_t/3$ with $\widetilde{\epsilon}_t=0.5\widetilde{\epsilon}_{t-1}+0.4\widetilde{\epsilon}_{t-2}+N(0,1)$ being an $AR(2)$ process. The linear regression in Example \ref{example:linear} is applied.

\item[\textbf{Model 7 (i.i.d. LAD):} ] The same setting in Model 5 is considered, except that the LAD regression in Example \ref{example:LAD} is applied.

\item[\textbf{Model 8 (Mixing LAD):} ] The same setting in Model 6 is considered, except that the LAD regression in Example \ref{example:LAD} is applied.
\end{enumerate}

We consider different sample sizes with $n=20000, 50000, 100000$. To see the impact of batch-size, we select $B_t=\floor{t^{\beta}}\vee 1$ with $\beta=0.2, 0.3, 0.33, 0.5, 0.7$. According to the discussion in Remark \ref{remark:tuning:selection}, we include the choice $\beta=0.33$ and use the learning rate $\gamma_t=(t+10)^{-2/3}$. For each setting, we make use of the bootstrap procedure in Section \ref{sec:bootstrap} with $N=500$ bootstrap samples and $V_t\sim Exp(1)$ to construct the  $95\%$ confidence intervals. The corresponding coverage probabilities (CP) and root mean square errors (RMSE) are recorded based on $500$ replications.  Finally,  we also consider the SGD and the bootstrap SGD (denoted as $\beta=0$) procedure in \cite{fang2018online, ramprasad2022online} as a competitor with the same setting. 

The numerical results are summarized in Figures \ref{figure:rmse:model1:4}-\ref{figure:cp:model8}, and several interesting findings can be observed. First, as shown in Figures \ref{figure:rmse:model1:4} and \ref{figure:RMSE:model5:8},  the RMSE decreases when the sample size is increasing for all settings and algorithms. Moreover, for a fixed sample size, employing a smaller batch size ($\beta$) generally leads to lower RMSE values. Second, considering Models 1, 3, 5, and 7, the bootstrap SGD ($\beta=0$) procedure exhibits consistency, maintaining a CP around $95\%$ regardless of the sample size. However, in cases where data is correlated, such as Models 2, 4, 6, and 8, the confidence intervals provided by the bootstrap SGD are not reliable. For instance, Figure \ref{figure:cp:model1:4} demonstrates that the CP of Model 2 hovers around $85\%$. This observation aligns with the implications of Proposition \ref{prop:failure:bootstrap:sgd}. Specifically, in Model 2, the observations are positively correlated with $r(1)=1/4>0$ and $r(k)=0$ for $k\geq 2$. Therefore, Proposition \ref{prop:failure:bootstrap:sgd} indicates that the bootstrap SGD procedure yields shorter confidence intervals due to ignoring the correlation. Third, Figure \ref{figure:cp:model7} underscores the significance of batch size in our proposed algorithm. Notably, CPs for $\theta_1-\theta_3$ in Model 7 deviate from the $95\%$ nominal level when $\beta=0.7$, across all sample sizes. Finally, for all settings, the choice  $\beta=0.33$  can produce confidence intervals with valid coverage probabilities, thereby supporting our parameter tuning recommendations in Remark \ref{remark:tuning:selection}.

%
%

\begin{figure}[htp!]
\centering
\includegraphics[width=2.5 in]{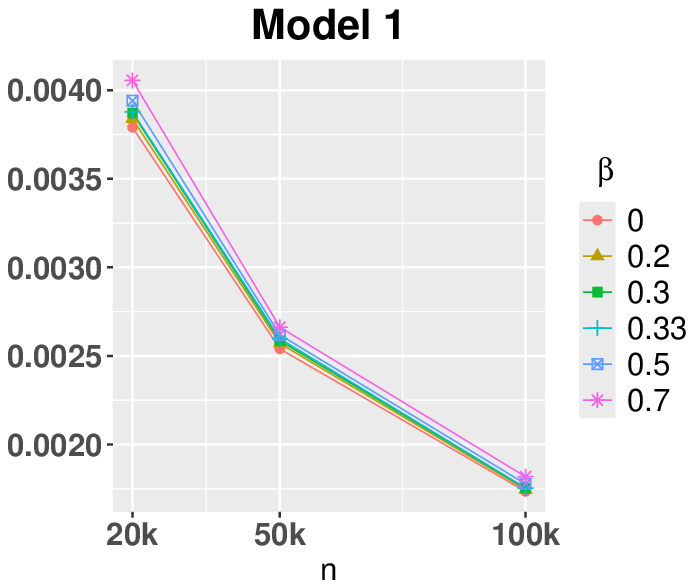}
\includegraphics[width=2.5 in]{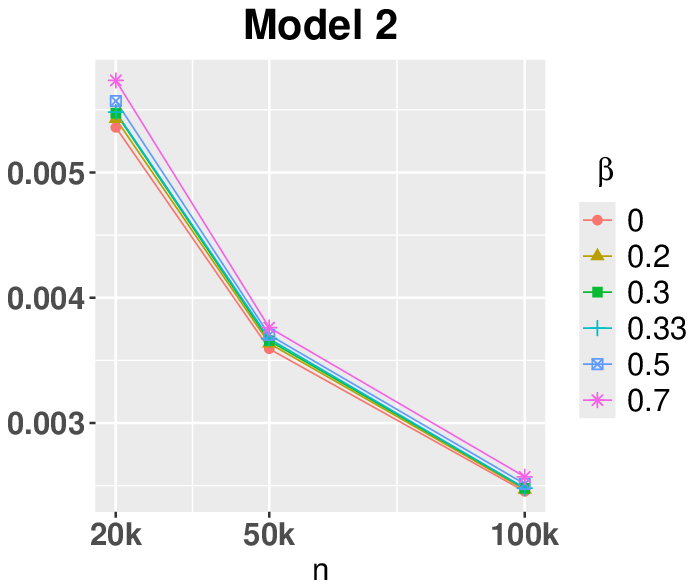}

\includegraphics[width=2.5 in]{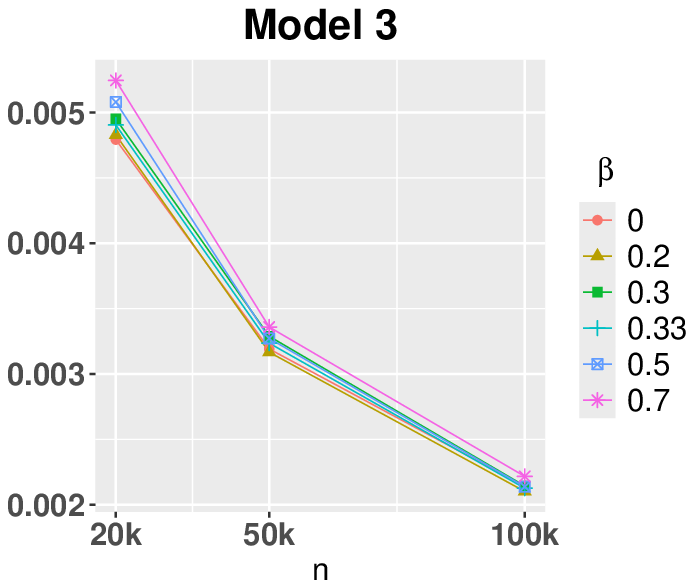}
\includegraphics[width=2.5 in]{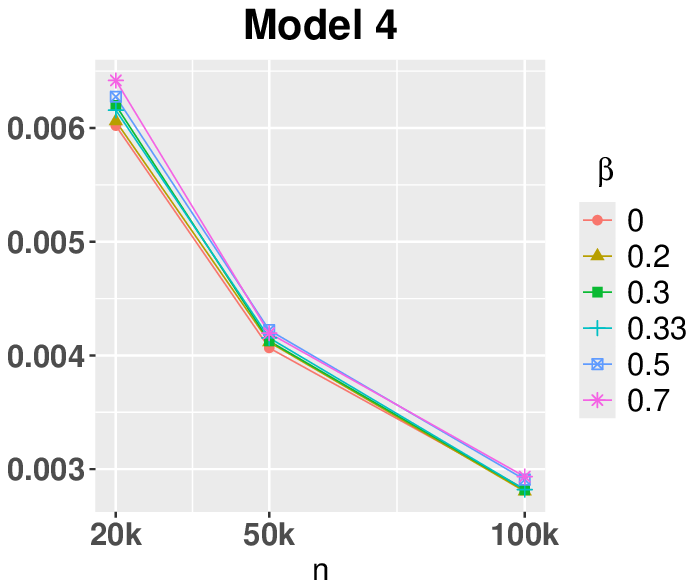}
\caption{RMSE for Models 1-4}
\label{figure:rmse:model1:4}
\end{figure}

\begin{figure}[htp!]
\centering
\includegraphics[width=2.5 in]{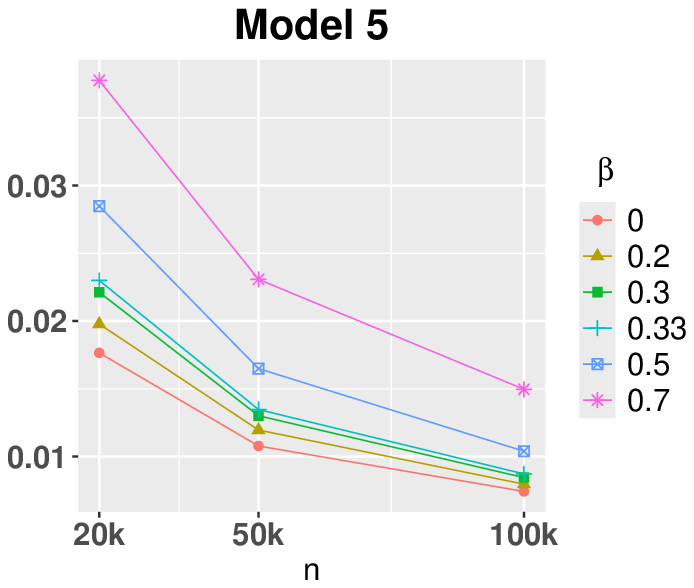}
\includegraphics[width=2.5 in]{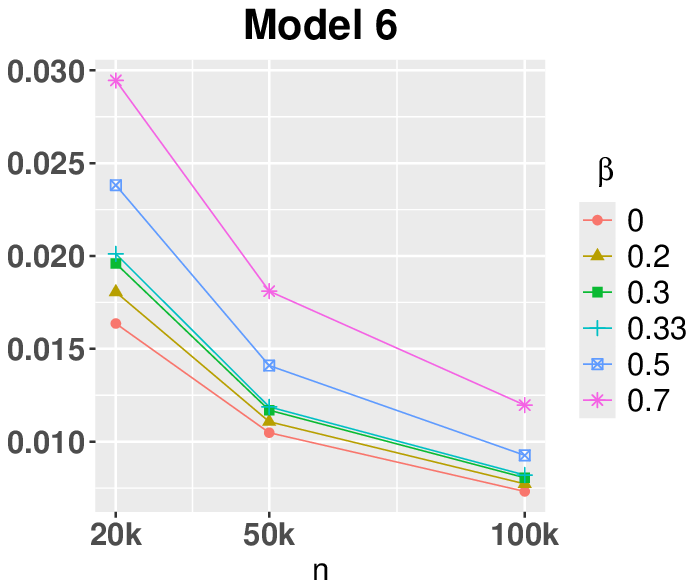}

\includegraphics[width=2.5 in]{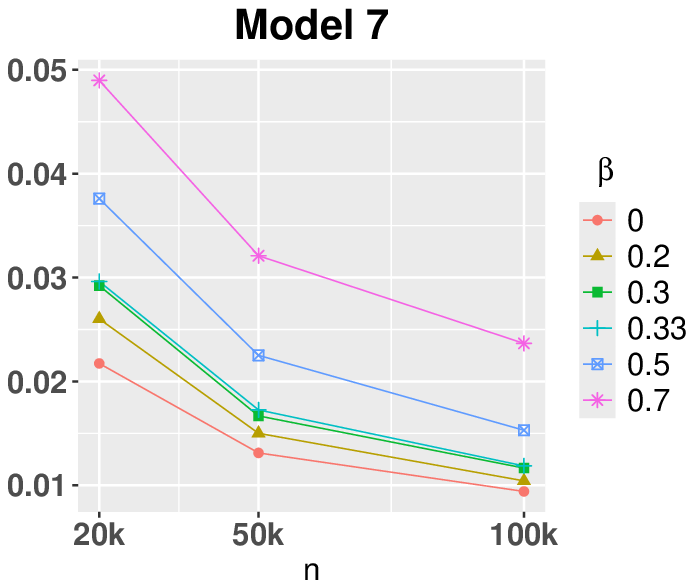}
\includegraphics[width=2.5 in]{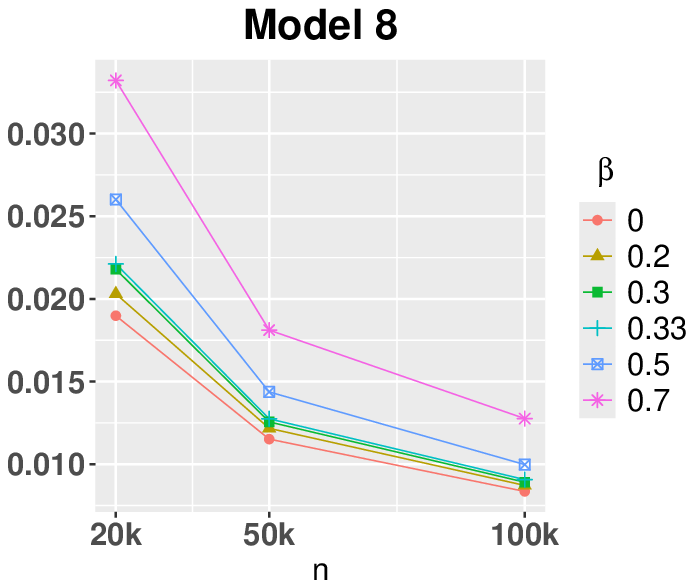}
\caption{RMSE for Models 5-8}
\label{figure:RMSE:model5:8}
\end{figure}

\begin{figure}[htp!]
\centering
\includegraphics[width=2.5 in]{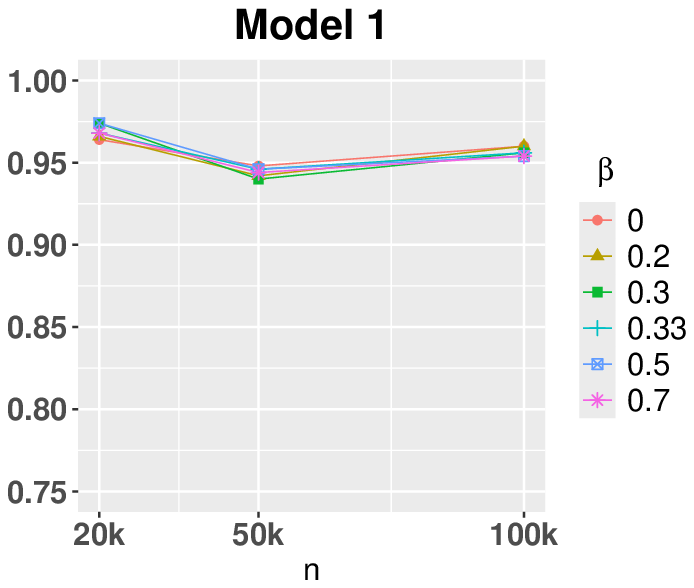}
\includegraphics[width=2.5 in]{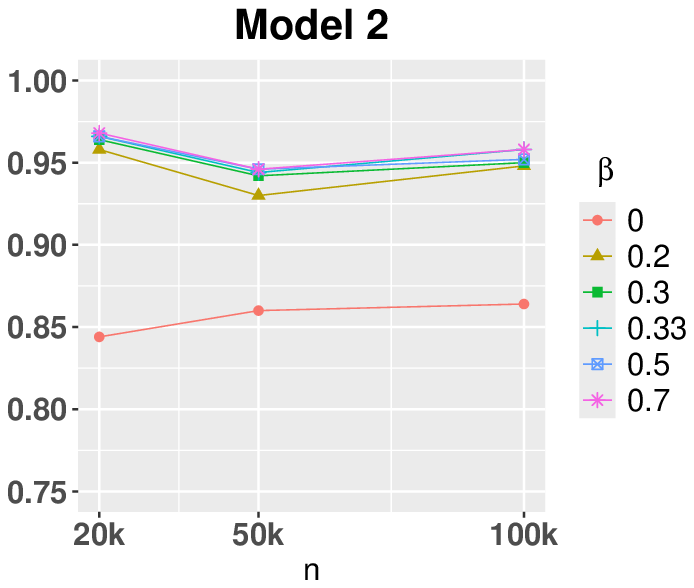}

\includegraphics[width=2.5 in]{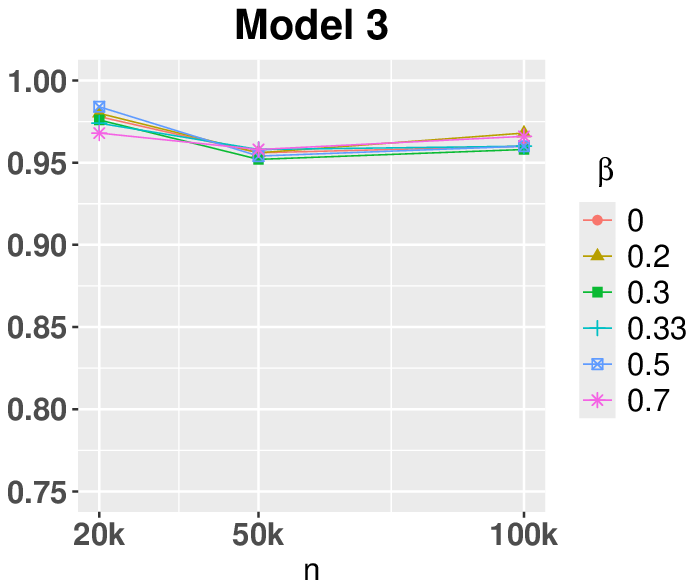}
\includegraphics[width=2.5 in]{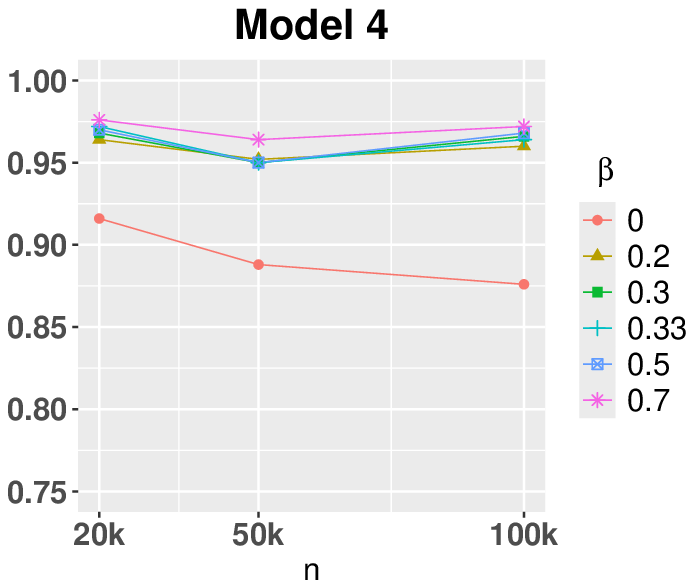}
\caption{CP for Models 1-4}
\label{figure:cp:model1:4}
\end{figure}

\begin{figure}[htp!]
\centering
	\includegraphics[width=1.8 in]{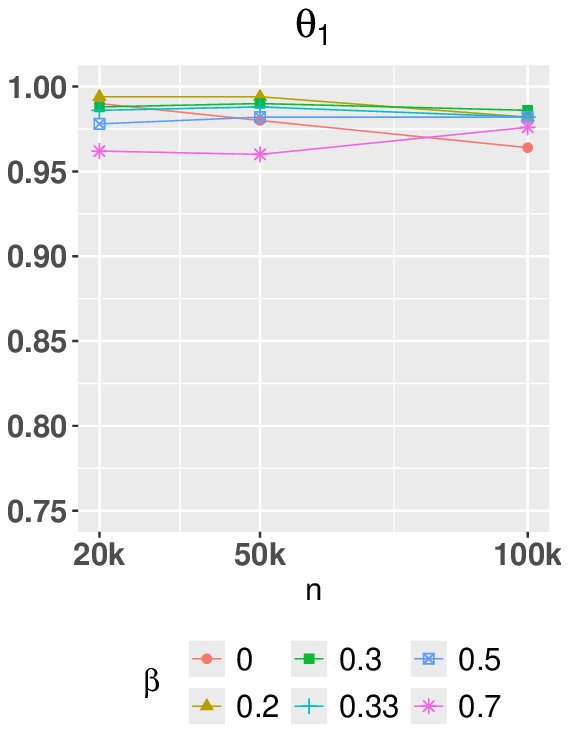}
	\includegraphics[width=1.8 in]{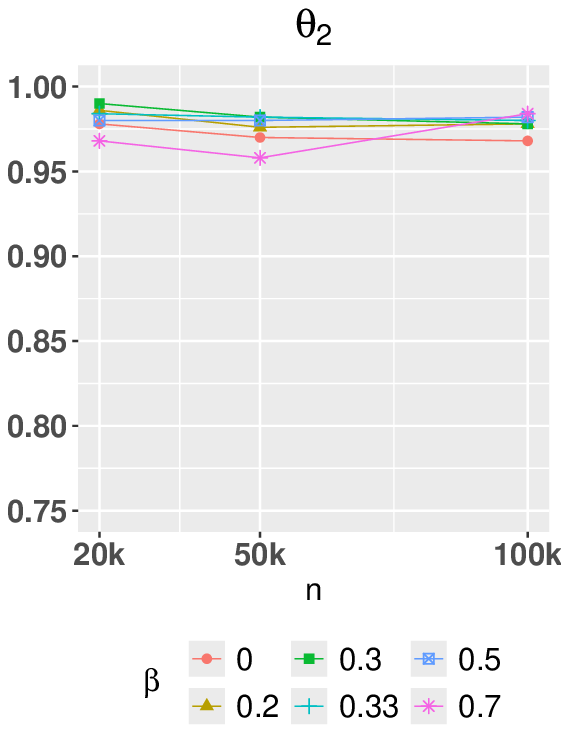}
	\includegraphics[width=1.8 in]{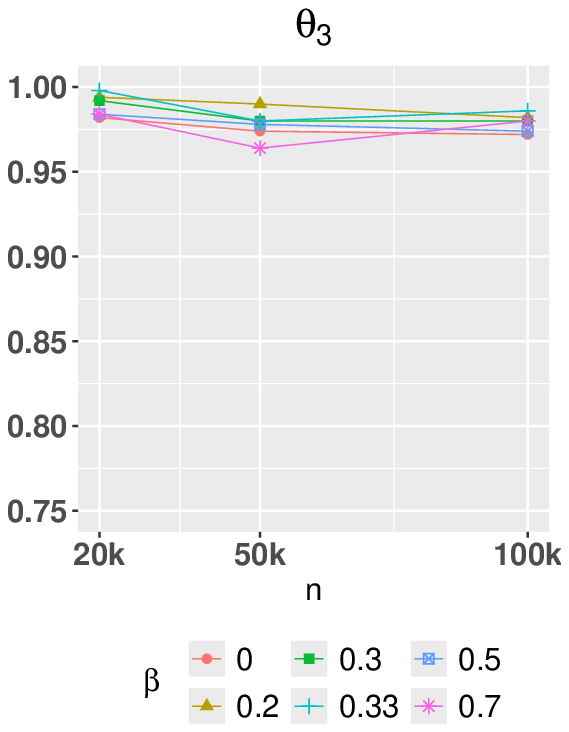}
\caption{CP for Model 5}
\label{figure:cp:model5}
\end{figure}

\begin{figure}[htp!]
\centering
\includegraphics[width=1.8 in]{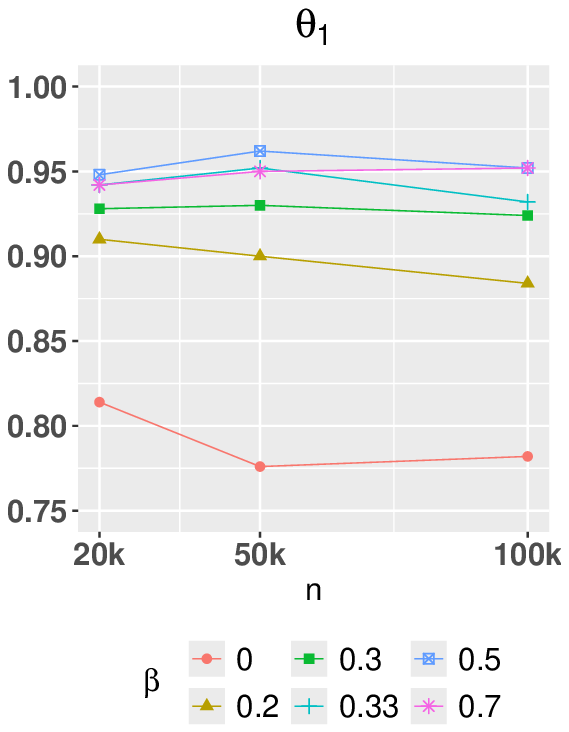}
\includegraphics[width=1.8 in]{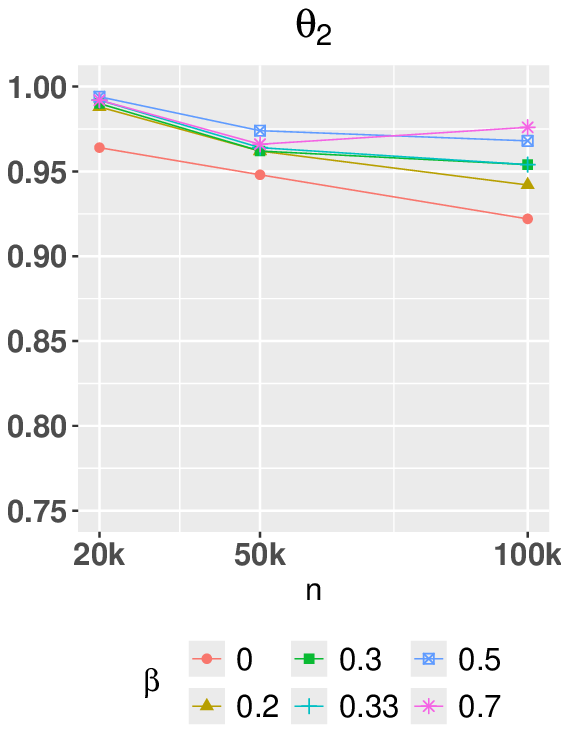}
\includegraphics[width=1.8 in]{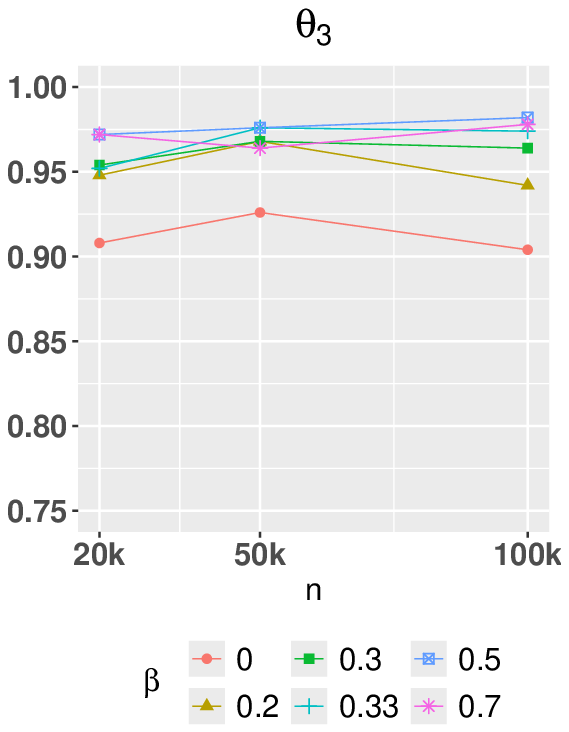}
\caption{CP for Model 6}
\label{figure:cp:model6}
\end{figure}

\begin{figure}[htp!]
\centering
\includegraphics[width=1.8 in]{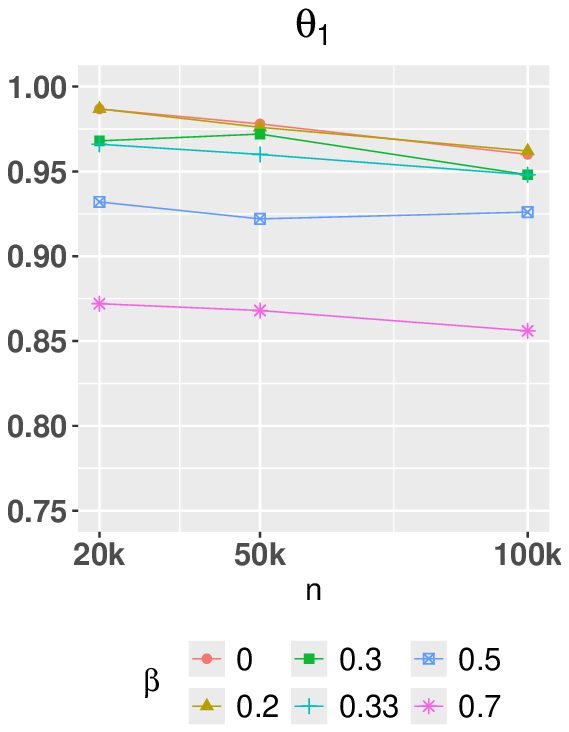}
\includegraphics[width=1.8 in]{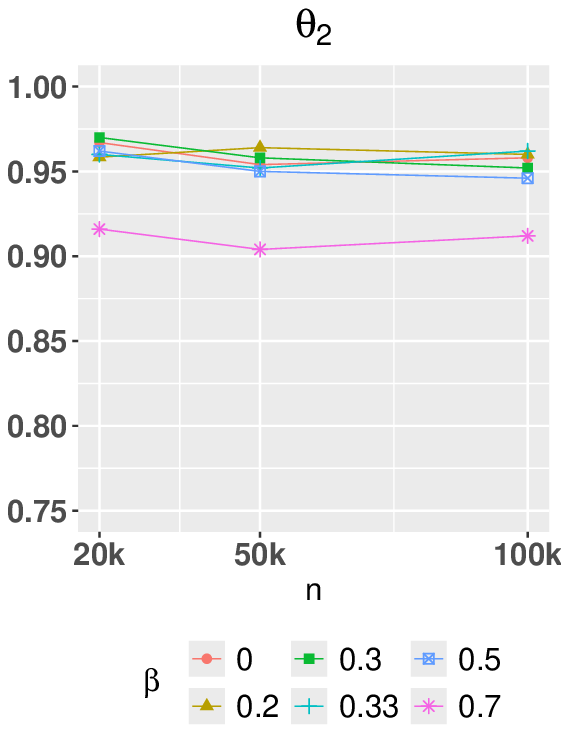}
\includegraphics[width=1.8 in]{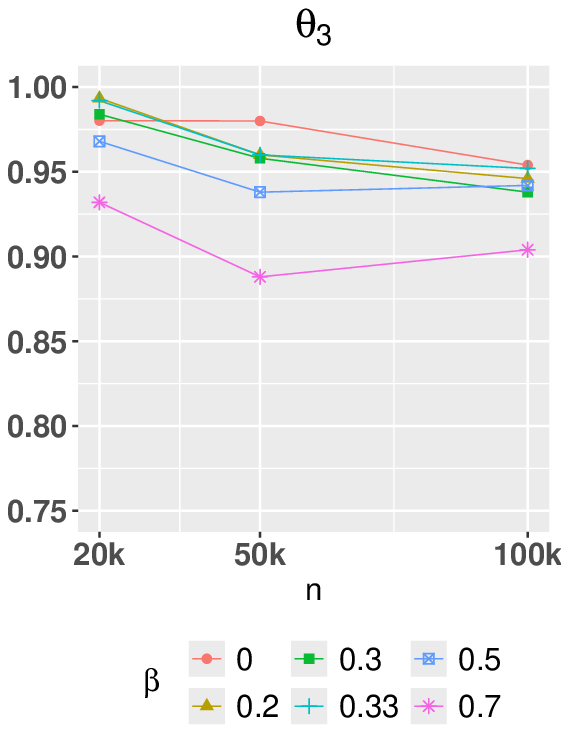}
\caption{CP for Model 7}
\label{figure:cp:model7}
\end{figure}

\begin{figure}[htp!]
\centering
\includegraphics[width=1.8 in]{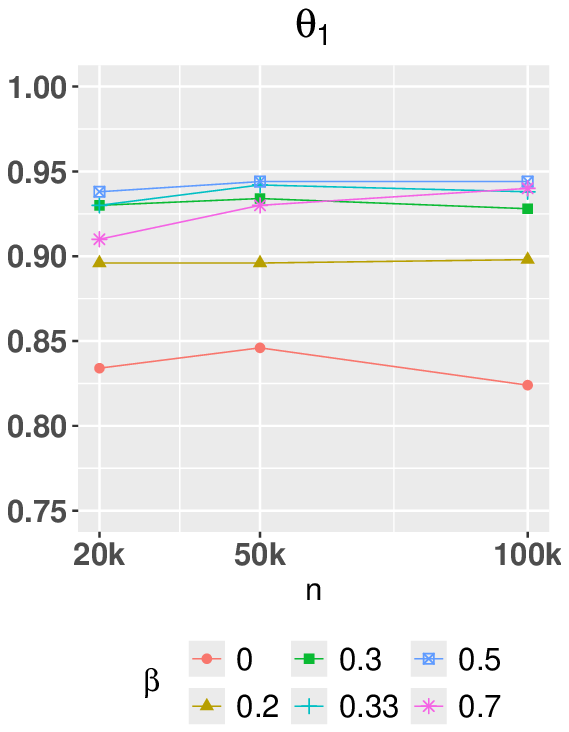}
\includegraphics[width=1.8 in]{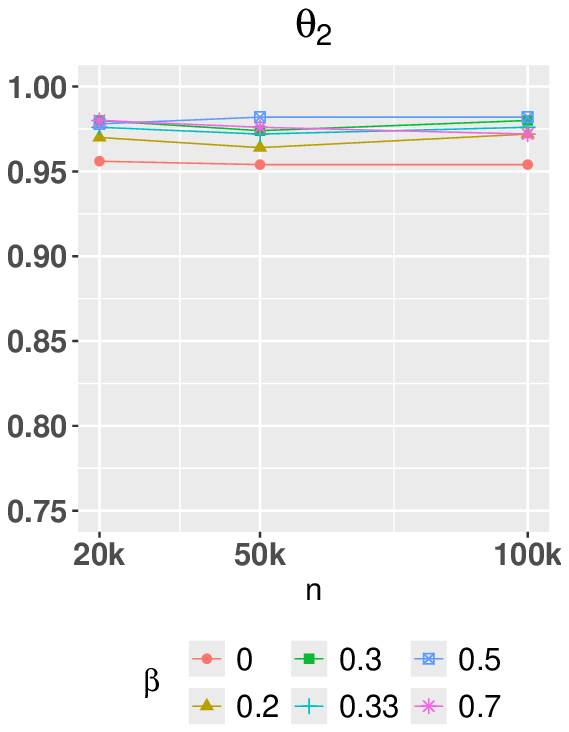}
\includegraphics[width=1.8 in]{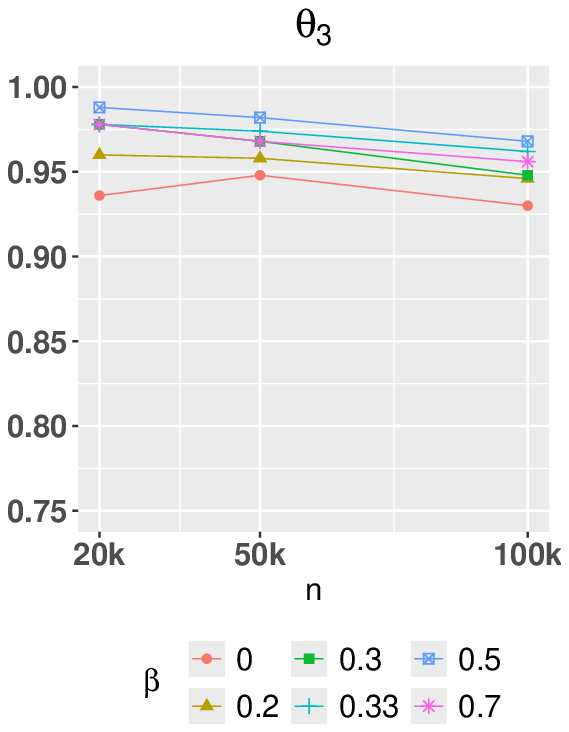}
\caption{CP for Model 8}
\label{figure:cp:model8}
\end{figure}

\section{Empirical Application}\label{sec:application}
In this section, we apply the proposed procedure to the POWER dataset from UCI Machine Learning Repository (\citealp{miscpowerdata2012}). The POWER dataset contains $2075259$ measurements of electric power consumption gathered in a house located in Sceaux during $47$ months. After removing missing values, we keep a record of $2049280$ observations. The time of a day is divided into three periods: (a) morning from 0:00 to 12:00; (b) afternoon from 12:00 to 18:00; (c) evening from  18:00 to 24:00. We  apply the OLS and the LAD regression based on the following four models:
\begin{enumerate}[wide, labelwidth=!, labelindent=0pt]
\item[{Model 1:}]   $\textrm{Power}=\theta_{11}\times \textrm{morning}+\theta_{12}\times\textrm{afternoon}+\theta_{13}\times\textrm{evening}+\textrm{error}$;
\item[{Model 2:}]    $\textrm{Sub-Metering 1}=\theta_{21}\times\textrm{morning}+\theta_{22}\times\textrm{afternoon}+\theta_{23}\times\textrm{evening}+\textrm{error}$;
\item[{Model 3:}]   $\textrm{Sub-Metering 2}=\theta_{31}\times\textrm{morning}+\theta_{32}\times\textrm{afternoon}+\theta_{33}\times\textrm{evening}+\textrm{error}$;
\item[{Model 4:}]   $\textrm{Sub-Metering 3}=\theta_{41}\times\textrm{morning}+\theta_{42}\times\textrm{afternoon}+\theta_{43}\times\textrm{evening}+\textrm{error}.$
\end{enumerate}
Here $\textrm{Power}$ is the household global minute-averaged active power  (in kilowatt), Sub-metering 1 corresponds to the active energy used by kitchen, Sub-metering 2 is the energy consumption in  the laundry room,  Sub-Metering 3 records the energy consumed by electric water-heaters and air-conditioners, and the covariates are the corresponding dummies. The observations of Sub-metering 1-3 are recorded in watt-hour. The batch size is set as $B_t=\floor{t^{0.33}}\vee 1$, while the choice of the initial values and the bootstrap setting follow those in the simulation studies. 

The point estimates and their corresponding $95\%$ confidence intervals are presented in Tables \ref{table:ols:power} and \ref{table:lad:power}. These results indicate that the error terms are right-skewed, as evidenced by most OLS estimates being larger than LAD estimates. Moreover, both tables underscore that electronic power consumption in the afternoon and evening typically surpasses that in the morning. Notably, Models 2 and 3 showcase estimates close to zero in Table \ref{table:lad:power}, which contrasts significantly with the estimates in Table \ref{table:ols:power}. This discrepancy can be attributed to the fact that approximately $91.75\%$ and $70.11\%$ of Sub-Metering 1 and Sub-Metering 2 observations, respectively, are zeros, while Sub-Metering 3 contains $41.58\%$ zeros. Furthermore, histograms (density) depicting $N=500$ bootstrap samples for each coefficient estimate are summarized in Figures \ref{figure:ols:bootstrap} and \ref{figure:lad:bootstrap}. These histograms exhibit bell-shaped distributions, demonstrating the effectiveness of the bootstrap method in constructing reliable confidence intervals.

%
%
%
%

%

\begin{table}[htp!]
  \centering
  \caption{OLS estimation of Power dataset}
    \label{table:ols:power}%
  \resizebox{\textwidth}{!}{  
    \begin{tabular}{lccrccrccrcc}
    \hline\hline
          & \multicolumn{2}{c}{Model 1} &       & \multicolumn{2}{c}{Model 2} &       & \multicolumn{2}{c}{Model 3} &       & \multicolumn{2}{c}{Model 4} \bigstrut\\
\cline{2-3}\cline{5-6}\cline{8-9}\cline{11-12}          & Est.  & 95\% CI &       & Est.  & 95\% CI &       & Est.  & 95\% CI &       & Est.  & 95\% CI \bigstrut\\
    \hline
    Morning & 0.8908 & (0.8806, 0.9006) &       & 0.5711 & (0.5328, 0.6013) &       & 0.7410 & (0.7054, 0.7790) &       & 6.2113 & (6.1137, 6.3084) \bigstrut[t]\\
    Afternoon & 1.0853 & (1.0688, 1.1003) &       & 1.2681 & (1.1912, 1.3437) &       & 2.2221 & (2.1011, 2.3329) &       & 6.8987 & (6.7576, 7.0491) \\
    Evening & 1.5602 & (1.5393, 1.5827) &       & 2.1668 & (2.0672, 2.2685) &       & 1.6886 & (1.5968, 1.7894) &       & 6.4014 & (6.2701, 6.5515) \bigstrut[b]\\
    \hline
    \end{tabular}%
    }
\end{table}%

\begin{table}[htp!]
  \centering
  \caption{LAD estimation of Power dataset}
    \label{table:lad:power}%
      \resizebox{\textwidth}{!}{  
    \begin{tabular}{lccrccrccrcc}
    \hline\hline
          & \multicolumn{2}{c}{Model 1} &       & \multicolumn{2}{c}{Model 2} &       & \multicolumn{2}{c}{Model 3} &       & \multicolumn{2}{c}{Model 4} \bigstrut\\
\cline{2-3}\cline{5-6}\cline{8-9}\cline{11-12}          & Est.  & 95\% CI &       & Est.  & 95\% CI &       & Est.  & 95\% CI &       & Est.  & 95\% CI \bigstrut\\
    \hline
    Morning & 0.4547 & (0.4451, 0.4691) &       & 0.0003 & (0.0003, 0.0005) &       & 0.0029 & (0.0033, 0.0045) &       & 0.5831 & (0.5738, 0.5960) \bigstrut[t]\\
    Afternoon & 0.7156 & (0.6919, 0.7413) &       & 0.0002 & (0.0002, 0.0007) &       & 0.0031 & (0.0037, 0.0054) &       & 2.8416 & (1.8638, 4.2264) \\
    Evening & 1.3459 & (1.3160, 1.3760) &       & 0.0006 & (0.0007, 0.0011) &       & 0.0040 & (0.0042, 0.0071) &       & 7.6681 & (6.4352, 9.0148) \bigstrut[b]\\
    \hline
    \end{tabular}%
    }
\end{table}%

\begin{figure}[htp!]
\centering
\includegraphics[width=2 in]{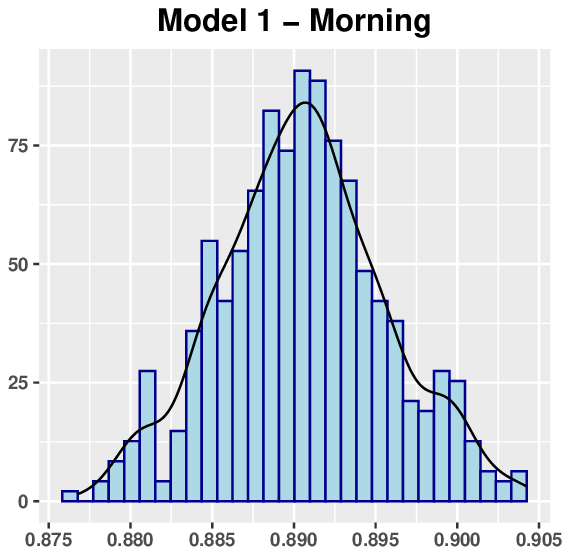}
\includegraphics[width=2 in]{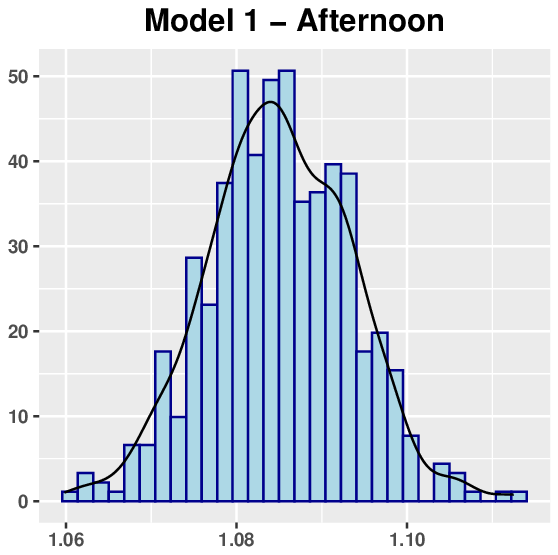}
\includegraphics[width=2 in]{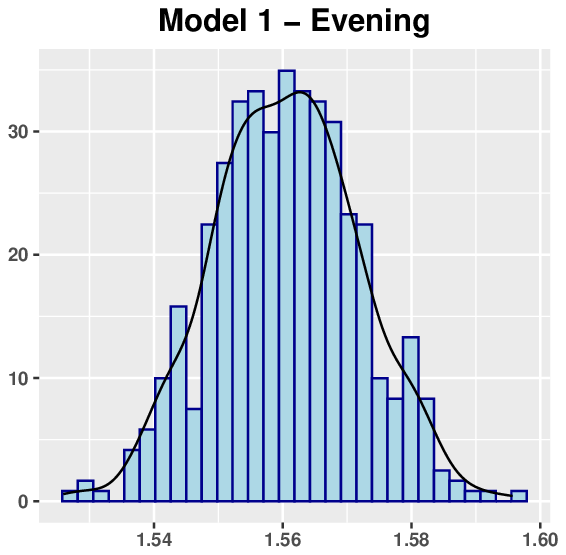}

\includegraphics[width=2 in]{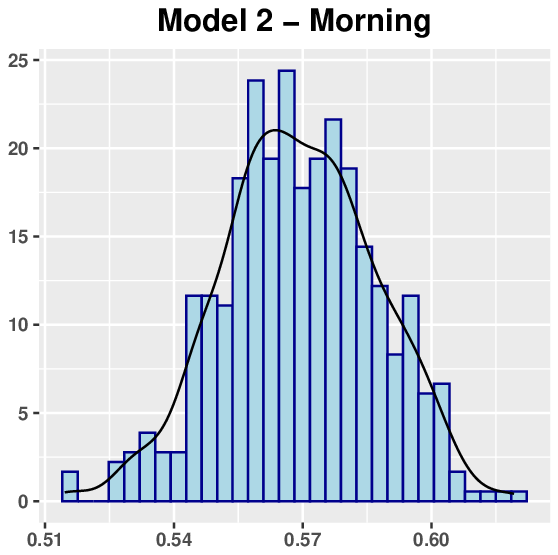}
\includegraphics[width=2 in]{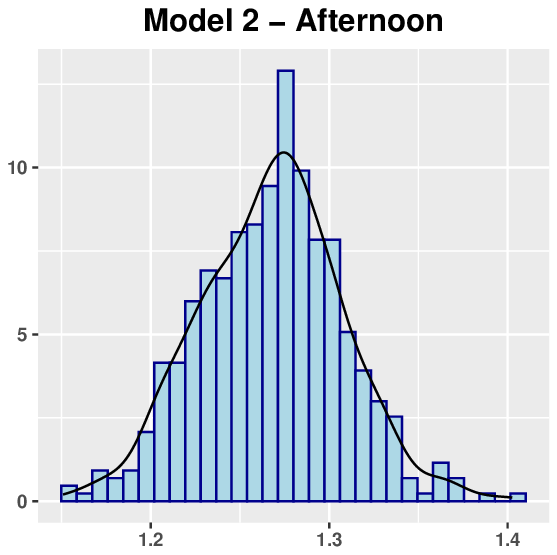}
\includegraphics[width=2 in]{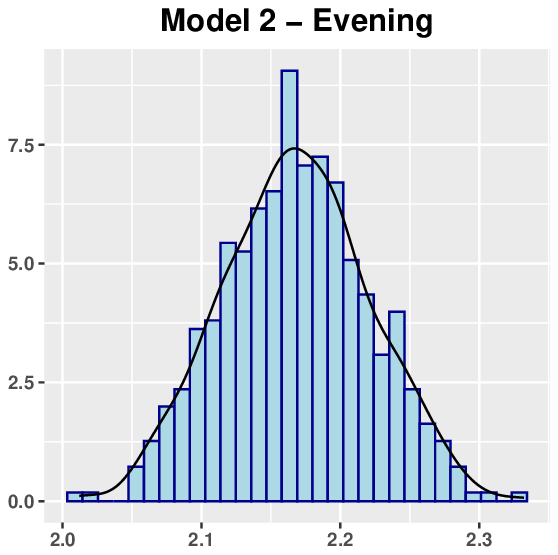}

\includegraphics[width=2 in]{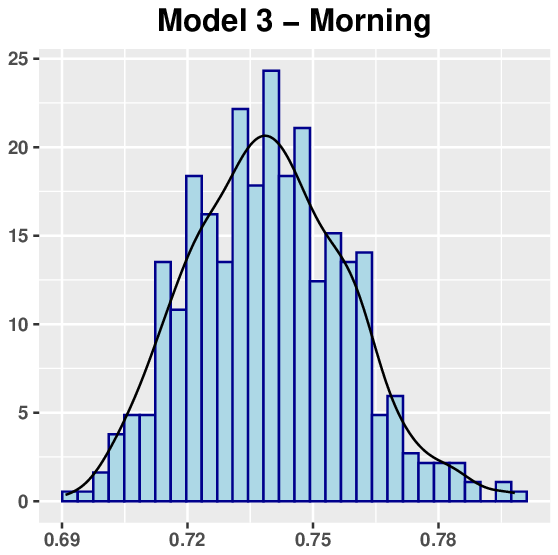}
\includegraphics[width=2 in]{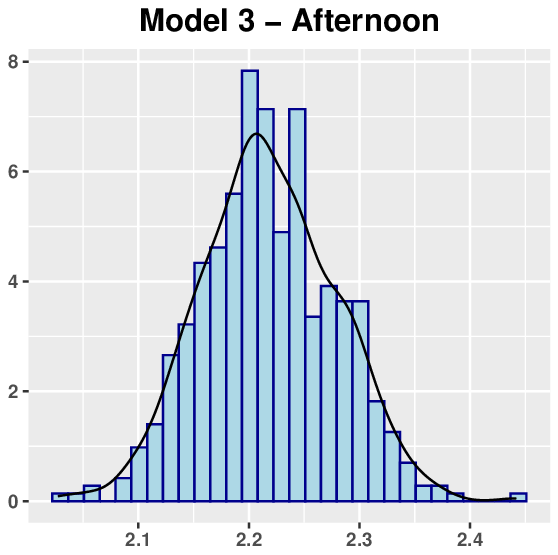}
\includegraphics[width=2 in]{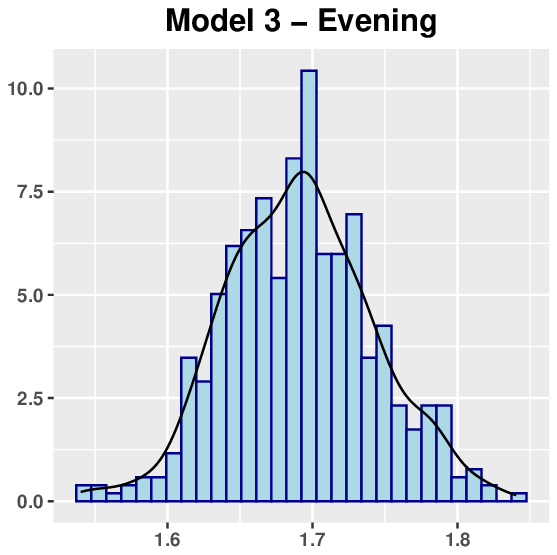}

\includegraphics[width=2 in]{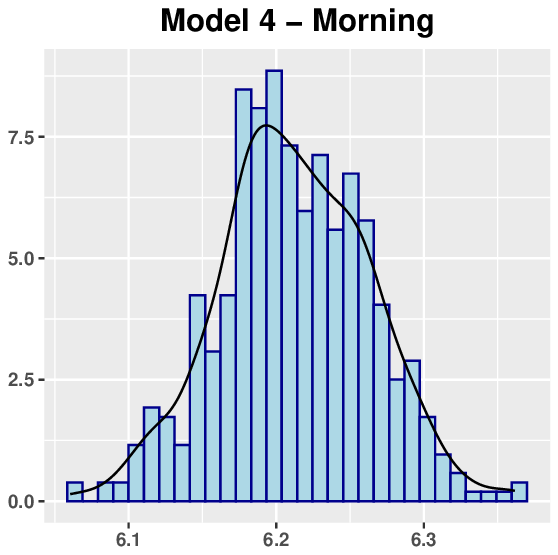}
\includegraphics[width=2 in]{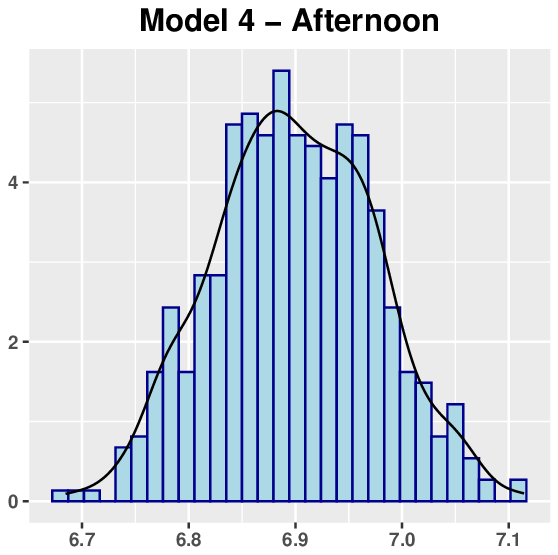}
\includegraphics[width=2 in]{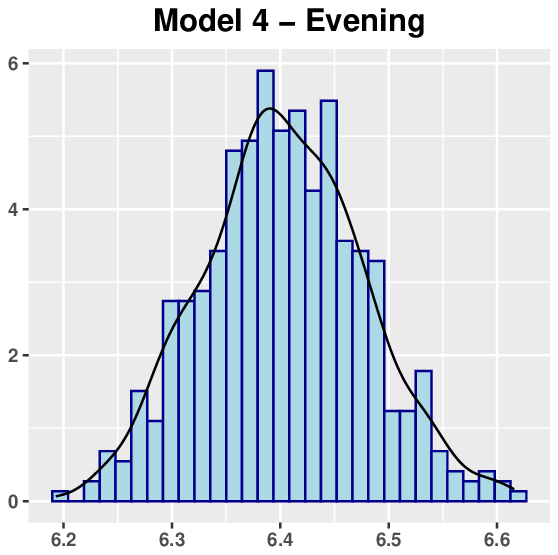}
\caption{Bootstrap histograms for OLS}
\label{figure:ols:bootstrap} 
\end{figure}

\begin{figure}[htp!]
\centering
\includegraphics[width=2 in]{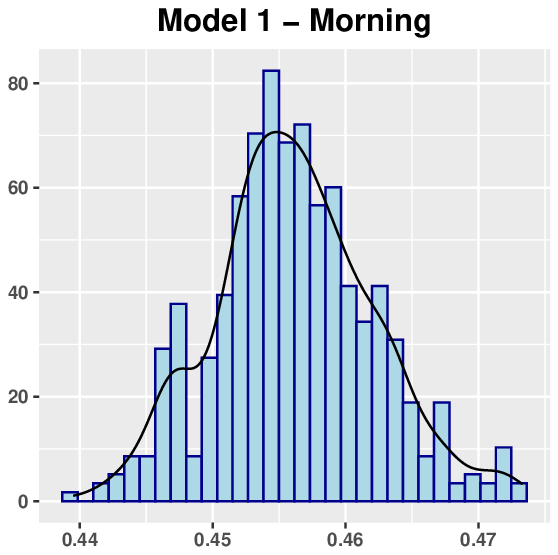}
\includegraphics[width=2 in]{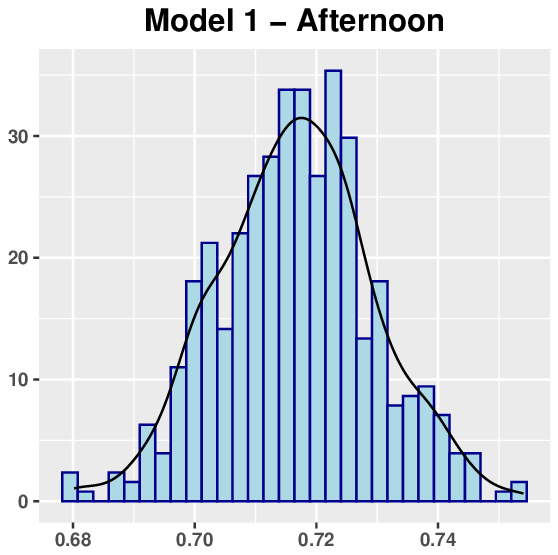}
\includegraphics[width=2 in]{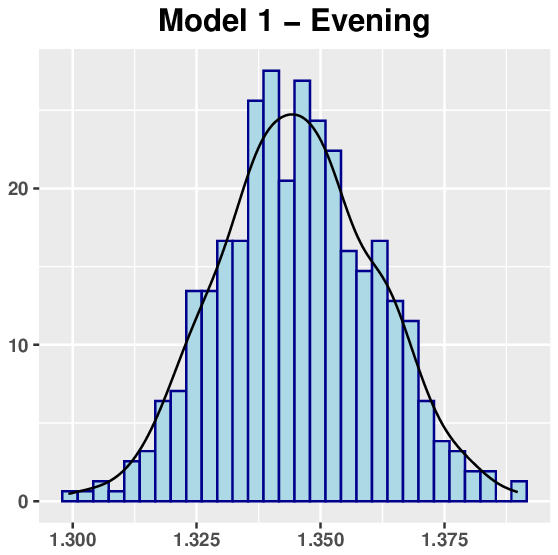}

\includegraphics[width=2 in]{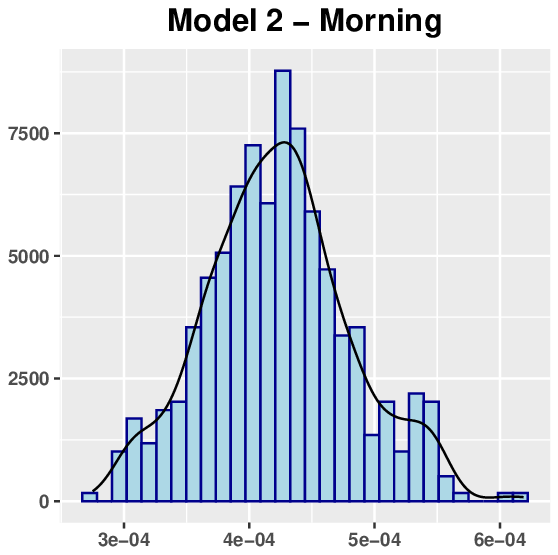}
\includegraphics[width=2 in]{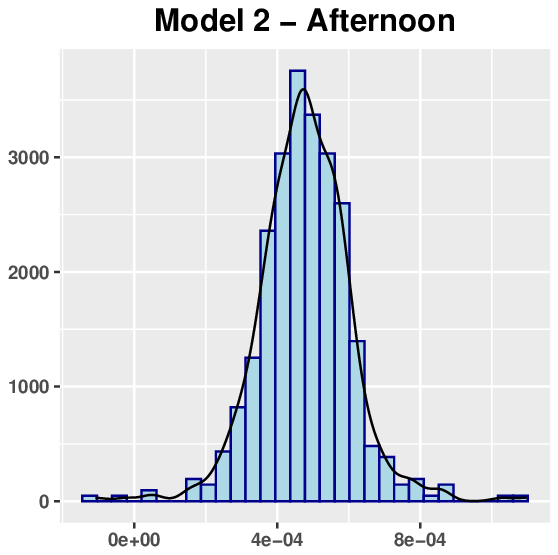}
\includegraphics[width=2 in]{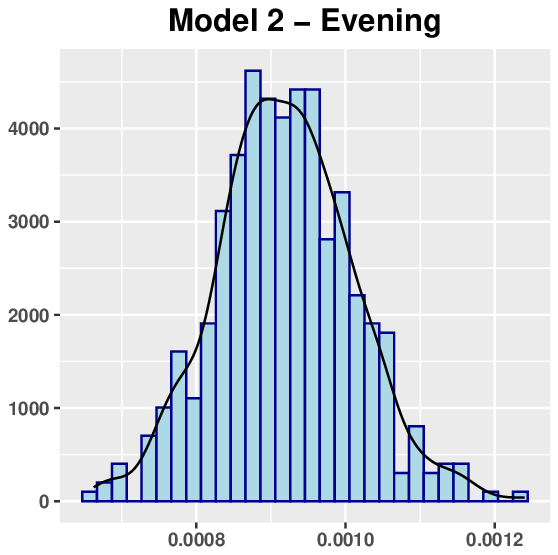}

\includegraphics[width=2 in]{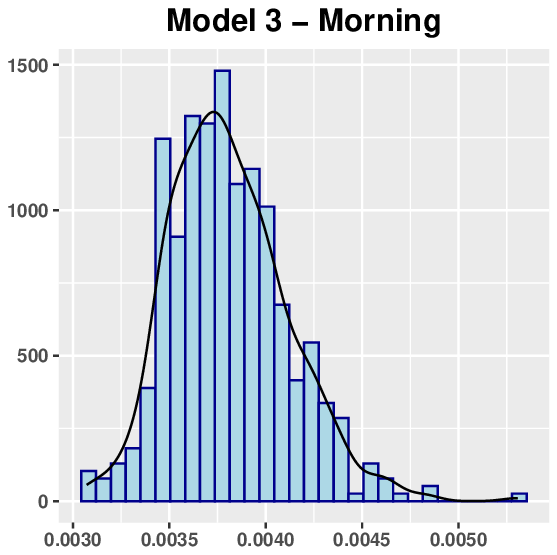}
\includegraphics[width=2 in]{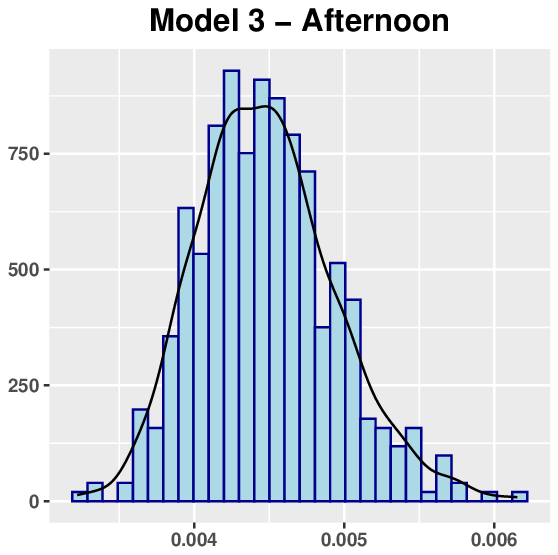}
\includegraphics[width=2 in]{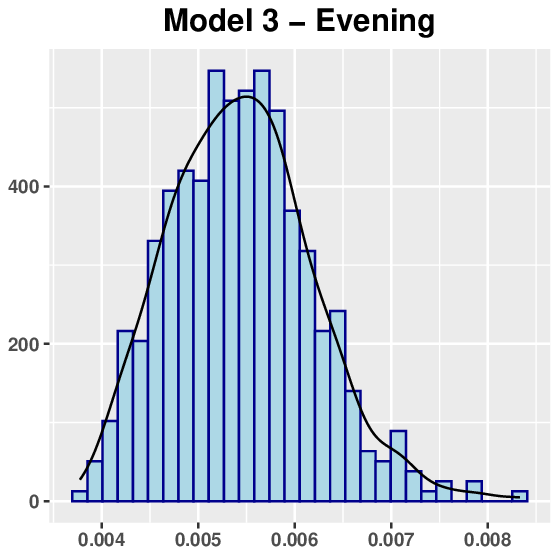}

\includegraphics[width=2 in]{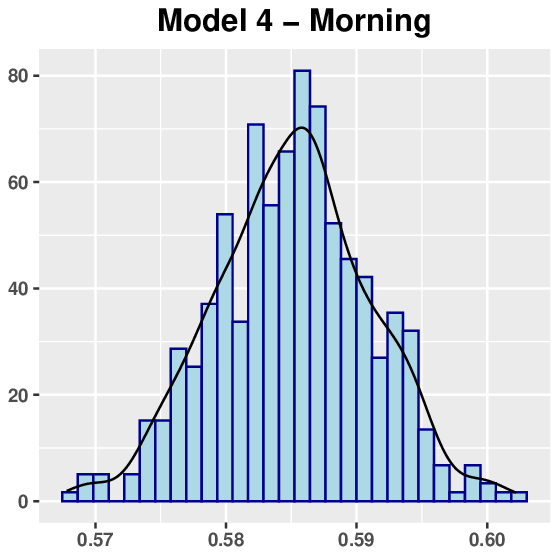}
\includegraphics[width=2 in]{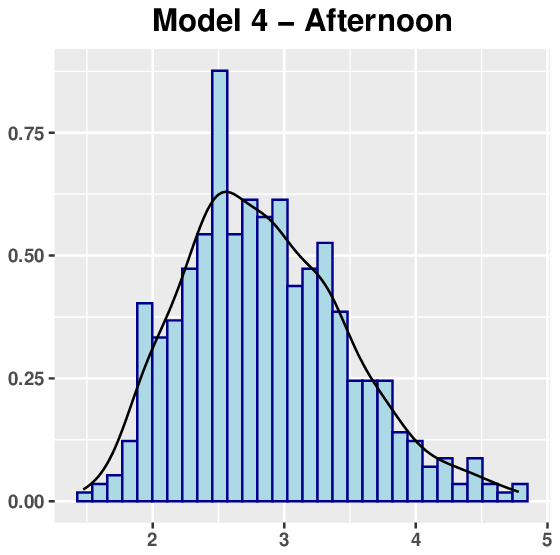}
\includegraphics[width=2 in]{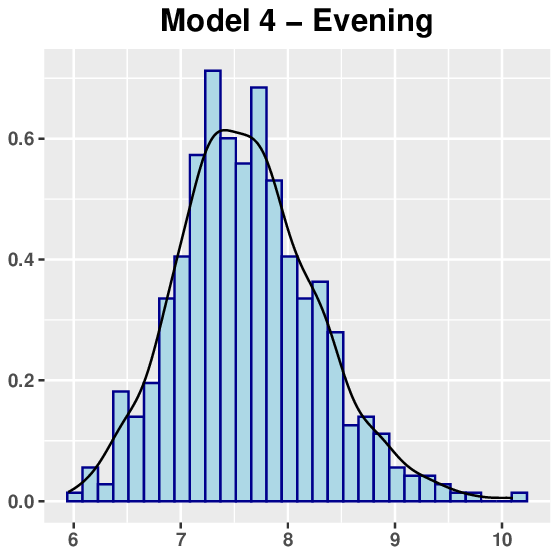}
\caption{Bootstrap histograms for LAD}
\label{figure:lad:bootstrap} 
\end{figure}

\section{Conclusion and Future Works}
In this paper, we proposed a mini-batch SGD algorithm to conduct statistical inference about the system parameters when the observations are $\phi$-mixing. The proposed algorithm is applicable to both smooth and non-smooth loss functions, which can cover many popular statistical models. In addition, a mini-batch bootstrap SGD procedure is developed to construct confidence intervals. The limiting distribution of the mini-batch SGD estimator and the validity of the bootstrap procedure are both established. To demonstrate the risk of ignoring the correlation, we also design a concrete example showing the invalidity of the bootstrap SGD in \cite{fang2018online} when the observations are $\phi$-mixing. Monte Carlo simulations and an application to a real-world dataset are conducted to examine the finite-sample properties of the proposed method, and the results confirm our theoretical findings. 

There are several interesting next directions rooted on our approach. First, a natural step is to extend our framework from $\phi$-mixing to more general dependence like strong mixing ($\alpha$-mixing) or near‐epoch dependence (NED). Furthermore, it would be valuable to develop a functional central limit theorem of our estimator, which enables interval estimation using the random scaling method in \cite{lee2022fast, lee2022fastb}. Lastly, extending the results in Theorems \ref{theorem:consistency} and \ref{theorem:asymptotic:normality} from mini-batch SGD to the case of standard SGD ($B_t=1$) under $\phi$-mixing data would contribute significantly to the field.


\section*{\textcolor{white}{Supplement}}
\renewcommand
\theequation{S.\arabic{equation}}\setcounter{equation}{0} %
\setcounter{section}{0} \renewcommand
\thelemma{S.\arabic{lemma}} \renewcommand
\thesection{S.\arabic{section}}\renewcommand
\thesubsection{S.\arabic{subsection}}
\setcounter{subsection}{0}
\renewcommand{\thesubsection}{S.\Roman{subsection}}
\setcounter{equation}{0}
\renewcommand{\theequation}{S.\arabic{equation}}
\setcounter{lemma}{0}
\renewcommand{\thelemma}{S.\arabic{lemma}}
\setcounter{theorem}{0}
\renewcommand{\thetheorem}{S.\arabic{theorem}}
\setcounter{proposition}{0}
\renewcommand{\theproposition}{S.\arabic{proposition}}
\begin{center}
{\Large \textbf{Supplement to ``Statistical Inference with Stochastic Gradient Methods under $\phi$-mixing Data"}}\\[0.25in]
\if1\blind
{
Ruiqi Liu, Xi Chen and Zuofeng Shang\\[0.25in]
} \fi

%



\end{center}
This supplement includes the proofs of the main theorems. Section \ref{section:prelinmiary:lemma} provides the proofs of some technical lemmas. The proofs of Theorems \ref{theorem:consistency}-\ref{theorem:bootstrap:asymptotic} are given in Sections \ref{section:proof:consistency}-\ref{section:proof:bootstrap}, respectively. We prove Proposition \ref{prop:failure:bootstrap:sgd} in Section \ref{section:proof:fail:sgd}. Additional theoretical and numerical results are included in  Sections \ref{sec:additional:theoretical:result} and \ref{sec:additional:simulation}.

Since $\theta_T^a$ (or $\theta_T^{*a}$) plays a similar role as $\theta_T^b$ (or $\theta_T^{*b}$), it suffices to study the former. For ease of presentation, we drop the superscript in (\ref{eq:iteration:estimator}) and (\ref{eq:iteration:estimator:bootstrap}). Moreover, iteration procedures in (\ref{eq:iteration:estimator}) and (\ref{eq:iteration:estimator:bootstrap}) can be generalized as follows:
\begin{eqnarray}
\theta_t=\Pi\left\{\theta_{t-1}-\gamma_t U_t\widehat{H}_t(W_t,  \theta_{t-1})\right\}.\label{eq:iteration:general}
\end{eqnarray}
where $U_t$'s are i.i.d random variables satisfying Assumption \ref{Assumption:A4}, and $W_t=W_t^a$ after dropping the superscript.

\begin{AssumptionB}\label{Assumption:A4}
The i.i.d. random variables $U_j$'s  are independent from the observations $Z_j$'s. Moreover, it holds that $\ev(U_j)=1$ and $\ev(|U_j|^p)<\infty$. Here $p$ is the constant introduced in Assumption \ref{Assumption:A1}\ref{A1:moment:condition}.
\end{AssumptionB}

Clearly, $\theta_t$ becomes the estimators in (\ref{eq:iteration:estimator}) when $U_t=1$, and it is identical to bootstrap estimators in (\ref{eq:iteration:estimator:bootstrap}) when $U_t=V_t$. Before proceeding, let us define an ancillary time series $\widetilde{\bfZ}=\{\widetilde{Z}_t\}_{t=1}^\infty$, which  has the same distribution as $\bfZ=\{Z_t\}_{t=1}^\infty$ and is independent from $\bfZ$ and $U_t$'s. Similarly, we define $\widetilde{W}_t=\{\widetilde{Z}_i, i\in I_t\}$, which is identically distributed as $W_t$. Moreover, we will use the following notation: 
\begin{eqnarray*}
\ev_{t-1}(\cdot)&=&\ev(\cdot|W_1, \ldots, W_{t-1}),\\
H(\theta)&=&\ev\{\nabla l(Z,\theta)\}=\nabla L(\theta),\\
e_t&=&\ev_{t-1}\{\widehat{H}_t(W_t, \theta_{t-1})\}-H(\theta_{t-1}),\\
\zeta_t&=&U_t\widehat{H}_t({W}_t,  \theta_{t-1})-\ev_{t-1}\{\widehat{H}_t(W_t, \theta_{t-1})\},\\
\widetilde{\zeta}_t&=&U_t\widehat{H}_t(\widetilde{W}_t,  \theta_{t-1})-\ev_{t-1}\{\widehat{H}_t(\widetilde{W}_t, \theta_{t-1})\}=U_t\widehat{H}_t(\widetilde{W}_t,  \theta_{t-1})-H(\theta_{t-1}),\\
\Delta_t&=&\theta_t-\theta^*.
\end{eqnarray*}
Hence, the iteration in (\ref{eq:iteration:general}) can be written as
\begin{eqnarray*}
\theta_t=\Pi\left\{\theta_{t-1}-\gamma_t  H(\theta_{t-1})-\gamma_t  e_t-\gamma_t \zeta_t\right\}.
\end{eqnarray*}


\subsection{Preliminary Lemmas}\label{section:prelinmiary:lemma}

\begin{lemma}\label{lemma:comparison:test}
Suppose that $\{a_n\}_{n=1}^\infty$ is a positive and non-increasing sequence with $\sum_{n=1}^\infty a_n<\infty$. Then it holds that $\lim_{n\to \infty}na_n=0$.
\end{lemma}
\begin{proof}
Since $a_n$ is non-increasing, we have
\begin{eqnarray*}
2na_{2n}=2\sum_{k=n+1}^{2n}a_{2n}\leq 2\sum_{k=n+1}^{2n}a_{k}\leq 2\sum_{k=n+1}^{\infty}a_{k}.
\end{eqnarray*}
Noting that $\sum_{n=1}^\infty a_n<\infty$, all the terms in the preceding display will converge to zero as $n\to \infty$.
\end{proof}

\begin{lemma}\label{lemma:weighted:cesaro:sum}
Suppose that $\{a_n\}_{n=1}^\infty$ is a positive sequence with $\sum_{n=1}^\infty a_n=\infty$, and $\{b_n\}_{n=1}^\infty$ is a sequence with $\lim_{n\to \infty}b_n=b$. Then it holds that
\begin{eqnarray*}
\lim_{T\to \infty} \frac{\sum_{n=1}^T a_nb_n}{\sum_{n=1}^T a_n}=b.
\end{eqnarray*}
\end{lemma}
\begin{proof}
For any $\epsilon>0$, there exists a constant $N>0$ such that $|b_n-b|<\epsilon$ for all $n\geq N$. As a consequence, it follows that
\begin{eqnarray*}
 \bigg|\frac{\sum_{n=1}^T a_nb_n}{\sum_{n=1}^T a_n}-b\bigg|&=& \bigg|\frac{\sum_{n=1}^T a_n(b_n-b)}{\sum_{n=1}^T a_n}\bigg|\\
 &\leq&\bigg|\frac{\sum_{n=1}^{N} a_n(b_n-b)}{\sum_{n=1}^T a_n}\bigg|+\bigg|\frac{\sum_{n=N+1}^{T} a_n(b_n-b)}{\sum_{n=1}^T a_n}\bigg|\\
 &\leq& \bigg|\frac{\sum_{n=1}^{N} a_n(b_n-b)}{\sum_{n=1}^T a_n}\bigg|+\epsilon.
\end{eqnarray*}
Since $\sum_{n=1}^\infty a_n=\infty$, we conclude that $\lim_{T\to \infty} \left|{\sum_{n=1}^T a_nb_n}/{\sum_{n=1}^T a_n}-b\right|\leq \epsilon.$ Noting that $\epsilon>0$ is arbitrary, we complete the proof.
\end{proof}
 
\begin{lemma}\label{lemma:weighted:cesaro:sum:2}
Suppose that $\{a_n\}_{n=1}^\infty$ is a positive sequence with $\sum_{n=1}^\infty a_n=\infty$, and $\{b_n\}_{n=1}^\infty$  is a  sequence with $\sum_{n=1}^\infty |b_n|<\infty$. Then it holds that
\begin{eqnarray*}
\lim_{T\to \infty} \frac{\sum_{n=1}^T a_n \sum_{k=n}^T b_k}{\sum_{n=1}^T a_n}=0.
\end{eqnarray*}
\end{lemma}
\begin{proof}
Notice that $|\sum_{k=n}^T b_k|\leq \sum_{k=n}^\infty| b_k|\to 0$ as $n\to \infty$.  Applying Lemma \ref{lemma:weighted:cesaro:sum}, we obtain the desired result.
\end{proof}

\begin{lemma}\label{lemma:weighted:cesaro:sum:3}
Suppose that $\{a_n\}_{n=1}^\infty$ is a positive sequence with $\sum_{n=1}^\infty a_n=\infty$, and $\{b_n\}_{n=1}^\infty$  is a  sequence with $\lim_{n\to \infty}b_n/a_n=0$. Then it holds that
\begin{eqnarray*}
\lim_{T\to \infty} \frac{\sum_{n=1}^T b_n}{\sum_{n=1}^T a_n}=0.
\end{eqnarray*}
\end{lemma}
\begin{proof}
Noting that $|b_n|=a_n(|b_n|/a_n)$, applying Lemma \ref{lemma:weighted:cesaro:sum} completes the proof.
\end{proof}

\begin{lemma}\label{lemma:moment:inequality:phi:mixing}
Let $X_1, X_2, \ldots$ be a stationary sequence with the $\phi$-mixing coefficients bounded by $\phi(t)$. Moreover, assume that $\sum_{t=1}^\infty \phi^{1/2}(t)<\infty$,  $\ev(X_t)=0$, and $\ev(\|X_t\|^k)<\infty$ for some constant $k>2$. Then it follows that $\ev(\|\sum_{t=1}^T X_t\|^k)\leq C_k T^{k/2}$, where $C_k$ is some constant relying on $k$ and the dimension of $X_i$'s.
\end{lemma}
\begin{proof}
It follows from  Theorem 3 in \cite{yokoyama1980moment}.
\end{proof}

\begin{lemma}\label{lemma:maxmum:is:ratio}
Let $X$ be a real random variable on a probability space $(\Omega, \mcA, P)$ with $\|X\|_\infty <\infty$. Then
\begin{eqnarray*}
\sup_{A\in \mcA, P(A)>0}\frac{1}{P(A)}\left|\int_A XdP\right|=\|X\|_\infty.
\end{eqnarray*}
\end{lemma}
\begin{proof}
If $M\geq 0$ is a constant such that $\|X\|_\infty\leq M$, and $A\in \mcA$ is such that $P(A)>0$, then it follows that $P^{-1}(A)|\int_A XdP|\leq M$, which further implies that  $P^{-1}(A)|\int_A XdP|\leq \|X\|_\infty$.

Now for any $\epsilon>0$ small enough such that $\|X\|_\infty-\epsilon>0$, let us define
\begin{eqnarray*}
F^+&=&\{X\geq \|X\|_\infty-\epsilon\},\\
F^-&=&\{-X\geq \|X\|_\infty-\epsilon\},\\
F&=&\begin{cases}
F^+& \textrm{ if } P(F^+)\geq P(F^-);\\
F^- & \textrm{ otherwise}.
\end{cases}
\end{eqnarray*}
Then we show that $P^{-1}(F)|\int_F XdP|\geq \|X\|_\infty-\epsilon$. Since $\epsilon>0$ is arbitrary, we complete the proof.
\end{proof}

\begin{lemma}\label{lemma:basis:phi:mixing:inequality}
Consider the probability space $(\Omega, \mcA\bigotimes \mcB, Q)$. It follows that $\sum_{i=1}^k |Q(B_i|\mcA)-Q	_2(B_i)|\leq 2\phi_Q(\mcB, \mcA)$ for any integer $k$ and  any disjoint sets $B_1,\ldots, B_k\in \mcB$. 
\end{lemma}
\begin{proof}
First let us prove 
\begin{eqnarray}
|Q(B|\mcA)-Q(B)|\leq \phi_Q(\mcB, \mcA).\label{eq:lemma:basis:phi:mixing:inequality:eq1}
\end{eqnarray}
Let $X=Q(B|\mcA)-Q(B)$, and direct examination leads to
\begin{eqnarray*}
\frac{1}{Q(A)}\left|\int_A XdQ\right|=\frac{1}{Q(A)}\left|Q(B\cap A)-Q(B)Q(A)\right|=\left|Q(B|A)-Q(B)\right|.
\end{eqnarray*} 
By Lemma \ref{lemma:maxmum:is:ratio} and the above equation, we have
\begin{eqnarray*}
\left|Q(B|\mcA)-Q(B)\right|=|X|\leq \|X\|_\infty=\sup_{A\in \mcA, Q(A)>0}\left|Q(B|A)-Q(B)\right|\leq \phi_Q(\mcB, \mcA),
\end{eqnarray*}
which proves (\ref{eq:lemma:basis:phi:mixing:inequality:eq1}). Let us define
\begin{eqnarray*}
I_+=\{i: Q(B_i|\mcA)-Q(B_i)\geq 0\} \textrm{ and } I_-=\{i: Q(B_i|\mcA)-Q(B_i)< 0\}.
\end{eqnarray*}
As a consequence, we have
\begin{eqnarray*}
\sum_{i=1}^k |Q(B_i|\mcA)-Q(B_i)|&=&\sum_{i\in I_+}[Q(B_i|\mcA)-Q(B_i)]-\sum_{i\in I_-}[Q(B_i|\mcA)-Q(B_i)]\\
&=&[Q(\cup_{i\in I_+}B_i|\mcA)-Q(\cup_{i\in I_+}B_i)]-[Q(\cup_{i\in I_-}B_i|\mcA)-Q(\cup_{i\in I_-}B_i)]\\
&\leq& 2\phi_Q(\mcB, \mcA),
\end{eqnarray*}
where the last inequality follows from (\ref{eq:lemma:basis:phi:mixing:inequality:eq1}).
\end{proof}

\begin{lemma}\label{lemma:conditional:difference:measure}
Consider the probability space $(\Omega, \mcA\bigotimes \mcB, Q)$, and the marginal probability of $Q$ is $Q_1$ and $Q_2$. Let $P=Q_1\times Q_2$ be the product measure. For any measurable function $h\in \mcA\bigotimes \mcB$, it follows that
\begin{eqnarray*}
|Q(h|\mcA)-P(h|\mcA)|^p \leq C_p [Q(|h|^p)+P(|h|^p)]\phi_Q^{p-1}(\mcB, \mcA),
\end{eqnarray*}
where $C_p>0$ is a constant relying on $p$. In particular, if $|h|\leq K$, it holds that
\begin{eqnarray*}
|Q(h|\mcA)(x)-P(h|\mcA)(x)|\leq K\phi_Q(\mcB, \mcA).
\end{eqnarray*}
\end{lemma}
\begin{proof}
For simplicity, let $q=(1-1/p)^{-1}$. First, let us assume $h=\sum_{i=1}^k a_iI_{D_i}(x, y)$ is a simple function, where $D_1, \ldots, D_k \in \mcA \bigotimes \mcB$ are disjoint and $a_1,\ldots, a_k\in \mathbb{R}$. By simple algebra, it follows that
\begin{eqnarray*}
Q(I_{D_i}|\mcA)(x)=\int I_{D_i}(x, y)Q(x, dy)=\int I_{D_i^x}(y)Q(x, dy)=Q(D_i^x|\mcA)(x).
\end{eqnarray*}
Here $D_i^x=\{y: (x, y)\in D_i\}, i=1,\ldots, k$ are also disjoint. Similarly, we have $P(I_{D_i}|\mcA)(x)=P(D_i^x|\mcA)(x)=Q_2(D_i^x)(x).$ For any $x\in \Omega_\mcA$, it follows that
\begin{eqnarray*}
&&|Q(h|\mcA)(x)-P(h|\mcA)(x)|^p\\
&=&  \left|\sum_{i=1}^k a_i \left[\int I_{D_i^x}(y)Q(x, dy)-\int I_{D_i^x}(y)Q_2(dy)\right] \right|^p\\
&=&  \left|\sum_{i=1}^k a_i \left[Q(D_i^x|\mcA)(x)-Q_2(D_i^x)(x)\right] \right|^p\\
&\leq&\left[\sum_{i=1}^k |a_i| \left|Q(D_i^x|\mcA)(x)-Q_2(D_i^x)(x)\right| \right]^p\\
&\leq&\left[\sum_{i=1}^k |a_i| \left|Q(D_i^x|\mcA)(x)-Q_2(D_i^x)(x)\right|^{1/p} \left|Q(D_i^x|\mcA)(x)-Q_2(D_i^x)(x)\right|^{1/q}\right]^p\\
&\leq&\left[\sum_{i=1}^k |a_i|^p \left|Q(D_i^x|\mcA)(x)-Q_2(D_i^x)(x)\right|\right] \left[ \sum_{i=1}^k \left|Q(D_i^x|\mcA)(x)-Q_2(D_i^x)(x)\right|\right]^{p/q}\\
&\leq &[Q(|h|^p|\mcA)(x)+P(|h|^p|\mcA)(x)] \left[ \sum_{i=1}^k \left|Q(D_i^x|\mcA)(x)-Q_2(D_i^x)(x)\right|\right]^{p/q}.
\end{eqnarray*}
Using Lemma \ref{lemma:basis:phi:mixing:inequality}, we show that
\begin{eqnarray}
|Q(h|\mcA)(x)-P(h|\mcA)(x)|^p\leq 2^{p/q} [Q(|h|^p|\mcA)+P(|h|^p|\mcA)] \phi_Q^{p/q}(\mcB, \mcA).\label{eq:lemma:conditional:difference:measure:eq1}
\end{eqnarray}
Similarly, when $|a_i|\leq K$, we have
\begin{eqnarray*}
|Q(h|\mcA)(x)-P(h|\mcA)(x)|&=&  \left|\sum_{i=1}^k a_i \left[\int I_{D_i^x}(y)Q(x, dy)-\int I_{D_i^x}(y)Q_2(dy)\right] \right|\\
&=&  \left|\sum_{i=1}^k a_i \left[Q(D_i^x|\mcA)(x)-Q_2(D_i^x)(x)\right] \right|\\
&\leq&  \max_{1\leq i\leq k}|a_i|\sum_{i=1}^k  \left|Q(D_i^x|\mcA)(x)-Q_2(D_i^x)(x)\right|.
\end{eqnarray*}
Lemma  \ref{lemma:basis:phi:mixing:inequality} implies that
\begin{eqnarray}
|Q(h|\mcA)(x)-P(h|\mcA)(x)|\leq K\phi_Q(\mcB, \mcA).\label{eq:lemma:conditional:difference:measure:eq2}
\end{eqnarray}

For a general measure function $h$, there exists a sequence of simple functions $h_n$ such that $h_n\to h$ pointwise and $h_n\leq h$. Using monotone convergence theorem, we show that the (\ref{eq:lemma:conditional:difference:measure:eq1}) holds for $h$. In addition, if $|h|\leq K$, we can apply similar argument to show that (\ref{eq:lemma:conditional:difference:measure:eq2}) holds for $h$.
\end{proof}

\begin{lemma}\label{lemma:conditional:difference}
Let $(X, Y)$ and $(X, \widetilde{Y})$ be random vectors such that $X$ and $\widetilde{Y}$ are independent. Moreover, let $Y$ and $\widetilde{Y}$ have the same marginal distribution. Then it follows that 
\begin{eqnarray*}
\|\ev[h(X,Y)| X]-\ev[h(X,\widetilde{Y})| X]\| \leq m\phi(X, Y)+  \frac{\ev(\|h(X, Y)\|^{p}|X)}{m^{p-1}}+ \frac{\ev(\|h(X, \widetilde{Y})\|^{p}|X)}{m^{p-1}},
\end{eqnarray*}
where $m>0$ is an arbitrary constant. Moreover, it holds that
\begin{eqnarray*}
\|\ev[h(X,Y)| X]-\ev[h(X,\widetilde{Y})| X]\|^p\leq C_p\left(\ev[\|h(X,Y)\|^p | X]+\ev[\|h(X,\widetilde{Y})\|^p | X]\right)\phi^{p-1}(X, Y),
\end{eqnarray*}
where $C_p>0$ is a constant relying on $p$.
\end{lemma}
\begin{proof}
For simplicity, let us assume $h(x, y)$ is a scalar, as it is easy to extend to the vector case. Let $h_1(x, y)=h(x, y)I(|h(x, y)|\leq m)$ and $h_2(x, y)=h(x, y)I(|h(x, y)|> m)$. Using Lemma \ref{lemma:conditional:difference:measure}, we have
\begin{eqnarray*}
|\ev[h_1(X,Y)| X]-\ev[h_1(X,\widetilde{Y})| X]|\leq m\phi(X, Y).
\end{eqnarray*}
Moreover, it follows that
\begin{eqnarray*}
\ev(|h_2(X, Y)| |X)\leq \frac{\ev(|h(X, Y)|^{p}|X)}{m^{p-1}},\quad  \ev(|h_2(X, \widetilde{Y})| |X)\leq \frac{\ev(|h(X, \widetilde{Y})|^{p}|X)}{m^{p-1}}.
\end{eqnarray*}
Combining the above inequalities, we complete the proof of the first statement. The second statement follows from Lemma \ref{lemma:conditional:difference:measure}.
\end{proof}

\begin{lemma}\label{lemma:my:iteration:bound}
Let $c_1$ and $c_2$  be arbitrary positive constants. Let  $\{\gamma_t\}_{t=1}^\infty$ and  $\{M_t\}_{t=1}^\infty$ be two positive sequences such that 
\begin{eqnarray*}
A_t=(1-c_1\gamma_t) A_{t-1}+c_2M_t,
\end{eqnarray*}
with $\gamma_t\asymp t^{-\rho}$ for some $\rho\in (1/2, 1)$, $M_t/\gamma_t\leq M_{t-1}/\gamma_{t-1}$, and $\sum_{t=1}^\infty M_t<\infty$. Then there is a constant $C>0$ such that $A_t\leq Ce^{-C^{-1}t^{1-\rho}}+CM_{\floor{t/2}}/\gamma_t$ for all $t\geq 1$. As a consequence, if $M_t\lesssim t^{-b}$ for some $b>1$, then we have    $A_t\leq Ct^{-b+\rho}$ for all $t\geq 1$.
\end{lemma}
\begin{proof}
Since $A_t, M_t\geq 0$, by recursively substituting, we obtain that
\begin{eqnarray*}
A_t\leq \prod_{i=1}^t(1-c_1\gamma_i) A_0+c_2\sum_{i=1}^t M_i \prod_{k=i+1}^t (1-c_1\gamma_k):=S_1+S_2.
\end{eqnarray*} 
By the elementary inequality $1-x\leq e^{-x}$ for $x\geq 0$, it holds that
\begin{eqnarray*}
S_1\leq \exp\left(-c_1 \sum_{i=1}^t \gamma_i\right) A_0\leq  \exp\left( -\frac{c_1}{1-\rho}[t^{1-\rho}-2^{1-\rho}]\right)A_0.
\end{eqnarray*}
To handle $S_2$,  let $m=\floor{t/2}$, and similar calculation leads to
\begin{eqnarray*}
S_2&=&c_2\sum_{i=1}^{m} M_i \prod_{k=i+1}^t (1-c_1\gamma_k )+c_2\sum_{i=m+1}^{t} M_i \prod_{k=i+1}^t (1-c_1\gamma_k )\\
&=&c_2\sum_{i=1}^{m} M_i \prod_{k=i+1}^t (1-c_1\gamma_k )+c_2\sum_{i=m+1}^{t} \frac{M_i}{c_1\gamma_i} \left[ \prod_{k=i+1}^t (1-c_1\gamma_k )-\prod_{k=i}^t (1-c_1\gamma_k )\right]\\
&\leq & c_2 \prod_{k=m+1}^t (1-c_1\gamma_k )\sum_{i=1}^{m} M_i+  \frac{c_2M_m}{c_1 \gamma_m}\sum_{i=m+1}^{t} \left[ \prod_{k=i+1}^t (1-c_1\gamma_k )-\prod_{k=i}^t (1-c_1\gamma_k )\right]\\
&\leq& c_2\exp\left(-c_1 \sum_{k=m+1}^t \gamma_k\right) \sum_{i=1}^\infty M_i+\frac{c_2M_m}{c_1\gamma_m}\left[ 1-\prod_{k=m+1}^t (1-c_1\gamma_k )\right]\\
&\leq&c_2 \exp\left( -\frac{c_1}{1-\rho}[t^{1-\rho}-(m+1)^{1-\rho}]\right)\sum_{i=1}^\infty M_i +\frac{c_2M_m}{c_1\gamma_m}.
\end{eqnarray*}
Taking $m\asymp t/2$, since $\gamma_{m}\asymp \gamma_t\asymp t^{-\rho}$ and $t^{1-\rho}-(m+1)^{1-\rho}\gtrsim t^{1-\rho}$,  we verify that $S_2\leq Ce^{-C^{-1}t^{1-\rho}}+CM_{t}/\gamma_t$ for some $C>0$ and for all $t\geq 1$. Hence, combining the bounds of $S_1$ and $S_2$, we complete the proof.
\end{proof}

\subsection{\textbf{Consistency}}
\label{section:proof:consistency}
In this section, we first sketch the proofs of strong consistency and $L_2$ convergence.

\textbf{Proof sketch of strong consistency:}
By definition of $\theta_t$, it follows that
\begin{eqnarray*}
\Delta_t=\theta_t-\theta^*&=&\Pi\left\{\theta_{t-1}-\gamma_t  H(\theta_{t-1})-\gamma_t  e_t-\gamma_t \zeta_t\right\}-\theta^* \\
&=&\Pi\left\{\theta_{t-1}-\gamma_t  H(\theta_{t-1})-\gamma_t  e_t-\gamma_t \zeta_t\right\}-\Pi(\theta^*).
\end{eqnarray*}
Using the fact that $\|\Pi(x)-\Pi(y)\|\leq \|x-y\|$ for all $x, y \in \mathbb{R}^d$ due to the contraction property of projection, we have
\begin{eqnarray}
\|\Delta_t\|^2 &=&\left\|\Pi\left\{\theta_{t-1}-\gamma_t  H(\theta_{t-1})-\gamma_t  e_t-\gamma_t \zeta_t \right\}-\Pi(\theta^*)\right\|^2\nonumber\\
&\leq& \|\Delta_{t-1}-\gamma_t  H(\theta_{t-1})-\gamma_t  e_t-\gamma_t \zeta_t\|^2.\label{eq:sketch:proof:consistency:eq:0}
\end{eqnarray}
Using the moment bounds of $H(\theta_{t-1}), e_t, \zeta_t$ in Lemmas \ref{lemma:bound:h}-\ref{lemma:some:rate:consistency}, we will lead an inequality of the following type:
\begin{eqnarray}
\ev_{t-1}(\|\Delta_t\|^2)&\leq& (1+C\gamma_t^2)\|\Delta_{t-1}\|^2+C\gamma_t^2(1+v_t)+C\gamma_t \phi^{1/2-1/p}(B_{t-1})\sqrt{v_t}\nonumber\\
&&-C^{-1} \gamma_t\|\Delta_{t-1}\|^2. \label{eq:sketch:proof:consistency:eq:1}
\end{eqnarray}
Here $C>0$ is a constant free of $t$, and $v_t\geq 0$ is a random variable generated by the observations $W_1,\ldots, W_{t-1}$ such that $\ev(v_t)\leq C$.  By the conditions $\sum_{t=1}^\infty \gamma_t^2<\infty$, $\sum_{t=1}^\infty \gamma_t \phi^{1/2-1/p}(B_{t-1})<\infty$ in Assumption \ref{Assumption:A2} and Robbins-Siegmund Theorem (e.g., see \citealp{robbins1971convergence}), we can show that $\Delta_{t}\cae 0$. The formal proof is given in Lemma \ref{lemma:consistency}.

\textbf{Proof sketch of $L_2$ convergence:}
Taking expectation of (\ref{eq:sketch:proof:consistency:eq:1}), we have
\begin{eqnarray}
\ev(\|\Delta_t\|^2)\leq (1-C^{-1} \gamma_t+C\gamma_t^2)\ev(\|\Delta_{t-1}\|^2)+C\gamma_t^2+C\gamma_t \phi^{1/2-1/p}(B_{t-1}).\nonumber
\end{eqnarray}
Applying Lemma \ref{lemma:my:iteration:bound}, we can get the desired bound.  The formal proof is given in Lemma \ref{lemma:rate:evDelta2}.

\begin{lemma}\label{lemma:bound:h} Let  $h_p(\theta)=\ev(\|\nabla l(Z, \theta)-\nabla l(Z, \theta^*)\|^p)\to 0$. Under Assumptions \ref{Assumption:A0}-\ref{Assumption:A2} and \ref{Assumption:A4},  it follows that $h_p(\theta)\to 0$ when $\theta\to \theta^*$.
\end{lemma}
\begin{proof}
This is a direct consequence from dominated convergence theorem and Assumption \ref{Assumption:A1}\ref{A1:moment:condition}. 
\end{proof}

\begin{lemma}\label{lemma:some:inequalities}Under Assumptions \ref{Assumption:A0}-\ref{Assumption:A2} and \ref{Assumption:A4}, the following statements hold for some constant $C>0$.
\begin{enumerate}[label=(\roman*)]
\item \label{lemma:some:inequalities:item:1} $(\theta-\theta^*)^\top H(\theta)\geq C^{-1}\|\theta-\theta^*\|^2$;
\item \label{lemma:some:inequalities:item:2} $\sup_{\theta \in \Theta}\|\theta\|\leq C$;
\item \label{lemma:some:inequalities:item:3} $\sup_{\theta\in \Theta}\|H(\theta)\|\leq C$.
\end{enumerate}
\end{lemma}
\begin{proof}
Statement \ref{lemma:some:inequalities:item:1} follows from the property of strong convexity of $L$ in Assumption \ref{Assumption:A1}\ref{A1:identification}. Statements \ref{lemma:some:inequalities:item:2} and \ref{lemma:some:inequalities:item:3} are direct consequences of Assumptions \ref{Assumption:A1}\ref{A1:compact:Theta}, \ref{Assumption:A1}\ref{A1:lipchiz:hession}, \ref{Assumption:A1}\ref{A1:moment:condition}, and \ref{Assumption:A1}\ref{A1:expectation:H:theta}
\end{proof}

\begin{lemma}\label{lemma:some:rate:consistency} Suppose that Assumptions \ref{Assumption:A0}-\ref{Assumption:A2} and \ref{Assumption:A4} hold. Then there is a constant $C$ and a random variable $v_t$ depending on $W_1,\ldots, W_{t-1}$ with $\ev(v_t)\leq C$ such that 
\begin{enumerate}[label=(\roman*)]
\item \label{lemma:some:rate:consistency:item1}  $\|e_t\|^2\leq \phi^{1-2/p}(B_{t-1})v_t$;
\item \label{lemma:some:rate:consistency:item2}   $\ev[\|\widehat{H}_t(W_t, \theta)-H(\theta)\|^2]\leq CB_t^{-1}(1+\|\theta-\theta^*\|^2)$;
\item \label{lemma:some:rate:consistency:item3}  $\ev_{t-1}(\|\zeta_t\|^2)\leq \phi^{1-2/p}(B_{t-1})v_{t}+CB_t^{-1}(1+\|\theta_{t-1}-\theta^*\|^2)$.
\end{enumerate}
\end{lemma}
\begin{proof}
\begin{enumerate}[label=(\roman*)]
\item  By the definition of $e_t$ and Lemma \ref{lemma:moment:inequality:phi:mixing}, we have
\begin{eqnarray*}
\left\|e_t\right\|&=&\left\|\ev_{t-1}\left\{\widehat{H}_t({W}_t, \theta_{t-1})\right\}-\ev_{t-1}\left\{\widehat{H}_t(\widetilde{W}_t, \theta_{t-1})\right\}\right\|\\
&\leq&  m\phi(B_{t-1})+\frac{\ev_{t-1}\left(\left\|\widehat{H}_t({W}_t, \theta_{t-1})\right\|^{p/2}\right)}{m^{p/2-1}}+\frac{\ev_{t-1}\left(\left\|\widehat{H}_t(\widetilde{W}_t, \theta_{t-1})\right\|^{p/2}\right)}{m^{p/2-1}},
\end{eqnarray*}
which further leads to
\begin{eqnarray*}
\left\|e_t\right\|^2&\leq& 9m^2\phi^2(B_{t-1})+\frac{9\ev_{t-1}^2\left(\left\|\widehat{H}_t({W}_t, \theta_{t-1})\right\|^{p/2}\right)}{m^{p-2}}+\frac{9\ev_{t-1}^2\left(\left\|\widehat{H}_t(\widetilde{W}_t, \theta_{t-1})\right\|^{p/2}\right)}{m^{p-2}}\\
&\leq&9m^2\phi^2(B_{t-1})+\frac{9\ev_{t-1}\left(\left\|\widehat{H}_t({W}_t, \theta_{t-1})\right\|^{p}\right)}{m^{p-2}}+\frac{9\ev_{t-1}\left(\left\|\widehat{H}_t(\widetilde{W}_t, \theta_{t-1})\right\|^{p}\right)}{m^{p-2}}.
\end{eqnarray*}
Using Assumption \ref{Assumption:A1}\ref{A1:moment:condition} and Lemma \ref{lemma:some:inequalities}, we have
\begin{eqnarray*}
\ev^{1/p}\left(\left\|\widehat{H}_t({W}_t, \theta_{t-1})\right\|^{p}\right)&\leq& \frac{1}{B_t}\sum_{i\in I_t}\ev^{1/p}\left\{\|\nabla l(Z_i, \theta_{t-1})\|^p\right\}\\
&\leq& \frac{1}{B_t}\sum_{i\in I_t}\ev^{1/p}\left\{M^p(Z_i)\right\}=\ev^{1/p}\left\{M^p(Z)\right\}.
\end{eqnarray*}
Similarly, we can show that
\begin{eqnarray*}
\ev^{1/p}\left(\left\|\widehat{H}_t({W}_t, \theta_{t-1})\right\|^{p}\right)\leq \ev^{1/p}\left\{M^p(Z)\right\}.
\end{eqnarray*}
Combining the preceding three displays and letting $m=\phi^{-2/p}(B_{t-1})$, we conclude that $\|e_t\|^2\leq \phi^{2-4/p}(B_{t-1})v_{1t}\leq \phi^{1-2/p}(B_{t-1})v_{1t}$, where
\begin{eqnarray*}
v_{1t}=9+9\ev_{t-1}\left(\left\|\widehat{H}_t({W}_t, \theta_{t-1})\right\|^{p}\right)+9\ev_{t-1}\left(\left\|\widehat{H}_t(\widetilde{W}_t, \theta_{t-1})\right\|^{p}\right).
\end{eqnarray*}
\item This is a direct consequence of Lemma \ref{lemma:moment:inequality:phi:mixing} and Assumption \ref{Assumption:A2}\ref{A2:mixing:condition}. 

\item  Let $\widehat{\zeta}_t=U_t \widehat{H}_t(\widetilde{W}_t, \theta_{t-1})-\ev_{t-1}\left\{\widehat{H}_t(W_t, \theta_{t-1})\right\}.$ By the definition of $\zeta_t$, we can see that
\begin{eqnarray*}
|\ev_{t-1}(\|\zeta_t\|^2)-\ev_{t-1}(\|\widehat{\zeta}_t\|^2)|\leq \phi(B_{t-1})m+\frac{\ev_{t-1}(\|\zeta_t\|^p)}{m^{p/2-1}}+\frac{\ev_{t-1}(\|\widehat{\zeta}_t\|^p)}{m^{p/2-1}}
\end{eqnarray*}
Similarly to the proof of Statement \ref{lemma:some:rate:consistency:item1}, we can show that $\ev(\|\zeta_t\|^p)\leq C_p$ and $\ev(\|\widehat{\zeta}_t\|^p)\leq C_p$ for some constant $C_p>0$. Taking $m=\phi^{-2/p}(B_{t-1})$, we show that
\begin{eqnarray*}
|\ev_{t-1}(\|\zeta_t\|^2)-\ev_{t-1}(\|\widehat{\zeta}_t\|^2)|\leq \phi^{1-2/p}(B_{t-1})\left\{1+\ev_{t-1}(\|\zeta_t\|^p)+\ev_{t-1}(\|\widehat{\zeta}_t\|^p)\right\}.
\end{eqnarray*}
Moreover, direct examination leads to 
\begin{eqnarray*}
\ev_{t-1}(\|\widehat{\zeta}_t\|^2)\leq 2\ev_{t-1}(\|\widetilde{\zeta}_t\|^2)+2\|e_t\|^2,
\end{eqnarray*}
where  $\widetilde{\zeta}_t=U_t \widehat{H}_t(\widetilde{W}_t, \theta_{t-1})-\ev_{t-1}\left\{\widehat{H}_t(\widetilde{W}_t, \theta_{t-1})\right\}.$
Using Statements \ref{lemma:some:rate:consistency:item1} and \ref{lemma:some:rate:consistency:item2}, we show that
\begin{eqnarray*}
\ev_{t-1}(\|\zeta_t\|^2)\leq \phi^{1-2/p}(B_{t-1})v_{2t}+CB_t^{-1}(1+\|\theta_{t-1}-\theta^*\|^2),
\end{eqnarray*}
where $v_{t2}$ satisfies $\ev(v_{2t})\leq C$ for some constant $C$.
\end{enumerate}
We can take $v_t=v_{1t}+v_{2t}$ to complete the proof.
\end{proof}

\begin{lemma}\label{lemma:consistency} Under Assumptions \ref{Assumption:A0}-\ref{Assumption:A2} and \ref{Assumption:A4}, it follows that $\theta_T \cae \theta^*$ as $T\to \infty.$
\end{lemma}
\begin{proof}
By (\ref{eq:sketch:proof:consistency:eq:0}), we have
\begin{eqnarray*}
\|\Delta_t\|^2 &=&\left\|\Pi\left\{\theta_{t-1}-\gamma_t  H(\theta_{t-1})-\gamma_t  e_t-\gamma_t \zeta_t \right\}-\Pi(\theta^*)\right\|^2\\
&\leq& \|\Delta_{t-1}-\gamma_t  H(\theta_{t-1})-\gamma_t  e_t-\gamma_t \zeta_t\|^2\\
&=& \|\Delta_{t-1}\|^2+\gamma_t^2 \|H(\theta_{t-1})\|^2+\gamma_t^2  \|e_t\|^2+\gamma_t^2\|\zeta_t\|^2\\
&&-2\gamma_t \Delta_{t-1}^\top H(\theta_{t-1})-2\gamma_t  \Delta_{t-1}^\top e_t-2\gamma_t  \Delta_{t-1}^\top\zeta_t\\
&&+2\gamma_t^2  H^\top(\theta_{t-1})e_t+2\gamma_t^2 H^\top(\theta_{t-1})\zeta_t+2\gamma_t^2  e_t^\top \zeta_t\\
&\leq & \|\Delta_{t-1}\|^2+\gamma_t^2 \|H(\theta_{t-1})\|^2+\gamma_t^2  \|e_t\|^2+\gamma_t^2\|\zeta_t\|^2\\
&&-2\gamma_t \Delta_{t-1}^\top H(\theta_{t-1})+2\gamma_t  \|\Delta_{t-1}\| \|e_t\|-2\gamma_t  \Delta_{t-1}^\top\zeta_t\\
&&+2\gamma_t^2  \|H(\theta_{t-1})\|\|e_t\|+2\gamma_t^2 H^\top(\theta_{t-1})\zeta_t+2\gamma_t^2 \| e_t\|\| \zeta_t\|.
\end{eqnarray*}
Taking conditional expectation, it follows that
\begin{eqnarray*}
\ev_{t-1}(\|\Delta_t\|^2)&\leq&  \|\Delta_{t-1}\|^2+\gamma_t^2 \|H(\theta_{t-1})\|^2+\gamma_t^2  \|e_t\|^2+\gamma_t^2\ev_{t-1}(\|\zeta_t\|^2)\\
&&-2\gamma_t \Delta_{t-1}^\top H(\theta_{t-1})+2\gamma_t  \|\Delta_{t-1}\| \|e_t\|\\
&&+2\gamma_t^2  \|H(\theta_{t-1})\|\|e_t\|+2\gamma_t^2 \| e_t\|\ev_{t-1}(\| \zeta_t\|).
\end{eqnarray*}
Lemma \ref{lemma:some:rate:consistency} implies that $\|e_t\|^2\leq \phi^{1-2/p}(B_{t-1})v_t$ and $\ev_{t-1}(\|\zeta_t\|^2)\leq \phi^{1-2/p}(B_{t-1})v_{t}+CB_t^{-1}(1+\|\Delta_{t-1}\|^2)$ for some constant $C>0$, where $v_t$ is a random variable depending on $W_1,\ldots, W_{t-1}$ such that $\ev(v_t)\leq C$.  As a consequence of Lemma \ref{lemma:some:inequalities}, we have
\begin{eqnarray}
\ev_{t-1}(\|\Delta_t\|^2)&\leq&   \|\Delta_{t-1}\|^2+\gamma_t^2 C(1+\|\Delta_{t-1}\|^2)-2\gamma_t C^{-1}\|\Delta_{t-1}\|^2 \nonumber\\
&&+\gamma_t^2\phi^{1-1/p}(B_{t-1})v_t+\gamma_t^2\phi^{1-2/p}(B_{t-1})v_t+C\gamma_t^2B_t^{-1}(1+\|\Delta_{t-1}\|^2)\nonumber\\
&&+2C\gamma_t \sqrt{\phi^{1-2/p}(B_{t-1})v_t}+2C\gamma_t^2 \sqrt{\phi^{1-2/p}(B_{t-1})v_t}\nonumber\\
&&+2\gamma_t^2 \sqrt{\phi^{1-2/p}(B_{t-1})v_t}\sqrt{\phi^{1-2/p}(B_{t-1})v_{t}+CB_t^{-1}(1+\|\Delta_{t-1}\|^2)}\nonumber\\
&\leq&(1+C\gamma_t^2+C\gamma_t^2 B_t^{-1})\|\Delta_{t-1}\|^2\nonumber\\
&&+C\gamma_t^2(1+B_t^{-1})+4\gamma_t^2 \phi^{1-2/p}(B_{t-1})v_t+4C\gamma_t \sqrt{\phi^{1-2/p}(B_{t-1})v_t}\nonumber\\
&&+2\gamma_t^2 \sqrt{\phi^{1-2/p}(B_{t-1})v_t}\sqrt{CB_t^{-1}(1+C^2)}-2\gamma_tC^{-1}\|\Delta_{t-1}\|^2. \label{eq:lemma:consistency:eq1}
\end{eqnarray}
Moreover,  Assumption \ref{Assumption:A2} tells that  $\sum_{t=1}^\infty \gamma_t^2<\infty$ and $\sum_{t=1}^\infty \gamma_t \phi^{1/2-1/p}(B_{t-1})<\infty$. Hence,  Robbins-Siegmund Theorem (e.g., see \citealp{robbins1971convergence}) implies that $\Delta_t\cae \Delta^*$ for some random vector $\Delta^*$ and $\sum_{t=1}^\infty \gamma_t\|\Delta_{t-1}\|^2<\infty$ almost surely. The condition $\sum_{t=1}^\infty \gamma_t=\infty$ in Assumption \ref{Assumption:A2}\ref{A2:leanring:rate} implies that $\Delta_t\cae 0$.
\end{proof}


\begin{lemma}\label{lemma:rate:evDelta2}
Under Assumptions \ref{Assumption:A0}-\ref{Assumption:A2} and \ref{Assumption:A4}, it follows that
\begin{eqnarray*}
\ev(\|\theta_t-\theta^*\|^2)\leq C\left(\gamma_t+\phi^{\frac{1}{2}-\frac{1}{p}}(B_t)\right),
\end{eqnarray*}
where $C>0$ is a constant free of $t$. 
\end{lemma}
\begin{proof}
Using (\ref{eq:lemma:consistency:eq1}), it follows that
\begin{eqnarray}
\ev(\|\Delta_t\|^2)&\leq&(1+C\gamma_t^2+C\gamma_t^2 B_t^{-1})\ev(\|\Delta_{t-1}\|^2)\nonumber\\
&&+C\gamma_t^2(1+B_t^{-1})+4\gamma_t^2 \phi^{1-2/p}(B_{t-1})\ev(v_t)+4C\gamma_t \sqrt{\phi^{1-2/p}(B_{t-1})\ev(v_t)}\nonumber\\
&&+2\gamma_t^2 \sqrt{\phi^{1-2/p}(B_{t-1})\ev(v_t)}\sqrt{CB_t^{-1}(1+C^2)}-2\gamma_tC^{-1}\ev(\|\Delta_{t-1}\|^2)\nonumber\\
&\leq&(1-c_1\gamma_1)\ev(\|\Delta_{t-1}\|^2)+c_2\left(\gamma_t^2+\gamma_t \phi^{1/2-1/p}(B_{t-1})\right),\nonumber
\end{eqnarray}
for some constant $c_1, c_2>0$. Here we use the fact that $\ev(v_t)\leq C$ in Lemma \ref{lemma:some:rate:consistency}.  

Noting that $\sum_{t=1}^\infty \gamma_t<\infty$ and $\sum_{t=1}^\infty\gamma_t \phi^{1/2-1/p}(B_{t-1})<\infty$ by Assumption \ref{Assumption:A2}, the preceding display and  Lemmas \ref{lemma:my:iteration:bound} together imply
\begin{eqnarray*}
\ev(\|\Delta_t\|^2)&\leq& Ce^{-C^{-1}t^{1-\rho}}+C\gamma_t+C\gamma_t \gamma^{-1}_{\floor{t/2}}  \phi^{1/2-1/p}(B_{\floor{t/2}-1})\\
&\lesssim& t^{-\rho}+ \phi^{1/2-1/p}(t^{\beta})\asymp \gamma_t+\phi^{1/2-1/p}(B_t),
\end{eqnarray*}
where Assumption \ref{Assumption:A2} is used. Hence, we prove the desired result.
\end{proof}

\begin{proof}[\textbf{Proof of Theorem \ref{theorem:consistency}}]
Taking $U_j=1$,  it follows from Lemmas  \ref{lemma:weighted:cesaro:sum}, \ref{lemma:consistency}, and \ref{lemma:rate:evDelta2}. 
\end{proof}

\subsection{Asymptotic Normality}\label{section:proof:normal}
\textbf{Proof sketch:}
By the iteration formula, we have
\begin{eqnarray*}
\theta_t=\theta_{t-1}-\gamma_t  H(\theta_{t-1})-\gamma_t  e_t-\gamma_t \zeta_t+\xi_{t},
\end{eqnarray*}
where
\begin{eqnarray*}
\xi_t=\Pi\left\{\theta_{t-1}-\gamma_t  H(\theta_{t-1})-\gamma_t  e_t-\gamma_t \zeta_t\right\}-\left(\theta_{t-1}-\gamma_t  H(\theta_{t-1})-\gamma_t  e_t-\gamma_t \zeta_t\right).
\end{eqnarray*}
Hence, we have
\begin{eqnarray*}
 H(\theta_{t-1})=\frac{\theta_{t-1}-\theta_t}{\gamma_t}-e_t-\zeta_t+\frac{\xi_t}{\gamma_t}.
\end{eqnarray*}
For any $\theta\in \Theta$, let us define $r_\theta=H(\theta)-G(\theta-\theta^*).$
Combining the preceding three displays, it holds that
\begin{eqnarray*}
G(\theta_{t-1}-\theta^*)=\frac{\theta_{t-1}-\theta_t}{\gamma_t}-e_t-\zeta_t-r_{\theta_{t-1}}+\frac{\xi_t}{\gamma_t},
\end{eqnarray*}
which further leads to
\begin{eqnarray}
&&\sum_{t=2}^T B_{t-1}G(\theta_{t-1}-\theta^*)\nonumber\\
&=&\sum_{t=2}^T B_{t-1}\frac{\theta_{t-1}-\theta_t}{\gamma_t}-\sum_{t=2}^T B_{t-1}e_t-\sum_{t=2}^T B_{t-1}\zeta_t-\sum_{t=2}^T B_{t-1} r_{\theta_{t-1}}+\sum_{t=2}^T \frac{B_{t-1}\xi_t}{\gamma_t}. \label{eq:sketch:proof:asymptotic:normal:eq:1} 
\end{eqnarray}
In Lemmas \ref{lemma:r:theta:bound}-\ref{lemma:B:theta:difference:bound}, we will show that the term $\sum_{t=2}^T B_{t-1}\zeta_t$ contributes to the asymptotic normality, while the rest four terms are asymptotically negligible. The formal proof is given in Lemmas \ref{lemma:asymptotic:expansion} and \ref{lemma:pre:gaussian}.

\begin{lemma}\label{lemma:old:assumption}
Under Assumptions \ref{Assumption:A0}-\ref{Assumption:A3}, it holds that 
\begin{enumerate}[label=(\roman*)]
\item $\sum_{t=1}^\infty B_{t} \phi^{\frac{1}{2}-\frac{1}{p}}(B_{t})<\infty.$
\item $\sum_{t=1}^\infty \phi^{\frac{1}{2}-\frac{1}{p}}(t)<\infty.$
\end{enumerate}
\end{lemma}
\begin{proof}
Since $t^{{(2\rho+\beta+1)}/{\beta}}\phi^{{1}/{2}-{1}/{p}}(t)\to 0$ by Assumption \ref{Assumption:A4}\ref{A3:rate:batch}, we see that
\begin{eqnarray*}
\sum_{t=1}^\infty B_{t} \phi^{\frac{1}{2}-\frac{1}{p}}(B_{t})\lesssim \sum_{t=1}^\infty B_t^{-\frac{2\rho}{\beta}-\frac{1}{\beta}}\asymp \sum_{t=1}^\infty t^{-2\rho-1}<\infty.
\end{eqnarray*}
which is the first statement. 

For the second statement, similar calculation leads to
\begin{eqnarray*}
\sum_{t=1}^\infty \phi^{\frac{1}{2}-\frac{1}{p}}(t)\lesssim \sum_{t=1}^\infty t^{{-2\rho}/{\beta}-1-{1}/{\beta}}<\infty.
\end{eqnarray*}
The proof is complete.
\end{proof}

\begin{lemma}\label{lemma:bound:covariance:sum}
Under Assumptions \ref{Assumption:A0}-\ref{Assumption:A3} and \ref{Assumption:A4}, it holds that $\sum_{t=1}^\infty\|r(t)\|<\infty$.
\end{lemma}
\begin{proof}
Using the relation between $\alpha$-mixing coefficient and $\phi$-mixing coefficient (e.g., see \citealp{Bradleymixing2005}), Assumptions  \ref{Assumption:A1}\ref{A1:moment:condition}, \ref{Assumption:A2}\ref{A2:mixing:condition}, and Theorem 2.20 in \cite{fanYao2003} imply the desired result.
\end{proof}

\begin{lemma}\label{lemma:bound:tilde:zetat}
Under Assumptions \ref{Assumption:A0}-\ref{Assumption:A3} and \ref{Assumption:A4}, it follows that $\ev(\|\widetilde{\zeta}_t\|^2)\leq CB_T^{-1}\ev^{2/p}[h_p(\theta_{t-1})]$, where $h_p(\theta)=\ev(\|\nabla l(Z, \theta)-\nabla l(Z, \theta^*)\|^p)$.
\end{lemma}
\begin{proof}
Let  $\widetilde{m}_i=U_t[\nabla l(\widetilde{Z}_i, \theta_{t-1})-\nabla l(\widetilde{Z}_i, \theta^*)]$. Using Lemma \ref{lemma:conditional:difference}, it follows that
\begin{eqnarray*}
&&|\ev_{t-1}\{[\widetilde{m}_i-\ev_{t-1}(\widetilde{m}_i)]^\top [\widetilde{m}_j-\ev_{t-1}(\widetilde{m}_j)]\}|\\
&\leq& C_{p/2}^{2/p}\phi^{1-2/p}(i-j)\ev_{t-1}^{2/p}\{| [\widetilde{m}_i-\ev_{t-1}(\widetilde{m}_i)]^\top [\widetilde{m}_j-\ev_{t-1}(\widetilde{m}_j)]|^{p/2}\}\\
&\leq&  C_{p/2}^{2/p}\phi^{1-2/p}(i-j)\ev_{t-1}^{1/p}(\|\widetilde{m}_i-\ev_{t-1}(\widetilde{m}_i)\|^p)\ev_{t-1}^{1/p}(\|\widetilde{m}_j-\ev_{t-1}(\widetilde{m}_j)\|^p)\\
&\leq&  4C_{p/2}^{2/p}\phi^{1-2/p}(i-j)\ev_{t-1}^{1/p}(\|\widetilde{m}_i\|^p)\ev_{t-1}^{1/p}(\|\widetilde{m}_j\|^p)\\
&=&4C_{p/2}^{2/p}\phi^{1-2/p}(i-j)h^{2/p}_p(\theta_{t-1}).
\end{eqnarray*}
Taking expectation and using Jensen's inequality, we have
\begin{eqnarray*}
|\ev\{[\widetilde{m}_i-\ev_{t-1}(\widetilde{m}_i)]^\top [\widetilde{m}_j-\ev_{t-1}(\widetilde{m}_j)]\}|\leq   4C_{p/2}^{2/p}\phi^{1-2/p}(i-j)\ev^{2/p}[h_p(\theta_{t-1})]
\end{eqnarray*}
As a consequence, we show that
\begin{eqnarray*}
&&\ev(\|\widetilde{\zeta}_t\|^2)\\
&=&\frac{1}{B_t^2}\left(\sum_{i\in I_t}\ev\{\|\widetilde{m}_i-\ev_{t-1}(\widetilde{m}_i)\|^2\}+2\sum_{i\neq j, i,j\in I_t}\ev\{[\widetilde{m}_i-\ev_{t-1}(\widetilde{m}_i)]^\top [\widetilde{m}_j-\ev_{t-1}(\widetilde{m}_j)]\}\right)\\
&\leq& \frac{4C_{p/2}^{2/p}}{B_t}\ev^{2/p}[h_p(\theta_{t-1})]+\frac{8C_{p/2}^{2/p}}{B_t^2}\ev^{2/p}[h_p(\theta_{t-1})]\sum_{i\neq j, i,j\in I_t}\phi^{1-2/p}(i-j).
\end{eqnarray*}
Simple algebra leads to
\begin{eqnarray*}
\frac{1}{B_t}\sum_{i>j, i,j\in I_t}\phi^{1-2/p}(i-j)=\sum_{k=1}^{B_t-1} (1-k/B_t)\phi^{1-2/p}(k),
\end{eqnarray*}
which, by dominated convergence theorem and the statement $\sum_{t=1}^\infty \phi^{1-2/p}(t)\leq \sum_{t=1}^\infty \phi^{1/2-1/p}(t)<\infty$ in Lemma \ref{lemma:old:assumption}, further implies that
\begin{eqnarray*}
\lim_{B_t\to \infty}\frac{1}{B_t}\sum_{i>j, i,j\in I_t}\phi^{1-2/p}(i-j)=\sum_{k=1}^\infty \phi^{1-2/p}(k)\leq \sum_{k=1}^\infty \phi^{1/2-1/p}(k)<\infty.
\end{eqnarray*}
Hence, there is a constant $C>0$ such that $\ev(\|\widetilde{\zeta}_t\|^2)\leq CB_T^{-1}\ev^{2/p}[h_p(\theta_{t-1})].$
\end{proof}


\begin{lemma}\label{lemma:asymptotic:normal:rate:1}
Let  $\epsilon_t=U_t[\widehat{H}_t(W_t, \theta_{t-1})-\widehat{H}_t(W_t,\theta^*)]-\ev_{t-1}[\widehat{H}_t(W_t, \theta_{t-1})-\widehat{H}_t( W_t, \theta^*)]$.
Under Assumptions \ref{Assumption:A0}-\ref{Assumption:A3} and \ref{Assumption:A4}, it follows that $\ev(\|\epsilon_t\|^2)\leq CB_t^{-1}\ev^{2/p}[h_p(\theta_{t-1})]+C\phi^{1-2/p}(B_{t-1})$ for some constant $C>0$. Here $h_p(\theta)=\ev(\|\nabla l(Z, \theta)-\nabla l(Z, \theta^*)\|^p)$.
\end{lemma}
\begin{proof}
Let us define $m_i=U_t[\nabla l(Z_i, \theta_{t-1})-\nabla l(Z_i, \theta^*)]$ and $\widetilde{m}_i=U_t[\nabla l(\widetilde{Z}_i, \theta_{t-1})-\nabla l(\widetilde{Z}_i, \theta^*)]$. Since $\sup_{\theta \in \Theta}\|\nabla l(Z, \theta)\|\leq M(Z)$ by Assumption \ref{Assumption:A1}\ref{A1:moment:condition} and $\epsilon_t=\sum_{i\in I_t} [m_i-\ev_{t-1}(m_i)]/B_t$, it follows that 
\begin{eqnarray*}
\ev^{1/p}(\|\epsilon_t\|^p)\leq \ev^{1/p}(|U_t|^p) \frac{1}{B_T}\sum_{i\in I_t}2\ev^{1/p}(\|m_i\|^p)&\leq& \ev^{1/p}(|U_t|^p) \frac{1}{B_T}\sum_{i\in I_t}2\ev^{1/p}(\|M(Z_i)\|^p)\\
&=& 2\ev^{1/p}(|U|^p) \ev^{1/p}(\|M(Z)\|^p),
\end{eqnarray*}
where Assumption \ref{Assumption:A4} is used.
Similarly, we can show that  $\ev^{1/p}(\|\widetilde{\zeta}_t \|^p)\leq 2\ev^{1/p}(|U|^p) \ev^{1/p}(\|M(Z)\|^p)$, where $\widetilde{\zeta}_t=\sum_{i\in I_t} [\widetilde{m}_i-\ev_{t-1}(\widetilde{m}_i)]/B_t$.
Using Lemma \ref{lemma:conditional:difference}, we have
\begin{eqnarray*}
|\ev_{t-1}(\|\epsilon_t\|^2)-\ev_{t-1}(\|\widetilde{\zeta}_t\|^2)|^{p/2}\leq C_{p/2}\phi^{p/2-1}(B_{t-1}) [\ev_{t-1}(\|\epsilon_t\|^p)+\ev_{t-1}(\|\widetilde{\zeta}_t\|^p)],
\end{eqnarray*}
where $C_{p/2}$ is a constant depending on $p$. Combining the last three inequalities with Lemma \ref{lemma:bound:tilde:zetat}, we show that
\begin{eqnarray*}
\ev(\|\epsilon_t\|^2)&\leq& \ev(\|\widetilde{\zeta}_t\|^2)+C_{p/2}^{2/p}\phi^{1-2/p}(B_{t-1})[\ev^{2/p}(\|\epsilon_t\|^p)+\ev^{2/p}(\|\widetilde{\epsilon}_t\|^p)]\\
&\leq& CB_t^{-1}\ev^{2/p}[h_p(\theta_{t-1})]+C\phi^{1-2/p}(B_{t-1}),
\end{eqnarray*}
for some constant $C>0$.
\end{proof}

\begin{lemma}\label{lemma:asymptotic:normal:rate:2}
Under Assumptions \ref{Assumption:A0}-\ref{Assumption:A3} and \ref{Assumption:A4}, there is a constant $C>0$ such that $\ev(\|\ev_{t-1}[\widehat{H}_t(W_t, \theta^*)]\|^2)\leq C \phi^{2-2/p}(B_{t-1})B_t^{-1}$.
\end{lemma}
\begin{proof}
Lemma \ref{lemma:conditional:difference} implies that
\begin{eqnarray*}
\|\ev_{t-1}[\widehat{H}_t(W_t, \theta^*)]\|^p&=&\|\ev_{t-1}[\widehat{H}_t(W_t, \theta^*)]-\ev_{t-1}[\widehat{H}_t(\widetilde{W}_t, \theta^*)]\|^p\\
&\leq& C_p \phi^{p-1}(B_{t-1})\left(\ev_{t-1}[\|\widehat{H}_t(W_t, \theta^*)\|^p]+\ev_{t-1}[\|\widehat{H}_t(\widetilde{W}_t, \theta^*)\|^p]\right).
\end{eqnarray*}
Using Lemma \ref{lemma:moment:inequality:phi:mixing} and Assumption \ref{Assumption:A2}\ref{A2:mixing:condition}, we can show that 
\begin{eqnarray*}
\ev[\|\widehat{H}_t(W_t, \theta^*)\|^p]\leq C_pB_t^{-p/2}\; \textrm{ and }\; \ev[\|\widehat{H}_t(\widetilde{W}_t, \theta^*)\|^p]\leq C_pB_t^{-p/2}.
\end{eqnarray*}
As consequence, we have
\begin{eqnarray*}
\ev(\|\ev_{t-1}[\widehat{H}_t(W_t, \theta^*)]\|^2)\leq \ev^{2/p}(\|\ev_{t-1}[\widehat{H}_t(W_t, \theta^*)]\|^p)\leq C \phi^{2-2/p}(B_{t-1})B_t^{-1}.
\end{eqnarray*}
The proof is complete.
\end{proof}


\begin{lemma}\label{lemma:r:theta:bound}
Let $r_\theta=H(\theta)-G(\theta-\theta^*)$.  Under Assumptions \ref{Assumption:A0}-\ref{Assumption:A3} and \ref{Assumption:A4}, it holds that
\begin{eqnarray*}
\sum_{t=2}^T B_{t-1} r_{\theta_{t-1}}=o_P\left(\sqrt{\sum_{t=1}^T B_t}\right).
\end{eqnarray*}
\end{lemma}
\begin{proof}
The definition of $r_{\theta}$ implies that
\begin{eqnarray*}
\|r_\theta\|&=&\|H(\theta)-G(\theta-\theta^*)\|\\
&=&\|\nabla L(\theta)-G(\theta-\theta^*)-\nabla L(\theta^*)\|\\
&=&\left\|\left\{\nabla^2 L(\theta')-G\right\}(\theta-\theta^*)\right\|=\left\|\left\{\nabla^2 L(\theta')-\nabla^2 L(\theta^*)\right\}(\theta-\theta^*)\right\|,
\end{eqnarray*}
where $\theta'=c\theta+(1-c)\theta^*$ for some $c\in [0, 1]$. By Assumption \ref{Assumption:A1}\ref{A1:lipchiz:hession}, we see that $\|r_\theta\|\leq K \|\theta-\theta^*\|^2.$
Here $K$ is a constant free of $\theta$.  By Lemma \ref{lemma:rate:evDelta2} and Assumption \ref{Assumption:A2}, we have
\begin{eqnarray*}
\ev\left(\left\|\sum_{t=2}^T B_{t-1} r_{\theta_{t-1}}\right\|\right)&\leq& K \sum_{t=2}^T B_{t-1}\ev(\|\theta_{t-1}-\theta^*\|^2)\\
&\leq& C\sum_{t=1}^T B_{t}\gamma_t+C\sum_{t=1}^T B_{t}\phi^{1/2-1/p}(B_t):=S_1+S_2.
\end{eqnarray*}
Here the definitions of $S_1$ and $S_2$ are clear from the context.

To handle $S_1$, using the rate conditions in Assumption \ref{Assumption:A2}, we see that
\begin{eqnarray*}
S_1\asymp \sum_{t=1}^T t^{\beta-\rho}\asymp T^{\beta-\rho+1}=o\left(T^{(\beta+1)/2}\right),
\end{eqnarray*}
where  the last equation follows from the rate condition $\beta<2\rho-1$ in Assumption \ref{Assumption:A3}\ref{A3:rate:phi}. 

For $S_2$, the condition $t^{(\beta+1)/(2\beta)}\phi^{1/2-1/p}(t)\leq t^{{(2\rho+\beta+1)}/{\beta}}\phi^{{1}/{2}-{1}/{p}}(t)\to 0$ in Assumption \ref{Assumption:A4}\ref{A3:rate:batch}  implies that 
\begin{eqnarray*}
  \frac{B_t\phi^{1/2-1/p}(B_t)}{t^{\frac{\beta-1}{2}}} \asymp \frac{B_t\phi^{1/2-1/p}(B_t)}{B_t^{\frac{\beta-1}{2\beta}}}=B_t^{\frac{\beta+1}{2\beta}}\phi^{1/2-1/p}(B_t)\to 0.
\end{eqnarray*}
Using the preceding limit and Lemma \ref{lemma:weighted:cesaro:sum:3}, it follows that
\begin{eqnarray*}
S_2\asymp \sum_{t=1}^T B_t\phi^{1/2-1/p}(B_t)=o\left(T^{(\beta+1)/2}\right).
\end{eqnarray*}

Combining the above bounds of $S_1, S_2$ and noting that $\sum_{t=1}^T B_t\asymp T^{\beta+1}$, we complete the proof. 
%
\end{proof}

\begin{lemma}\label{lemma:B:e:bound}
Under Assumptions \ref{Assumption:A0}-\ref{Assumption:A3} and \ref{Assumption:A4}, it holds that
\begin{eqnarray*}
\sum_{t=2}^T B_{t-1}e_t=o_P\left(\sqrt{\sum_{t=1}^T B_t} \right).
\end{eqnarray*}
\end{lemma}
\begin{proof}
By Lemma \ref{lemma:some:rate:consistency}, it holds that
\begin{eqnarray*}
\left\|\sum_{t=2}^T B_{t-1}e_t\right\|\leq \sum_{t=2}^T B_{t-1}\|e_t\| \leq \sum_{t=2}^T B_{t-1} \phi^{1/2-1/p}(B_{t-1})\sqrt{v_t},
\end{eqnarray*}
where $v_t\geq 0$ is a random variable such that $\ev(v_t)\leq C$ for some $C>0$. Taking expectation, we have
\begin{eqnarray*}
\ev\left\{\left\|\sum_{t=2}^T B_{t-1} e_t\right\|\right\}\leq C \sum_{t=2}^T B_{t-1} \phi^{1/2-1/p}(B_{t-1})=o\left(\sqrt{\sum_{t=1}^T B_t} \right).
\end{eqnarray*}
Here the last equation follows from Lemma \ref{lemma:old:assumption}. Hence, we complete the proof.
\end{proof}

\begin{lemma}\label{lemma:B:zeta:bound}
Under Assumptions \ref{Assumption:A0}-\ref{Assumption:A3} and \ref{Assumption:A4}, it holds that
\begin{eqnarray*}
\sum_{t=2}^T B_{t-1}\zeta_t=\sum_{t=1}^T B_{t}U_t\widehat{H}_t(W_t, \theta^*)+o_P\left(\sqrt{\sum_{t=1}^T B_t}\right).
\end{eqnarray*}
\end{lemma}
\begin{proof}
The proof is divided into two steps.

\textbf{Step 1:}  First, we will prove that
\begin{eqnarray}
\sum_{t=2}^T B_{t-1}\zeta_t=\sum_{t=2}^T B_{t-1}U_t\widehat{H}_t(W_t, \theta^*)+o_P\left(\sqrt{\sum_{t=1}^T B_t}\right).\label{eq:lemma:B:zeta:bound:eq:0}
\end{eqnarray}
Consider the decomposition $\zeta_t=U_t\widehat{H}_t(W_t, \theta^*)+\epsilon_t-y_t$, where 
\begin{eqnarray*}
\epsilon_t&=&U_t[\widehat{H}_t(W_t, \theta_{t-1})-\widehat{H}_t(W_t, \theta^*)]-\ev_{t-1}[\widehat{H}_t(W_t, \theta_{t-1})-\widehat{H}_t(W_t, \theta^*)],\\
y_t&=&\ev_{t-1}[\widehat{H}_t(W_t, \theta^*)].
\end{eqnarray*}
It suffices to show 
\begin{eqnarray}
\ev\left(\left\|\sum_{t=2}^T B_{t-1} \epsilon_t\right\|^2 \right)=o\left({\sum_{t=1}^T B_t} \right)\;\textrm{ and }\; \ev\left(\left\|\sum_{t=2}^T B_{t-1} y_t\right\|^2\right)=o\left({\sum_{t=1}^T B_t} \right).\label{eq:lemma:B:zeta:bound:eq:1}
\end{eqnarray}
By Lemma \ref{lemma:asymptotic:normal:rate:1} and the fact that $\ev(\epsilon_j^\top \epsilon_i)=0$ for $i\neq j$, we have
\begin{eqnarray*}
\ev\left(\left\|\sum_{t=2}^T B_{t-1} \epsilon_t\right\|^2 \right)&=&\sum_{t=2}^T B_{t-1}^2 \ev(\|\epsilon_t\|^2)\\
&\leq&C\sum_{t=1}^T B_t\ev^{2/p}\left\{h_p(\theta_{t-1})\right\}+C\sum_{t=1}^T B_{t-1}^2 \phi^{1-2/p}(B_{t-1})\\
&=&C\sum_{t=1}^T B_t\ev^{2/p}\left\{h_p(\theta_{t-1})\right\}+C\sum_{t=1}^T B_{t-1}  \left(B_{t-1}\phi^{1-2/p}(B_{t-1})\right).
\end{eqnarray*}
Here $h_p(\theta)\geq 0$ is defined in Lemma \ref{lemma:bound:h} such that $\lim_{\theta \to \theta^*}h_p(\theta)=0$ and $\sup_{\theta\in \Theta}h_p(\theta)<\infty$. Since $\theta_t\cae \theta^*$ by Lemma \ref{lemma:consistency}, Lebesgue dominated convergence theorem implies that $\ev^{2/p}\{h_p(\theta_{t-1})\}\to 0$. Moreover, the statement $\sum_{t=1}^\infty B_t\phi^{1/2-1/p}(B_t)<\infty$ in Lemma \ref{lemma:old:assumption} and Lemma \ref{lemma:comparison:test} imply that $\lim_{t\to \infty}B_t\phi^{1-2/p}(B_t)\leq \lim_{t\to \infty}B_t\phi^{1/2-1/p}(B_t)=0$. Using Lemma \ref{lemma:weighted:cesaro:sum} and the preceding display, we verify the first inequality in (\ref{eq:lemma:B:zeta:bound:eq:1}).

Using  Lemma \ref{lemma:asymptotic:normal:rate:2}, we have
\begin{eqnarray*}
\ev\left(\left\| \sum_{t=2}^T  B_{t-1}y_t\right\|^2\right)&\leq& \sum_{i=2}^T \sum_{j=2}^T B_{i-1}B_{j-1} \ev(\|y_i\|\|y_j\|)\\
&\leq& C\sum_{i=2}^T \sum_{j=2}^T  B_{i-1}B_{j-1} \phi^{1-1/p}(B_{j-1})\phi^{1-1/p}(B_{i-1})B_j^{-1/2}B_i^{-1/2} \\
&= &C\left\{\sum_{j=2}^T   B_{j-1}\left(\phi^{1-1/p}(B_{j-1})B_j^{-1/2}\right)\right\}^2\\
&\leq& C\left\{\sum_{t=2}^T   B_{t}^{1/2}\phi^{1-1/p}(B_{t})\right\}^2=o\left(\sum_{t=1}^T B_t\right).
\end{eqnarray*}
Here the last equation follows from Lemma \ref{lemma:old:assumption}. Hence, the second inequality in (\ref{eq:lemma:B:zeta:bound:eq:1}) holds.

\textbf{Step 2:}  By direct examination, it holds that
\begin{eqnarray*}
\left\|\sum_{t=2}^T B_{t-1}U_t\widehat{H}_t(W_t, \theta^*)-\sum_{t=1}^T B_{t}U_t\widehat{H}_t(W_t, \theta^*)\right\|\leq \sum_{t=1}^T (B_t-B_{t-1})\left\|U_t\widehat{H}_t(W_t, \theta^*)\right\|,
\end{eqnarray*}
where $B_0=0$. Taking expectation and using Lemma \ref{lemma:moment:inequality:phi:mixing}, we see that
\begin{eqnarray*}
&&\ev\left\{\left\|\sum_{t=2}^T B_{t-1}U_t\widehat{H}_t(W_t, \theta^*)-\sum_{t=1}^T B_{t}U_t\widehat{H}_t(W_t, \theta^*)\right\|\right\}\\
&\leq&C\sum_{t=1}^T (B_t-B_{t-1})B_t^{-1/2}\asymp \sum_{t=1}^T t^{\beta-1} t^{-\beta/2}\asymp T^{\beta/2}=o\left(T^{(\beta+1)/2}\right).
\end{eqnarray*}
Combining the preceding display and (\ref{eq:lemma:B:zeta:bound:eq:0}) and noting that  $\sum_{t=1}^T B_t\asymp T^{\beta+1}$, we complete the proof.
\end{proof}

\begin{lemma}\label{lemma:B:theta:difference:bound}
Under Assumptions \ref{Assumption:A0}-\ref{Assumption:A3} and \ref{Assumption:A4}, it holds that
\begin{eqnarray*}
\sum_{t=2}^T B_{t-1}\frac{\theta_t-\theta_{t-1}}{\gamma_t}=\left(\sqrt{\sum_{t=1}^T B_t}\right).
\end{eqnarray*}
\end{lemma}
\begin{proof}
Let $a_{t-1}=B_{t-1}/\gamma_t$. Applying the Abel summation, we get
\begin{eqnarray*}
\sum_{t=2}^T B_{t-1}\frac{\theta_t-\theta_{t-1}}{\gamma_t}&=&\sum_{t=2}^T a_{t-1}(\theta_t-\theta_{t-1})\\
&=&\sum_{t=1}^{T-1} a_{t}(\Delta_{t+1}-\Delta_{t})\\
&=&-\sum_{t=2}^{T-1}\Delta_t(a_t-a_{t-1})+a_{T-1}\Delta_{T}-a_1 \Delta_1:=-S_1+S_2-S_3.
\end{eqnarray*}
Here the definitions of $S_1, S_2, S_3$ are clear from the context. 

First, let us investigate $S_1$. By direct examination, it holds that
\begin{eqnarray*}
S_1=\sum_{t=2}^{T-1}\Delta_t(a_t-a_{t-1})&=&\sum_{t=2}^{T-1}\Delta_t\left(\frac{B_t}{\gamma_{t+1}}-\frac{B_{t-1}}{\gamma_{t}}\right)\\
&=&\sum_{t=2}^{T-1}\Delta_t\left(\frac{B_t}{\gamma_{t+1}}-\frac{B_{t}}{\gamma_{t}}\right)+\sum_{t=2}^{T-1}\Delta_t\left(\frac{B_t}{\gamma_{t}}-\frac{B_{t-1}}{\gamma_{t}}\right)\\
&:=&S_{11}+S_{12},
\end{eqnarray*}
where the definitions of $S_{11}, S_{12}$ are clear. Using Lemma \ref{lemma:rate:evDelta2} and Assumption \ref{Assumption:A2}, we obtain that
\begin{eqnarray}
\ev(\|S_{11}\|)&\leq&\sum_{t=2}^{T-1}\sqrt{\ev(\|\Delta_t\|^2)}B_t\left(\frac{1}{\gamma_{t}}-\frac{1}{\gamma_{t+1}}\right)\nonumber\\
&\leq&C \sum_{t=2}^{T-1} \gamma_t^{1/2} B_t\left(\frac{1}{\gamma_{t}}-\frac{1}{\gamma_{t+1}}\right)+C \sum_{t=2}^{T-1} \phi^{\frac{1}{4}-\frac{1}{2p}}(B_t) B_t\left(\frac{1}{\gamma_{t}}-\frac{1}{\gamma_{t+1}}\right)\nonumber\\
&\asymp& \sum_{t=1}^T t^{-\rho/2} t^{\beta} t^{\rho-1}+\sum_{t=1}^T \phi^{\frac{1}{4}-\frac{1}{2p}}(B_t) B_t t^{\rho-1}.\label{eq:lemma:B:theta:difference:bound:eq:1}
\end{eqnarray}
Using the condition $\rho+\beta<1$ in Assumption \ref{Assumption:A3}\ref{A3:rate:phi}, we have
\begin{eqnarray*}
 \sum_{t=1}^T t^{-\rho/2} t^{\beta} t^{\rho-1}\asymp  \sum_{t=1}^T  t^{\rho/2+\beta-1}\asymp T^{\rho/2+\beta}=o\left(T^{(\beta+1)/2}\right).
\end{eqnarray*}
Moreover, the condition $t^{{(2\rho+\beta+1)}/{\beta}}\phi^{{1}/{2}-{1}/{p}}(t)\to 0$ in Assumption \ref{Assumption:A4}\ref{A3:rate:batch} implies that 
\begin{eqnarray*}
\frac{B_t \phi^{\frac{1}{4}-\frac{1}{2p}}(B_t)  t^{\rho-1}}{t^{\frac{\beta-1}{2}}}= B_t \phi^{\frac{1}{4}-\frac{1}{2p}}(B_t)  t^{\frac{2\rho-\beta+1}{2}}\asymp B_t^{\frac{2\rho+\beta+1}{2\beta}}\phi^{\frac{1}{4}-\frac{1}{2p} }(B_t)\to 0.
\end{eqnarray*}
Using the preceding limit and Lemma \ref{lemma:weighted:cesaro:sum:3}, it follows that
\begin{eqnarray*}
\sum_{t=1}^T B_t \phi^{\frac{1}{4}-\frac{1}{2p}}(B_t) t^{\rho-1} =o\left(T^{(\beta+1)/2}\right).
\end{eqnarray*}
Hence, the preceding four displays shows that $S_{11}=o_P(T^{(\beta+1)/2})$

Similarly, we also have
\begin{eqnarray*}
\ev(\|S_{12}\|)&\leq&\sum_{t=2}^{T-1}\sqrt{\ev(\|\Delta_t\|^2)} \gamma_t^{-1} (B_t-B_{t-1})\\
&\leq&C \sum_{t=2}^{T-1} \gamma_t^{-1/2} (B_t-B_{t-1})+C\sum_{t=2}^{T-1} \phi^{\frac{1}{4}-\frac{1}{2p}}(B_t)\gamma_t^{-1} (B_t-B_{t-1})  \\
&\asymp& \sum_{t=1}^T t^{\rho/2}t^{\beta-1}+\sum_{t=1}^{T} \phi^{\frac{1}{4}-\frac{1}{2p}}(B_t) t^{\rho+\beta-1}\\
&\asymp& T^{\rho/2+\beta}+\sum_{t=1}^{T}\phi^{\frac{1}{4}-\frac{1}{2p}}(B_t) t^{\beta-1}=o\left(T^{(\beta+1)/2}\right).
\end{eqnarray*}
Here the last equation follows from the same calculation as that in (\ref{eq:lemma:B:theta:difference:bound:eq:1}). Combining the bounds of $S_{11}$ and $S_{12}$, we show that $S_1=o_P(T^{(\beta+1)/2}).$

For $S_2$, Lemma \ref{lemma:rate:evDelta2} tells that
\begin{eqnarray*}
\ev(\|S_2\|)\leq a_{T-1}\sqrt{\ev(\|\Delta_T\|^2)}&\leq& CB_{T}\gamma_T^{-1/2}+CB_T\gamma_T^{-1}\phi^{\frac{1}{4}-\frac{1}{2p}}(B_T)\\
&\asymp& T^{\rho/2+\beta}+ T^{\rho}B_T\phi^{\frac{1}{4}-\frac{1}{2p}}(B_T).
\end{eqnarray*}
Since  $t^{{(2\rho+\beta+1)}/{\beta}}\phi^{{1}/{2}-{1}/{p}}(t)\to 0$ by Assumption \ref{Assumption:A4}\ref{A3:rate:batch}, we see that
\begin{eqnarray*}
B_T\phi^{\frac{1}{4}-\frac{1}{2p}}(B_T)\lesssim B_T^{-\frac{\rho}{\beta}+\frac{1}{2}-\frac{1}{2\beta}}\asymp T^{-\rho+\frac{\beta-1}{2}}.
\end{eqnarray*}
Combing the preceding two displays, we show that
\begin{eqnarray*}
\ev(\|S_2\|)\lesssim T^{\rho/2+\beta}+T^{\frac{\beta-1}{2}}=o\left(T^{(\beta+1)/2}\right).
\end{eqnarray*}

Since $S_3$ is stochastically bounded,  it holds that $S_3=O_P(1)=o_P(T^{(\beta+1)/2}).$ 

Combining the above bounds of $S_i$'s and noting that $\sum_{t=1}^T B_t\asymp T^{\beta+1}$, we complete the proof. 
\end{proof}

\begin{lemma}\label{lemma:asymptotic:expansion}
Under Assumptions \ref{Assumption:A0}-\ref{Assumption:A3} and \ref{Assumption:A4}, it holds that
\begin{eqnarray*}
\sum_{t=1}^T B_t(\theta_t-\theta^*)= \sum_{t=1}^T  U_t B_{t} G^{-1}\widehat{H}_t(W_t, \theta^*)+o_P\left(\sqrt{\sum_{t=1}^T B_t}\right).
\end{eqnarray*}
\end{lemma}
\begin{proof}
By (\ref{eq:sketch:proof:asymptotic:normal:eq:1}), we have
\begin{eqnarray}
&&\sum_{t=2}^T B_{t-1}G(\theta_{t-1}-\theta^*)\nonumber\\
&=&\sum_{t=2}^T B_{t-1}\frac{\theta_{t-1}-\theta_t}{\gamma_t}-\sum_{t=2}^T B_{t-1}e_t-\sum_{t=2}^T B_{t-1}\zeta_t-\sum_{t=2}^T B_{t-1} r_{\theta_{t-1}}+\sum_{t=2}^T \frac{B_{t-1}\xi_t}{\gamma_t}. \label{eq:lemma:asymptotic:expansion:eq:0}
\end{eqnarray}

Since $\theta_t \cae \theta^*$ by Lemma \ref{lemma:consistency}, and $\theta^*$ is in the interior of $\Theta$ due to Assumption \ref{Assumption:A1}\ref{A1:identification}, the projection only happens finite times. If we define $\tau=\sup\{t\geq 1: \xi_t\neq 0\}$, then it follows that $\pr(\tau<\infty)=1$. Therefore, we conclude that
\begin{eqnarray*}
\pr\left(\sum_{t=2}^T \frac{B_{t-1}\xi_t}{\gamma_t}<\infty\right)\leq \pr\left(\sum_{t=2}^\infty \frac{B_{t-1}\|\xi_t\|}{\gamma_t}<\infty\right)=\pr\left(\sum_{t=2}^\tau \frac{B_{t-1}\|\xi_t\|}{\gamma_t}<\infty\right)=1.
\end{eqnarray*}
By the preceding display, (\ref{eq:lemma:asymptotic:expansion:eq:0}) and Lemmas \ref{lemma:r:theta:bound}-\ref{lemma:B:theta:difference:bound}, we have
\begin{eqnarray}
\sum_{t=1}^{T-1} B_{t}G(\theta_{t}-\theta^*)=\sum_{t=2}^T B_{t-1}G(\theta_{t-1}-\theta^*)=\sum_{t=1}^T B_{t}U_t \widehat{H}_t(W_t, \theta^*)+o_P\left(\sqrt{\sum_{t=1}^T B_t}\right).\label{eq:lemma:asymptotic:expansion:eq:1}
\end{eqnarray}
Moreover, Lemma \ref{lemma:rate:evDelta2}  and Assumption \ref{Assumption:A2} tell that 
\begin{eqnarray*}
\ev\left\{B_T^2\|\theta_T-\theta^*\|^2)\right\}\leq CB_T^2\gamma_T.+CB_T^2\phi^{\frac{1}{2}-\frac{1}{p}}(B_T).
\end{eqnarray*}
Because of the condition $t^{{(2\rho+\beta+1)}/{\beta}}\phi^{{1}/{2}-{1}/{p}}(t)\to 0$ in Assumption \ref{Assumption:A4}\ref{A3:rate:batch}, we see that
\begin{eqnarray*}
B_T^2\phi^{\frac{1}{2}-\frac{1}{p}}(B_T)\lesssim B_T^{-\frac{2\rho}{\beta}+1-\frac{1}{\beta}}\asymp T^{-2\rho+\beta-1}.
\end{eqnarray*}
Combining the preceding two displays, we conclude that
\begin{eqnarray}
\ev\left\{B_T^2\|\theta_T-\theta^*\|^2)\right\}\lesssim T^{2\beta-\rho}+T^{-2\rho+\beta-1}=o\left(T^{\beta+1}\right).\label{eq:lemma:asymptotic:expansion:eq:2}
\end{eqnarray}
Here the last equation is due to $\beta<\rho+1$ from Assumption \ref{Assumption:A3}\ref{A3:rate:phi}.

Finally, using (\ref{eq:lemma:asymptotic:expansion:eq:1}), (\ref{eq:lemma:asymptotic:expansion:eq:2}), and  the fact  $\sum_{t=1}^T B_t\asymp T^{\beta+1}$, we complete the proof.
\end{proof}

\begin{lemma}\label{lemma:pre:gaussian}
Suppose that Assumptions \ref{Assumption:A0}-\ref{Assumption:A3} hold. Then it follows that
\begin{eqnarray*}
\frac{1}{\sqrt{2\sum_{t=1}^T B_t}}\sum_{t=1}^T B_t\left\{\widehat{H}_t^a(W_t^a, \theta^*)+\widehat{H}_t^b(W_t^b, \theta^*)\right\}\cid N(0,  V),
\end{eqnarray*}
where $V= r(0)+2\sum_{k=1}^\infty r(k)$.
\end{lemma}
\begin{proof}
First, we notice that
\begin{eqnarray*}
B_t\left\{\widehat{H}_t^a(W_t^a, \theta^*)+\widehat{H}_t^b(W_t^b, \theta^*)\right\}=\sum_{i\in I_t\cup J_t} \nabla l(Z_i, \theta^*),
\end{eqnarray*}
which is a sum of a $\phi$-mixing sequence. Using the CLT for mixing sequence (e.g., Theorem 2.21 in \citealp{fanYao2003}) and the statement $\sum_{t=1}^\infty \phi^{1/2-1/p}(t)<\infty$ in Lemma \ref{lemma:old:assumption}, we complete the proof.
\end{proof}

\begin{proof}[\textbf{Proof of Theorem \ref{theorem:asymptotic:normality}}]
 By Lemma \ref{lemma:asymptotic:expansion} with $U_t=1$, we obtain that
\begin{eqnarray*}
\sum_{t=1}^T B_t(\theta_t^a-\theta^*)&=&\sum_{t=1}^T  B_{t} G^{-1}\widehat{H}_t(W_t^a, \theta^*)+o_P\left(\sqrt{\sum_{t=1}^T B_t}\right),\\
\sum_{t=1}^T B_t(\theta_t^b-\theta^*)&=&\sum_{t=1}^T  B_{t} G^{-1}\widehat{H}_t(W_t^b, \theta^*)+o_P\left(\sqrt{\sum_{t=1}^T B_t}\right).
\end{eqnarray*}
Therefore, we have
\begin{eqnarray*}
\overline{\theta}_T-\theta^*=\frac{1}{2\sum_{t=1}^T B_t}\sum_{t=1}^T  B_t G^{-1}\left\{\widehat{H}_t^a(W_t^a, \theta^*)+\widehat{H}_t^b(W_t^b, \theta^*)\right\}+o_P\left\{\left(\sum_{t=1}^T B_t\right)^{-1/2}\right\}.
\end{eqnarray*}
Applying Lemma \ref{lemma:pre:gaussian}, we complete the proof.
\end{proof}

\subsection{\textbf{Proof of Theorem \ref{theorem:bootstrap:asymptotic}}}\label{section:proof:bootstrap}
\begin{lemma}\label{lemma:variance:asymptotic}
 Suppose that Assumptions \ref{Assumption:A0}-\ref{Assumption:A3} hold. Then it follows that
\begin{eqnarray*}
\frac{1}{2\sum_{t=1}^T B_t}\sum_{t=1}^T \ev\left(B_t^2 \left\{\widehat{H}_t^a(W_t^a, \theta^*)+\widehat{H}_t^b(W_t^b, \theta^*)\right\}^2\right)\to r(0)+2\sum_{k=1}^\infty r(k).
\end{eqnarray*}
\end{lemma}
\begin{proof}
For simplicity, we assume $d=1$. The extension to $d>1$ can be made using the transformation $\widehat{H}_t^a+\widehat{H}_t^b\to v^\top (\widehat{H}_t^a+\widehat{H}_t^b)$ for $v\in \mathbb{R}^d$ and Cram\'er–Wold theorem.
\end{proof}

Let us define $D_t=I_t\cup J_t$, $u_i=\nabla l(Z_i, \theta^*)$, and $r(t)=\ev(u_{i+t} u_i)$.  By direct examination, it holds that
\begin{eqnarray*}
\ev\left\{\left(\sum_{i\in D_t}u_i\right)^2\right\}&=&2B_t r(0)+2\sum_{k=1}^{2B_t} (2B_t-k)r(k).
\end{eqnarray*}
Since $\sum_{m=1}^\infty |r(m)|<\infty$ by Lemma \ref{lemma:bound:covariance:sum}, dominated convergence theorem  implies that
\begin{eqnarray*}
\lim_{t\to \infty} \frac{1}{2B_{t}} \ev\left\{\left(\sum_{i\in D_t}u_i\right)^2\right\}=r(0)+ \lim_{t\to \infty}2\sum_{k=1}^{2B_t} \left(1-\frac{k}{2B_t}\right)r(k)=r(0)+2\sum_{k=1}^\infty r(k).
\end{eqnarray*}
Combining the preceding display with Lemma \ref{lemma:weighted:cesaro:sum}, we conclude that
\begin{eqnarray*}
\frac{1}{2\sum_{t=1}^T B_t}\sum_{t=1}^T \ev\left(B_t^2 \left\{\widehat{H}_t^a(W_t^a, \theta^*)+\widehat{H}_t^b(W_t^b, \theta^*)\right\}^2\right)&=&\frac{1}{2\sum_{t=1}^T B_t}\sum_{t=1}^T \ev\left\{\left(\sum_{i\in D_t}u_i\right)^2\right\}\\
&=&\frac{1}{2\sum_{t=1}^T B_t}\sum_{t=1}^T 2B_t \frac{ \ev\left\{\left(\sum_{i\in D_t}u_i\right)^2\right\}}{2B_t}\\
&\to& r(0)+2\sum_{k=1}^\infty r(k).
\end{eqnarray*}
Hence, we complete the proof.

\begin{proof}[\textbf{Proof of Theorem \ref{theorem:bootstrap:asymptotic}}]
Let $\ev^*(\cdot)$ denote the conditional expectation given $\mcD_T$, and let $v_t=B_t\{\widehat{H}_t^a(W_t^a, \theta^*)+\widehat{H}_t^b(W_t^b, \theta^*)\}$. Then it follows from Lemma \ref{lemma:asymptotic:expansion} with $U_t=V_t$ that
\begin{eqnarray*}
\sqrt{2\sum_{t=1}^T B_t}(\overline{\theta}_T^*-\overline{\theta}_T)&=&\frac{1}{\sqrt{2\sum_{t=1}^T B_t}}\sum_{t=1}^T (V_t-1)G^{-1}B_t\left\{\widehat{H}_t^a(W_t^a, \theta^*)+\widehat{H}_t^b(W_t^b, \theta^*)\right\}+o_P(1)\\
&=&\frac{1}{\sqrt{2\sum_{t=1}^T B_t}}\sum_{t=1}^T (V_t-1)G^{-1}v_t+o_P(1).
\end{eqnarray*}
It suffices to study the limiting behavior of $Y_T:=\sum_{t=1}^T (V_t-1)v_t/\sqrt{2\sum_{t=1}^T B_t}$. Let $A_0=\{u\in \mathbb{R}^d : \|u\|=1\}$, and we can verify  that
\begin{eqnarray*}
\ev^*(|u^\top Y_T|^2)=u^\top \left(\frac{1}{2\sum_{t=1}^T B_t}\sum_{t=1}^T v_tv_t^\top\right)u.
\end{eqnarray*}
Lemma \ref{lemma:moment:inequality:phi:mixing} implies that
\begin{eqnarray*}
\ev(\|v_tv_t^\top-\ev(v_tv_t^\top)\|^2)\leq \ev(\|v_t\|^4)\leq CB_t^{2},
\end{eqnarray*}
which, by Assumption \ref{Assumption:A2}, further implies that
\begin{eqnarray*}
\sum_{t=1}^\infty \frac{\ev(\|v_tv_t^\top-\ev(v_tv_t^\top)\|^2)}{(\sum_{j=1}^t B_j)^2}\leq C\sum_{t=1}^\infty \frac{ B_t^{2}}{(\sum_{j=1}^t B_j)^2}\asymp \sum_{t=1}^\infty  \frac{t^{2\beta}}{t^{2\beta+2}}\asymp \sum_{t=1}^\infty t^{-2} <\infty.
\end{eqnarray*}
The preceding display and Assumption \ref{Assumption:A2} suggest that the conditions of Corollary 2 in \cite{kuczmaszewska2011strong} are satisfied. Hence, we show that
\begin{eqnarray*}
\frac{1}{\sum_{t=1}^T B_t}\sum_{t=1}^T \{v_tv_t^\top-\ev(v_tv_t^\top)\}\cae 0.
\end{eqnarray*}
Moreover, Lemma \ref{lemma:variance:asymptotic} implies that
\begin{eqnarray}
V_T:=\frac{1}{2\sum_{t=1}^T B_t}\sum_{t=1}^T \ev(v_tv_t^\top)\to r(0)+2\sum_{k=1}^\infty r(k):=V. \label{eq:theorem:bootstrap:asymptotic:eq:1}
\end{eqnarray}
Hence, the preceding two equations lead to
\begin{eqnarray*}
u^\top V_T u=\ev^*(|u^\top Y_T|^2)\cae u^\top V u, \textrm{ uniformly for all } u\in A_0.
\end{eqnarray*}
Similarly, for any $\epsilon>0$, we can show that
\begin{eqnarray*}
g_T(u)&:=&\frac{1}{ 2 u^\top V_T u \sum_{j=1}^T B_j } \sum_{t=1}^T\ev^*\left\{|(U_t-1)u^\top v_t|^2I\left(|(U_t-1)u^\top v_t|>\epsilon \sqrt{2 u^\top V_T u\sum_{j=t}^T B_j}\right)\right\}\\
&\leq& \frac{\sum_{t=1}^T\ev^*(|(U_t-1)u^\top v_t|^4)}{4\epsilon^2 (u^\top V_T u)^{2}(\sum_{j=1}^T B_j)^{2}}\leq \frac{C\sum_{t=1}^T\| v_t\|^4}{\epsilon^2 \lambda^2_{\min}(V_T) (\sum_{j=1}^T B_j)^{2}}.
\end{eqnarray*}
Noting that $\lim_{T\to \infty}\pr(\lambda_{\min}(V_T) \geq \lambda_{\min}(V)/2)=1$ by (\ref{eq:theorem:bootstrap:asymptotic:eq:1}) and Weyl's inequality,  it holds for any $\delta>0$ that
\begin{eqnarray*}
\pr\left(g_T(u)>\delta, \textrm{ for all }  u\in A_0\right)&\leq& \frac{4C\sum_{t=1}^T\ev(\| v_t\|^4)}{\delta \epsilon^2 \lambda^2_{\min}(V) (\sum_{j=1}^T B_j)^{2}}+\pr(\lambda_{\min}(V_T) < \lambda_{\min}(V)/2)\\
&\leq& \frac{4C\sum_{t=1}^TB_t^{2}}{\delta \epsilon^2 \lambda^2_{\min}(V) (\sum_{j=1}^T B_j)^{2}}+\pr(\lambda_{\min}(V_T) < \lambda_{\min}(V)/2)\\
&\asymp& \frac{T^{2\beta+1}}{T^{2(\beta+1)}}+\pr(\lambda_{\min}(V_T) < \lambda_{\min}(V)/2)\to 0.
\end{eqnarray*}
Hence, Lindeberg's central limit theorem leads to $Y_T|\mcD_T \cid N(0, V)$ in probability.  By the definition of $Y_T$ and Theorem \ref{theorem:asymptotic:normality}, we complete the proof.
\end{proof}

\subsection{Proof of Proposition \ref{prop:failure:bootstrap:sgd}}\label{section:proof:fail:sgd}

\begin{proof}[\textbf{Proof of Proposition \ref{prop:failure:bootstrap:sgd}}]
Let $U_t=1$ or $U_t=V_t$, and we define 
\begin{eqnarray*}
{\theta}_t={\theta}_{t-1}+\gamma_t U_t (Y_t-{\theta}_{t-1}).
\end{eqnarray*}
We can see $\theta_t=\widehat{\theta}_t$ if $U_t=1$ and $\theta_t=\widehat{\theta}_t^*$ if $U_t=V_t$.
By induction, we have
\begin{eqnarray}
\theta_t&=&\theta_{t-1}+\gamma_t U_t (Y_t-\theta_{t-1})\label{eq:fail:bootstrap:eq:-1}\\
&=&(1-\gamma_t U_t)\theta_{t-1}+\gamma_t U_t Y_t\nonumber\\
&=&\prod_{k=1}^t (1-\gamma_k U_k)\theta_0+\sum_{i=1}^{t-1} \prod_{k=i+1}^t (1-\gamma_k U_k)\gamma_i U_i Y_i\nonumber\\
&:=& \prod_{k=1}^t (1-\gamma_k U_k)\theta_0+\sum_{i=1}^{t-1} d_i^t U_i Y_i, \label{eq:fail:bootstrap:eq:0}
\end{eqnarray}
where $U_t=1$ or $U_t=V_t$. For simplicity, we may assume $\theta^*=0$. If not, we can subtract $\theta^*$ on both sides of (\ref{eq:fail:bootstrap:eq:-1}) and apply the arguments to $\theta_t-\theta^*$. Let $v_2=\ev(U^2)$. We will divide the proof into four steps.

 \textbf{Step 1:} By the fact that $\ev(U_t)=1$ and the properties of $\gamma_t$, we have 
\begin{eqnarray}
\ev[(1-\gamma_kU_k)^2]=1-2\gamma_k+v_2\gamma_k^2\leq 1-\mu \gamma_k,\label{eq:fail:bootstrap:eq:1}
\end{eqnarray}
for some constant $\mu>0$. Since $\gamma_t=t^{-\rho}$, it follows that
\begin{eqnarray}
\sum_{k=i}^t \gamma_k \geq \int_i^{t} x^{-\rho} dx \geq (1-\rho)^{-1}  [t^{1-\rho}-i^{1-\rho}].\label{eq:fail:bootstrap:eq:1.5}
\end{eqnarray}
Therefore, for all $m=1,2,\ldots, t-1$, it follows that
\begin{eqnarray*}
\sum_{i=1}^{t-1} \gamma_i^2 \prod_{k=i+1}^t (1-\mu\gamma_k )&=&\sum_{i=1}^{m} \gamma_i^2 \prod_{k=i+1}^t (1-\mu\gamma_k )+\sum_{i=m+1}^{t-1} \gamma_i^2 \prod_{k=i+1}^t (1-\mu\gamma_k )\\
&\leq &  \prod_{k=m+1}^t (1-\mu\gamma_k )\sum_{i=1}^{m} \gamma_i^2+  \frac{\gamma_m}{\mu}\sum_{i=m+1}^{t-1} \left[ \prod_{k=i+1}^t (1-\mu\gamma_k )-\prod_{k=i}^t (1-\mu\gamma_k )\right]\\
&\leq& \exp\left(-\mu \sum_{k=m+1}^t \gamma_k\right) \sum_{i=1}^\infty \gamma_i^2+\frac{\gamma_m}{\mu}\left[ 1-\prod_{k=m+1}^t (1-\mu\gamma_k )\right]\\
&\leq& \exp\left( -\frac{\mu}{1-\rho}[t^{1-\rho}-(m+1)^{1-\rho}]\right) +\frac{\gamma_m}{\mu}.
\end{eqnarray*}
Taking $m\asymp t/2$, since $\gamma_{m}\asymp \gamma_t\asymp t^{-\rho}$ and $t^{1-\rho}-(m+1)^{1-\rho}\gtrsim t^{1-\rho}$,  the preceding display leads to
\begin{eqnarray}
\sum_{i=1}^{t=1} \gamma_i^2 \prod_{k=i+1}^t (1-\mu\gamma_k) \leq C\gamma_t, \label{eq:fail:bootstrap:eq:2}
\end{eqnarray}
where $C$ is a constant free of $t$. 

 \textbf{Step 2:}
By (\ref{eq:fail:bootstrap:eq:1}), we can show that
\begin{eqnarray*}
\ev[(d_i^t)^2]=v_2\gamma_i^2 \prod_{k=i+1}^t (1-2\gamma_k+v_2\gamma_k^2 )\leq v_2  \gamma_i^2 \prod_{k=i+1}^t (1-\mu \gamma_k).
\end{eqnarray*}
Moreover, notice that
\begin{eqnarray*}
d_i^t d_{i+1}^t&=&\gamma_i\gamma_{i+1}U_iU_{i+1} \prod_{k=i+1}^t (1-\gamma_k U_k) \prod_{k=i+2}^t (1-\gamma_k U_k)\\
&=& \gamma_i\gamma_{i+1}U_iU_{i+1}  (1-\gamma_{i+1} U_{i+1}) \prod_{k=i+2}^t (1-\gamma_k U_k)^2,
\end{eqnarray*}
taking expectation implies that
\begin{eqnarray*}
\ev(d_i^t d_{i+1}^t)=\gamma_i\gamma_{i+1}(1-v_2\gamma_{i+1}) \prod_{k=i+2}^t (1-2\gamma_k+v_2\gamma_k^2) &\leq& \gamma_i \gamma_{i+1}   \prod_{k=i+2}^t (1-\mu \gamma_k)\\
&\leq& C \gamma_{i+1}^2   \prod_{k=i+2}^t (1-\mu \gamma_k),
\end{eqnarray*}
where we use the fact that $\gamma_t\leq C\gamma_{t+1}$ for all $t\geq 1$ and some $C>0$. Since $\ev(Y_t Y_s)=0$ when $|t-s|\geq 2$, it implies that
\begin{eqnarray*}
\ev\left[\left(\sum_{i=1}^{t-1} d_i^t Y_i\right)^2\right]&=&\sum_{i=1}^{t-1} \ev(Y_i^2)\ev[(d_i^t)^2]+2\sum_{i=1}^{t-2}\ev(Y_iY_{i+1})\ev(d_i^t d_{i+1}^t)\\
&\leq& 2\ev(Y^2)\left( \sum_{i=1}^{t-1} \gamma_i^2    \prod_{k=i+1}^t (1-\mu \gamma_k)+C\sum_{i=1}^{t-2} \gamma_{i+1}^2\prod_{k=i+2}^t (1-\mu \gamma_k)\right)\\
&\leq& 2(C+1)\ev(Y^2)\gamma_t,
\end{eqnarray*}
where (\ref{eq:fail:bootstrap:eq:2}) is used. Similarly, we can verify from (\ref{eq:fail:bootstrap:eq:1.5}) that
\begin{eqnarray*}
\ev\left(\left[\prod_{k=1}^t (1-\gamma_k U_k)\right]^2\right)=\prod_{k=1}^t \ev[(1-\gamma_kU_k)^2]&=&\prod_{k=1}^t (1-2\gamma_k+v_2\gamma_k^2 )\\
&\leq& \prod_{k=1}^t (1-\mu \gamma_k)\\
&\leq&\exp\left(-\mu\sum_{k=1}^t \gamma_k\right)\\
&\leq& \exp\left(-\frac{\mu}{1-\rho}(t^{1-\rho}-1)\right).
\end{eqnarray*}
Combining the above two inequality and  (\ref{eq:fail:bootstrap:eq:0}), we show that
\begin{eqnarray}
\ev(\theta_t^2)\leq C\gamma_t \label{eq:fail:bootstrap:eq:3}
\end{eqnarray}
for some $C>0$ and all $t\geq 1$.

 \textbf{Step 3:}  Using (\ref{eq:fail:bootstrap:eq:-1}), we have
\begin{eqnarray*}
\frac{\theta_t-\theta_{t-1}}{\gamma_t}=U_tY_t-U_t\theta_{t-1}=U_tY_t-(U_t-1)\theta_{t-1}-\theta_{t-1}.
\end{eqnarray*}
Hence, taking average leads to
\begin{eqnarray}
\frac{1}{T}\sum_{t=1}^T \theta_{t-1}=\frac{1}{T}\sum_{t=1}^T U_tY_t-\frac{1}{T}\sum_{t=1}^T (U_t-1)\theta_{t-1}-\frac{1}{T}\sum_{t=1}^T\frac{\theta_t-\theta_{t-1}}{\gamma_t}:=S_1-S_2-S_3. \label{eq:fail:bootstrap:eq:4}
\end{eqnarray}
By direction examination, it follows that $\ev(S_2^2|\mcD_T)\leq v_2\sum_{t=1}^T \theta_{t-1}^2/T^2$, which, by (\ref{eq:fail:bootstrap:eq:3}), further implies that
\begin{eqnarray*}
\ev(S^2_2)\lesssim T^{-2}\sum_{t=1}^T \gamma_t \asymp T^{-2}\sum_{t=1}^T t^{-\rho} \asymp T^{-1-\rho}=o(T).
\end{eqnarray*}
Abel summation implies that
\begin{eqnarray*}
\sum_{t=1}^T\frac{\theta_t-\theta_{t-1}}{\gamma_t}=(\gamma_T^{-1}\theta_T-\gamma_1^{-1}\theta_0)-\sum_{t=2}^T \theta_{t-1}(\gamma_t^{-1}-\gamma_{t-1}^{-1}).
\end{eqnarray*}
Since $\ev(\theta_T^2)\leq C\gamma_T$ by (\ref{eq:fail:bootstrap:eq:3}), it follows that $\gamma_T^{-1}\theta_T=O_P(\gamma_T^{-1/2})$. Similarly, it follows that
\begin{eqnarray*}
\ev\left(\sum_{t=2}^T |\gamma_t^{-1}-\gamma_{t-1}^{-1}| |\theta_{t-1}| \right)&\leq& \sum_{t=2}^T |\gamma_t^{-1}-\gamma_{t-1}^{-1}|  \ev^{1/2}(\theta_{t-1}^2)\\
&\lesssim&  \sum_{t=1}^T t^{-1+\rho} \gamma_t^{1/2}\lesssim T^{\rho/2}=o(T^{1/2}),
\end{eqnarray*}
where we use the fact that $\gamma_t^{-1}-\gamma_t^{-1}\asymp t^\rho-(t-1)^\rho \asymp t^{-1+\rho}$.  Hence, we prove that $\ev(|S_3|)=o(T^{-1/2})$. Combining the above bounds with (\ref{eq:fail:bootstrap:eq:4}), we show that
\begin{eqnarray}
\frac{1}{T}\sum_{t=1}^T \theta_t =\frac{1}{T}\sum_{t=1}^TU_t Y_t+o_P(T^{1/2}).  \label{eq:fail:bootstrap:eq:5}
\end{eqnarray}

 \textbf{Step 4:} Using (\ref{eq:fail:bootstrap:eq:5}) and the definition of $\widehat{\theta}_t$ and $\widehat{\theta}_t^*$, we can prove the desired statements using similar arguments in the proof of Theorems \ref{theorem:asymptotic:normality} and \ref{theorem:bootstrap:asymptotic}.

\end{proof}

\subsection{Additional Theoretical Results}\label{sec:additional:theoretical:result}
In this section, we provide additional theoretical analysis for the proposed estimator, which is of independent interest.
\begin{lemma}\label{lemma:some:rate:consistency:higher:moment} Suppose that Assumptions \ref{Assumption:A0}-\ref{Assumption:A2} and \ref{Assumption:A4} hold. Then the following statements hold.
\begin{enumerate}[label=(\roman*)]
\item \label{lemma:some:rate:consistency:higher:moment:item:1}  $\ev(\|e_t\|^k)\leq C\phi^{k-2k/p}(B_{t-1})$ for all $2\leq k<p$ and $t\geq1$.
\item  \label{lemma:some:rate:consistency:higher:moment:item:2} $\ev(\|\zeta_t\|^k)\leq C$ for all $t\geq 1$.
\end{enumerate}
 Here $C>0$ is a constant free of $t$.
\end{lemma}
\begin{proof}
\begin{enumerate}[label=(\roman*)]
\item Following the proof of Lemma \ref{lemma:some:rate:consistency}, we have
\begin{eqnarray*}
\left\|e_t\right\|^k&\leq& 3^k m^k\phi^k(B_{t-1})+\frac{3^k \ev_{t-1}^k\left(\left\|\widehat{H}_t({W}_t, \theta_{t-1})\right\|^{p/k}\right)}{m^{p-k}}+\frac{3^k \ev_{t-1}^k\left(\left\|\widehat{H}_t(\widetilde{W}_t, \theta_{t-1})\right\|^{p/k}\right)}{m^{p-k}}\\
&\leq&3^k m^k\phi^k(B_{t-1})+\frac{3^k \ev_{t-1}\left(\left\|\widehat{H}_t({W}_t, \theta_{t-1})\right\|^{p}\right)}{m^{p-k}}+\frac{3^k \ev_{t-1}\left(\left\|\widehat{H}_t(\widetilde{W}_t, \theta_{t-1})\right\|^{p}\right)}{m^{p-k}}.
\end{eqnarray*}
Taking  $m=\phi^{-2/p}(B_{t-1})$ and using the similar arguments in the proof of Lemma \ref{lemma:some:rate:consistency}, we conclude that $\ev(\|e_t\|^k)\leq C\phi^{k-2k/p}(B_{t-1})$.
\item  Let $\widehat{\zeta}_t=U_t \widehat{H}_t(\widetilde{W}_t, \theta_{t-1})-\ev_{t-1}\left\{\widehat{H}_t(W_t, \theta_{t-1})\right\}.$ By the definition of $\zeta_t$, we can see that
\begin{eqnarray*}
|\ev_{t-1}(\|\zeta_t\|^k)-\ev_{t-1}(\|\widehat{\zeta}_t\|^k)|\leq \phi(B_{t-1})m+\frac{\ev_{t-1}(\|\zeta_t\|^p)}{m^{p/k-1}}+\frac{\ev_{t-1}(\|\widehat{\zeta}_t\|^p)}{m^{p/k-1}}
\end{eqnarray*}
Similarly to the proof of Lemma \ref{lemma:some:rate:consistency}, we can show that $\ev(\|\zeta_t\|^p)\leq C_p$ and $\ev(\|\widehat{\zeta}_t\|^p)\leq C_p$ for some constant $C_p>0$. Taking $m=1$, we show that
\begin{eqnarray*}
|\ev_{t-1}(\|\zeta_t\|^k)-\ev_{t-1}(\|\widehat{\zeta}_t\|^k)|\leq \phi(B_t)+\ev_{t-1}(\|\zeta_t\|^p)+\ev_{t-1}(\|\widehat{\zeta}_t\|^p).
\end{eqnarray*}
Moreover, direct examination leads to 
\begin{eqnarray*}
\ev_{t-1}(\|\widehat{\zeta}_t\|^k)\leq 2^{k-1}\ev_{t-1}(\|\widetilde{\zeta}_t\|^k)+2^{k-1}\|e_t\|^k,
\end{eqnarray*}
where  $\widetilde{\zeta}_t=U_t \widehat{H}_t(\widetilde{W}_t, \theta_{t-1})-\ev_{t-1}\left\{\widehat{H}_t(\widetilde{W}_t, \theta_{t-1})\right\}.$
Using Statement \ref{lemma:some:rate:consistency:higher:moment:item:1},  Lemma \ref{lemma:moment:inequality:phi:mixing} and Assumption \ref{Assumption:A2}\ref{A2:mixing:condition}, we show that
\begin{eqnarray*}
\ev(\|\widehat{\zeta}_t\|^k)\leq 2^{k-1}\ev(\|\widetilde{\zeta}_t\|^k)+2^{k-1}\ev(\|e_t\|^k)\leq C2^{k-1}B_{t}^{-1}+C\phi^{k-2k/p}(B_{t-1}).
\end{eqnarray*}
Hence, we complete the proof.
\end{enumerate}
\end{proof}

\begin{lemma}\label{lemma:additiona:lp:convergence:rate}
Suppose that Assumptions \ref{Assumption:A0}-\ref{Assumption:A2} and \ref{Assumption:A4} hold. 
\begin{enumerate}[label=(\roman*)]
\item  \label{lemma:additiona:lp:convergence:rate:item:1} In addition, if $\phi^{\frac{1}{2}-\frac{1}{p}}(B_t)\leq Ct^{-\rho}$ for some $C>0$ and all $t\geq 1$, then there is a constant $C_\delta>0$ such that $\ev(\|\Delta_t\|^{2+\delta})\leq C_\delta t^{-\frac{\rho(2+\delta)}{2}}$ for all $t\geq 1$ and $0\leq \delta<1-2/p$.
\item \label{lemma:additiona:lp:convergence:rate:item:2}  In addition, if  Assumptions \ref{Assumption:A0}-\ref{Assumption:A2} hold for all $p\geq 4$ and $\phi^{\frac{1}{2}}(B_t)\leq Ct^{-\rho}$  for some $C>0$ and all $t\geq 1$,  then there is a constant $C_\delta>0$ such that $\ev(\|\Delta_t\|^{2+\delta})\leq C_\delta t^{-\frac{\rho(2+\delta)}{2}}$ for all $t\geq 1$ and $\delta\geq 0$. 
\end{enumerate}

\end{lemma}
\begin{proof}
\begin{enumerate}[label=(\roman*)]
\item  Since projection is a contraction map, we see that
\begin{eqnarray*}
\|\Delta_t\|=\left\|\Pi \left\{\Delta_{t-1}-\gamma_t  H(\theta_{t-1})-\gamma_t  e_t-\gamma_t \zeta_t\right\}\right\|\leq \|\Delta_{t-1}-\gamma_t  H(\theta_{t-1})-\gamma_t  e_t-\gamma_t \zeta_t\|.
\end{eqnarray*}
By the proof of Proposition 3.1 in \cite{chen2021online}, it holds for any vectors $a, b$ that
\begin{eqnarray*}
\|a+b\|^{2+\delta}\leq \|a\|^{2+\delta}+(2+\delta)a^\top b\|a\|^{\delta}+C\|a\|^{\delta}\|b\|^{2}+C\|b\|^{2+\delta}.
\end{eqnarray*}
Here $C>0$ is a constant depending on $\delta$.
Let $a=\Delta_{t-1}$ and $b=-\gamma_t  H(\theta_{t-1})-\gamma_t  e_t-\gamma_t$. The preceding two displays imply that
\begin{eqnarray*}
\|\Delta_t\|^{2+\delta}&\leq&  \|\Delta_{t-1}\|^{2+\delta}+(2+\delta)\Delta_{t-1}^\top b \|\Delta_{t-1}\|^{\delta}+C\|\Delta_{t-1}\|^{\delta}\|b\|^{2}+C\|b\|^{2+\delta}\\
&=& \|\Delta_{t-1}\|^{2+\delta}-(2+\delta)\gamma_t \Delta_{t-1}^\top H(\theta_{t-1})\|\Delta_{t-1}\|^{\delta}-(2+\delta)\gamma_t \Delta_{t-1}^\top (e_t+\zeta_t)\|\Delta_{t-1}\|^{\delta}\\
&&+C\|\Delta_{t-1}\|^{\delta}\|b\|^{2}+C\|b\|^{2+\delta}\\
&\leq& \left(1-(2+\delta)K\gamma_t \right)\|\Delta_{t-1}\|^{2+\delta}+(2+\delta)\gamma_t\|e_t\|\|\Delta_{t-1}\|^{1+\delta}-(2+\delta)\gamma_t \Delta_{t-1}^\top \zeta_t\|\Delta_{t-1}\|^{\delta}\\
&&+C\|\Delta_{t-1}\|^{\delta}\|b\|^{2}+C\|b\|^{2+\delta}.
\end{eqnarray*}
Here we use Lemma \ref{lemma:some:inequalities}\ref{lemma:some:inequalities:item:1}. Moreover, it follows from Lemmas \ref{lemma:some:inequalities} and \ref{lemma:some:rate:consistency:higher:moment} that
\begin{eqnarray*}
\sup_{\theta\in \Theta}\|H(\theta)\|\leq C,\quad \ev(\|e_t\|^k)\leq C\phi^{k-2k/p}(B_{t-1}),\quad  \ev(\|\zeta_t\|^k)\leq  C,\quad \ev_{t-1}(\zeta_t)=0
\end{eqnarray*}
for some $C>0$ and for all $k<p, t\geq 1$. Taking expectation, it holds that
\begin{eqnarray*}
&&\ev\left(\|\Delta_t\|^{2+\delta}\right)\\
&\leq&  \left(1-(2+\delta)K\gamma_t \right) \ev\left(\|\Delta_{t-1}\|^{2+\delta}\right)+(2+\delta)\gamma_t \ev\left(\|e_t\|\|\Delta_{t-1}\|^{1+\delta}\right)\\
&&+C \ev\left(\|\Delta_{t-1}\|^{\delta}\|b\|^{2}\right)+C \ev\left(\|b\|^{2+\delta}\right)\\
&\leq&\left(1-(2+\delta)K\gamma_t \right) \ev\left(\|\Delta_{t-1}\|^{2+\delta}\right)+(2+\delta) \gamma_t\ev^{\frac{1-\delta}{2}}\left(\|e_t\|^{\frac{2}{1-\delta}}\right)\ev^{\frac{1+\delta}{2}}\left\{\|\Delta_{t-1}\|^{2}\right\}\\
&&+C \ev^{\frac{\delta}{2}}\left(\|\Delta_{t-1}\|^{2}\right)\ev^{\frac{2-\delta}{2}}\left(\|b\|^{\frac{4}{2-\delta}}\right)+C \ev\left\{\|b\|^{2+\delta}\right\}.
\end{eqnarray*}
Moreover, the rate condition $\phi^{\frac{1}{2}-\frac{1}{p}}(B_t)\lesssim t^{-\rho}$ and Lemma \ref{lemma:rate:evDelta2} imply that $\ev(\|\Delta_t\|^2)\leq Ct^{-\rho}$. Noting that $\delta<1-2/p$ implies that $2/(1-\delta)<p$ and $4/(2-\delta)<p$, by Holder's inequality, the preceding display leads to
\begin{eqnarray*}
\ev\left(\|\Delta_t\|^{2+\delta}\right)&\leq& (1-c\gamma_t)\ev\left(\|\Delta_{t-1}\|^{2+\delta}\right)+C
\gamma_t \phi^{1-2/p}(B_{t-1})\ev^{\frac{1+\delta}{2}}\left\{\|\Delta_{t-1}\|^{2}\right\}\\
&&+C \gamma_t^2 \ev^{\frac{\delta}{2}}\left(\|\Delta_{t-1}\|^{2}\right)+C\gamma_t^{2+\delta}\\
&\leq& (1-c\gamma_t)\ev\left(\|\Delta_{t-1}\|^{2+\delta}\right)+C t^{-2\rho-\frac{\rho(1+\delta)}{2}}+Ct^{-2\rho-\frac{\rho\delta}{2}}+Ct^{-\rho(2+\delta)}.
\end{eqnarray*}
According to Lemma \ref{lemma:my:iteration:bound}, we conclude that
\begin{eqnarray*}
\ev\left(\|\Delta_t\|^{2+\delta}\right)\lesssim    t^{-\rho-\frac{\rho(1+\delta)}{2}}+t^{-\rho-\frac{\rho\delta}{2}}+t^{-\rho-\rho\delta}\lesssim t^{-\frac{\rho(2+\delta)}{2}}.
\end{eqnarray*}
Hence, we complete the proof of Statement \ref{lemma:additiona:lp:convergence:rate:item:1}.
\item By Statement \ref{lemma:additiona:lp:convergence:rate:item:1} and the conditions given, we see that $\ev(\|\Delta_{t}\|^{2+\delta})\leq C_\delta t^{-\frac{\rho(2+\delta)}{2}}$ for all $\delta<1$. Using similar arguments as the proof in Statement  \ref{lemma:additiona:lp:convergence:rate:item:1}, we can show for $0<\delta<2-u, 0<u<1$ that
\begin{eqnarray*}
\ev\left(\|\Delta_t\|^{2+\delta}\right)&\leq& (1-c\gamma_t)\ev\left(\|\Delta_{t-1}\|^{2+\delta}\right)+C
\gamma_t \ev^{\frac{2-u-\delta}{3-u}}\left(\|e_t\|^{\frac{3-u}{2-u-\delta}}\right)\ev^{\frac{1+\delta}{3-u}}\left\{\|\Delta_{t-1}\|^{3-u}\right\}\\
&&+C \gamma_t^2 \ev^{\frac{\delta}{2}}\left(\|\Delta_{t-1}\|^{2}\right)+C\gamma_t^{2+\delta}\\
&\leq& (1-c\gamma_t)\ev\left(\|\Delta_{t-1}\|^{2+\delta}\right)+C
\gamma_t \phi^{1-2/p}(B_{t-1})\ev^{\frac{1+\delta}{3-u}}\left\{\|\Delta_{t-1}\|^{3-u}\right\}\\
&&+C \gamma_t^2 \ev^{\frac{\delta}{2}}\left(\|\Delta_{t-1}\|^{2}\right)+C\gamma_t^{2+\delta}\\
&\leq& (1-c\gamma_t)\ev\left(\|\Delta_{t-1}\|^{2+\delta}\right)+C t^{-2\rho-\frac{\rho(1+\delta)}{2}}+Ct^{-2\rho-\frac{\rho\delta}{2}}+Ct^{-\rho(2+\delta)}.
\end{eqnarray*}
According to Lemma \ref{lemma:my:iteration:bound} and noting that $0<u<1$ is arbitrary, we conclude that
\begin{eqnarray*}
\ev\left(\|\Delta_t\|^{2+\delta}\right)\lesssim    t^{-\rho-\frac{\rho(1+\delta)}{2}}+t^{-\rho-\frac{\rho\delta}{2}}+t^{-\rho-\rho\delta}\lesssim t^{-\frac{\rho(2+\delta)}{2}}
\end{eqnarray*}
for all $0<\delta<2$. By induction, we can show the desired result holds for all $\delta\geq 0$.
\end{enumerate}
\end{proof}
\begin{remark}
A potential application of Statement \ref{lemma:additiona:lp:convergence:rate:item:2} in Lemma \ref{lemma:additiona:lp:convergence:rate} is to establish non-Asymptotic confidence bounds of the SGD estimator (e.g., \citealp{chen2023recursive}).
\end{remark}

\subsection{Additional Numerical Results}\label{sec:additional:simulation}

In this section, we present additional numerical exploring the influence of the bootstrap sample size. We adopt the parameter settings outlined in Section \ref{sec:simulation}, utilizing the proposed batch size $\beta=0.33$, and examine various bootstrap sample sizes: $N=50, 200, 500, 1000$. From Figures \ref{figure:addition:cp:model1:4} to \ref{figure:addition:cp:model8}, it becomes evident that the choices of $N=200, 500, 1000$ consistently yield satisfactory performance.
\begin{figure}[htpb]
\centering
\includegraphics[width=2.5 in]{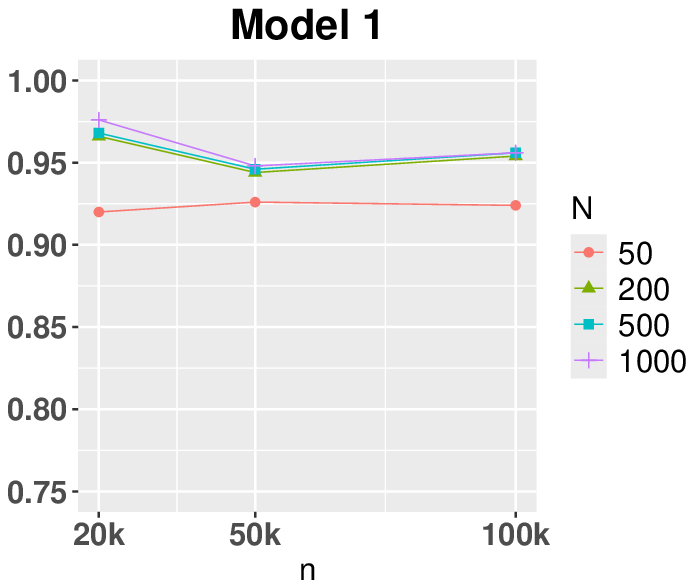}
\includegraphics[width=2.5 in]{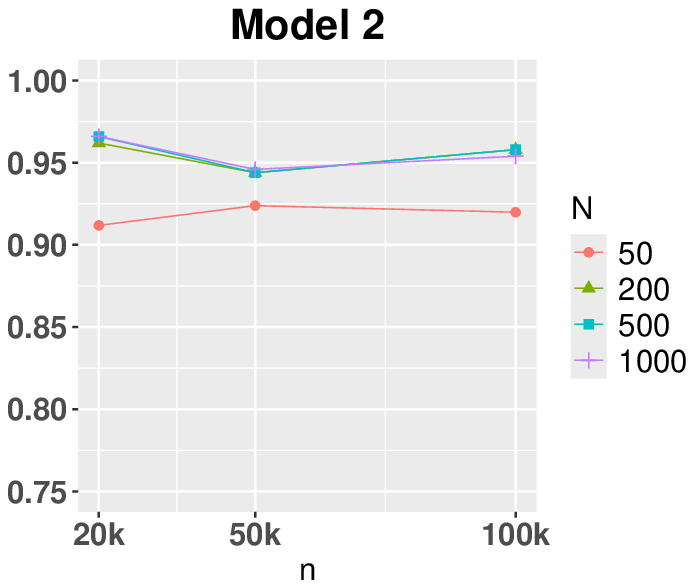}

\includegraphics[width=2.5 in]{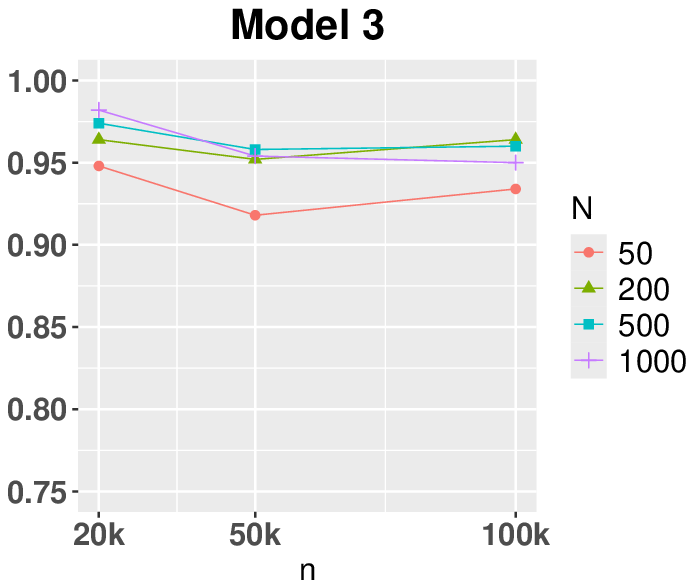}
\includegraphics[width=2.5 in]{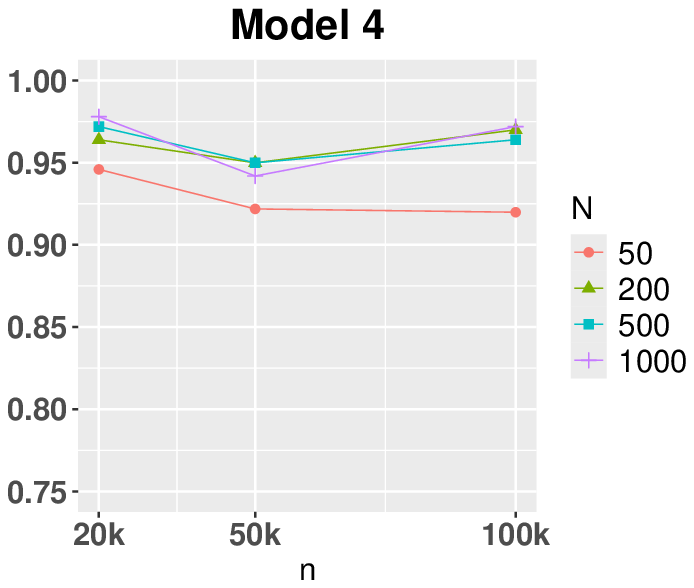}
\caption{CP for Models 1-4}
\label{figure:addition:cp:model1:4}
\end{figure}

\begin{figure}[htpb]
\centering
\includegraphics[width=1.8 in]{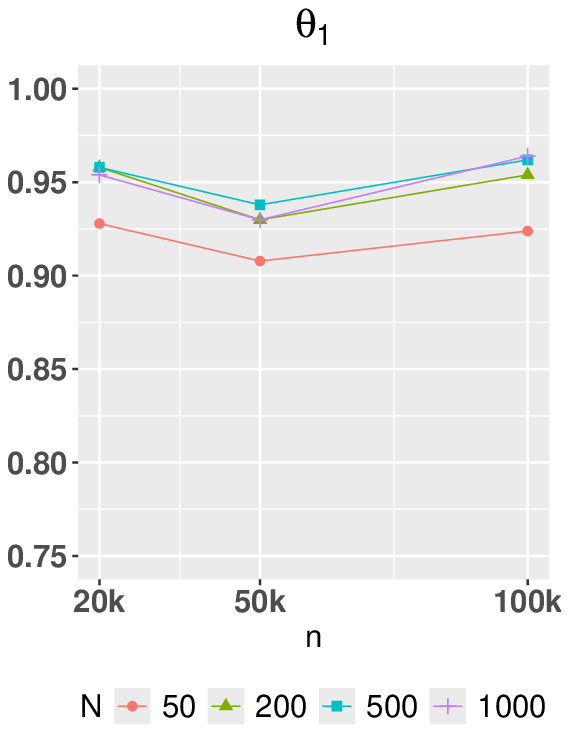}
\includegraphics[width=1.8 in]{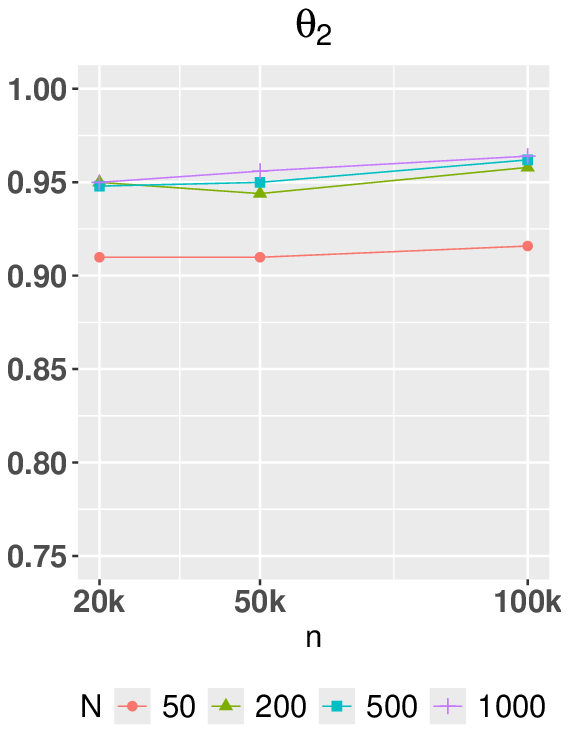}
\includegraphics[width=1.8 in]{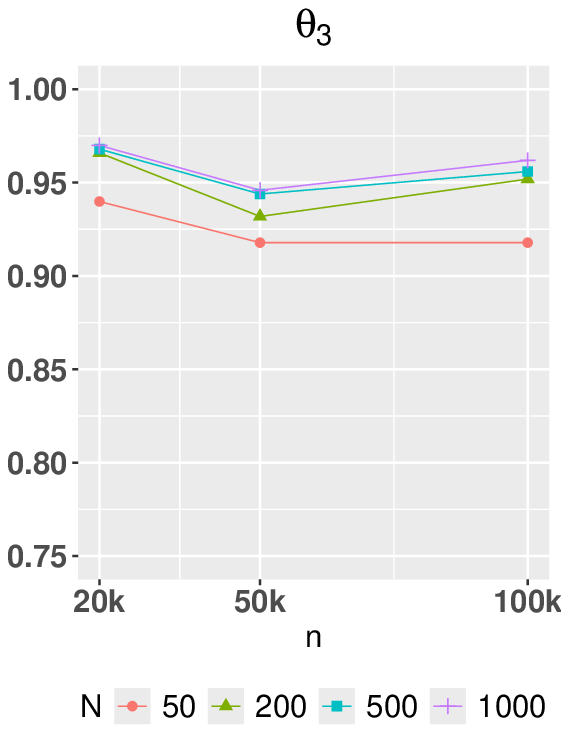}
\caption{CP for Model 5}
\label{figure:addition:cp:model5}
\end{figure}

\begin{figure}[htpb]
\centering
\includegraphics[width=1.8 in]{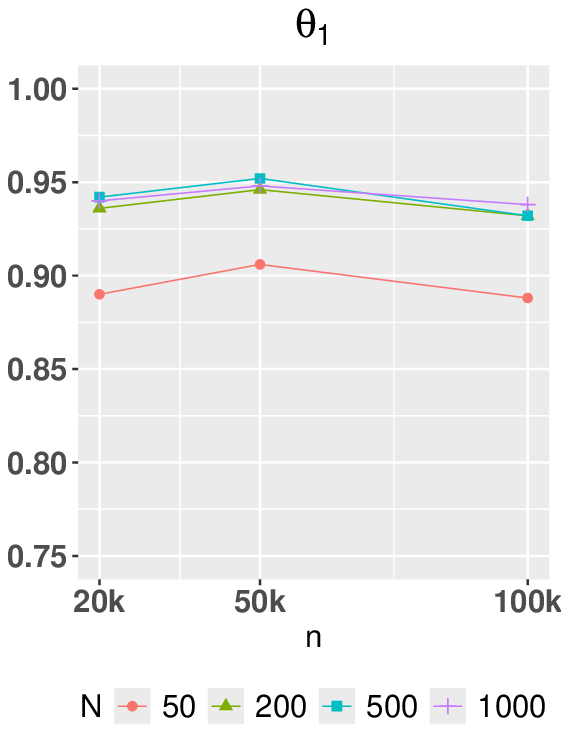}
\includegraphics[width=1.8 in]{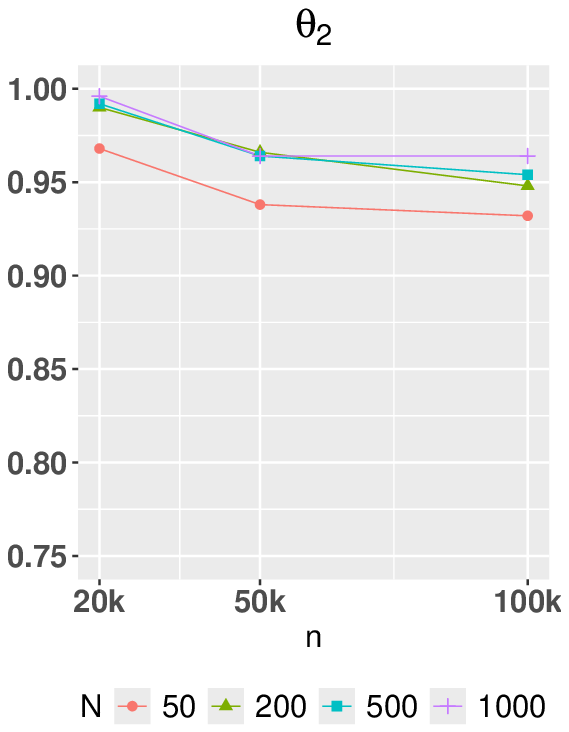}
\includegraphics[width=1.8 in]{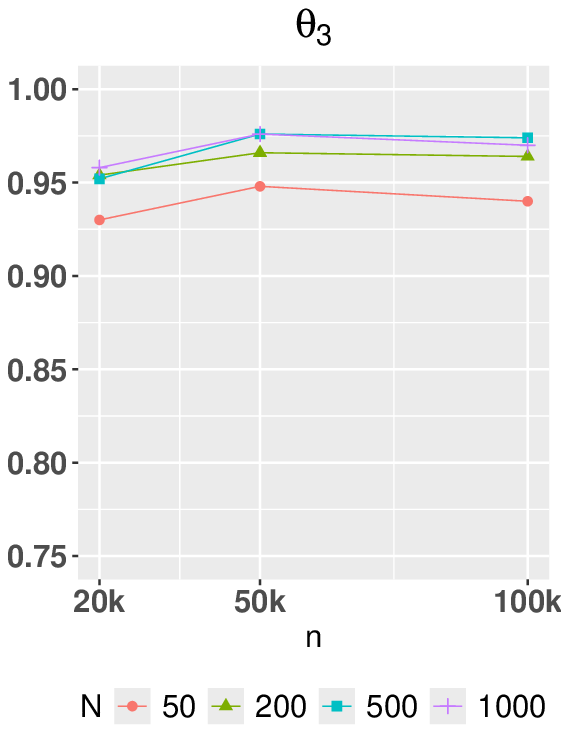}
\caption{CP for Model 6}
\label{figure:addition:cp:model6}
\end{figure}

\begin{figure}[htpb]
\centering
\includegraphics[width=1.8 in]{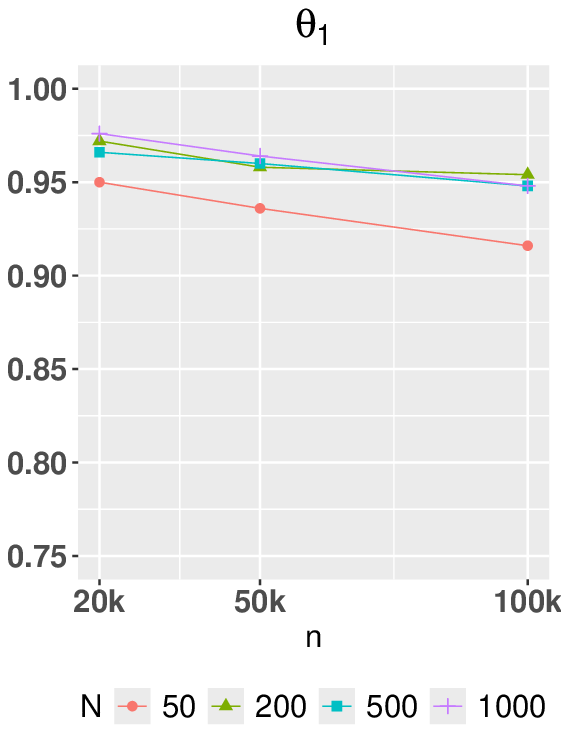}
\includegraphics[width=1.8 in]{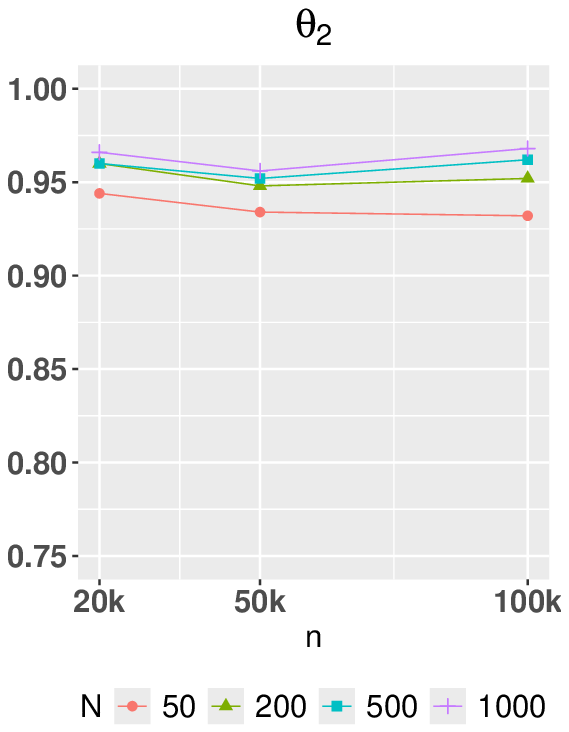}
\includegraphics[width=1.8 in]{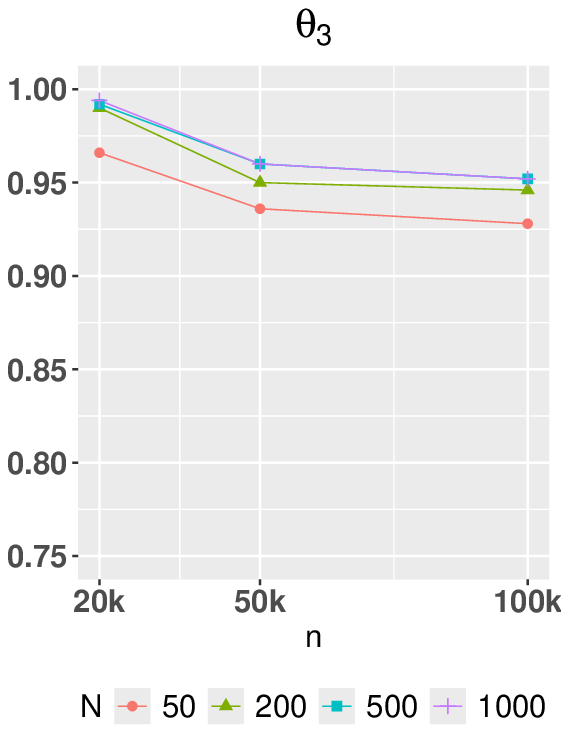}
\caption{CP for Model 7}
\label{figure:addition:cp:model7}
\end{figure}

\begin{figure}[htpb]
\centering
\includegraphics[width=1.8 in]{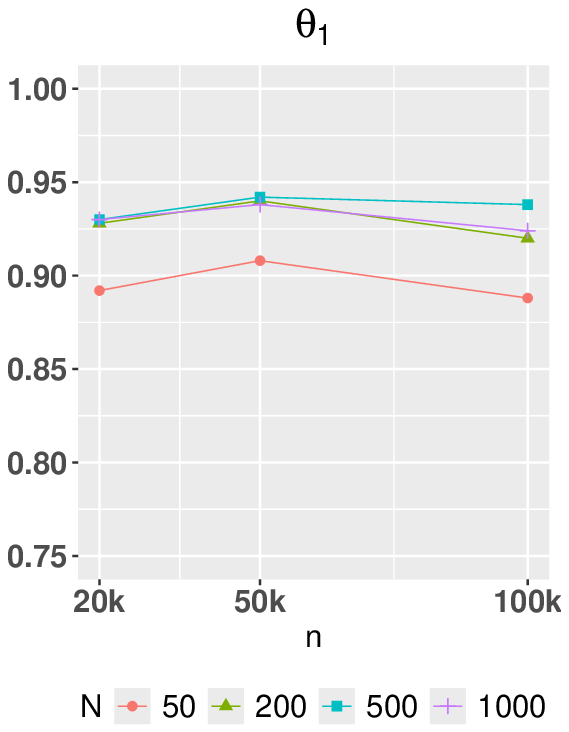}
\includegraphics[width=1.8 in]{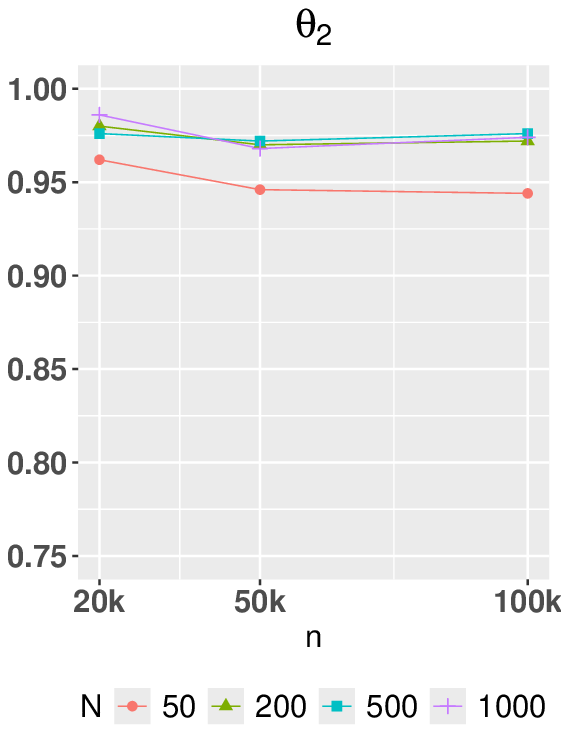}
\includegraphics[width=1.8 in]{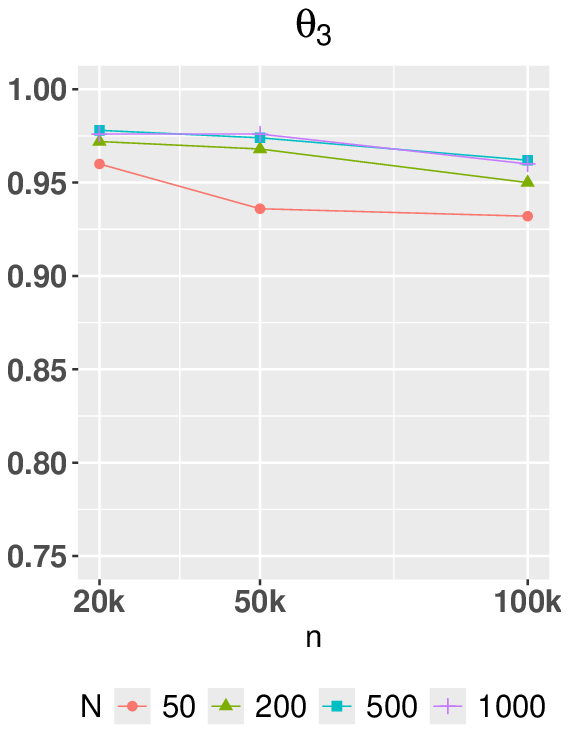}
\caption{CP for Model 8}
\label{figure:addition:cp:model8}
\end{figure}

\bibliographystyle{Chicago}
\bibliography{ref}{}
\clearpage
\end{document}